\documentclass[11pt]{article}

\usepackage[T1]{fontenc}
\usepackage{microtype}
\usepackage{amsmath,amsthm}
\allowdisplaybreaks[1]
\usepackage{newpxtext,newpxmath}
\usepackage[margin=15pt,font=small,labelfont={bf}]{caption}
\usepackage[margin=1in]{geometry}
\usepackage[utf8]{inputenc}
\usepackage{array}

\usepackage{csquotes}
\usepackage[english]{babel}
\usepackage[backend=biber,style=alphabetic,maxalphanames=3,minalphanames=3,maxbibnames=99,minbibnames=99]{biblatex}
\usepackage{graphicx}
\usepackage{algorithm}
\usepackage[noend]{algpseudocode}

\AtBeginDocument{
  }
\usepackage{xcolor}
\usepackage[colorlinks=true,
    citecolor=darkgreen,
    linkcolor=darkred
]{hyperref}

\definecolor{darkgreen}{rgb}{0,0.5,0.5}
\definecolor{darkred}{rgb}{0.5,0,0}
\definecolor{sdpgreen}{rgb}{0,0.5,0}

\newcommand\drop[1]{}
\newcommand*{\newaliastheorem}[3]{%
  \newtheorem{#1}[#2]{#3}%
  \expandafter\providecommand\csname #1autorefname\endcsname{#3}%
}
\newtheorem{theorem}{Theorem}
\newaliastheorem{proposition}{theorem}{Proposition}
\newaliastheorem{lemma}{theorem}{Lemma}
\newaliastheorem{corollary}{theorem}{Corollary}
\newaliastheorem{definition}{theorem}{Definition}
\newaliastheorem{conjecture}{theorem}{Conjecture}
\newaliastheorem{example}{theorem}{Example}
\newaliastheorem{remark}{theorem}{Remark}
\newaliastheorem{fact}{theorem}{Fact}
\newaliastheorem{claim}{theorem}{Claim}

\usepackage{thm-restate}

\makeatletter
\IfFormatAtLeastTF{2026-06-01}
  {}
  {\usepackage{aliascnt}}

\IfFormatAtLeastTF{2026-06-01}
  {\newcounteralias{appsubsection}{subsection}}
  {\newaliascnt{appsubsection}{subsection}}

\let\l@appsubsection\l@subsection
\let\appsubsectionmark\@gobble
\def\toclevel@appsubsection{2}
\makeatother
\addto\extrasenglish{

    \def\equationautorefname~#1\null{Equation~(#1)\null}
}
\newcommand{\R}{\mathbb{R}}
\newcommand{\E}{\mathbb{E}}
\newcommand{\Prb}{\mathbb{P}}
\newcommand{\cA}{\mathcal{A}}
\newcommand{\cB}{\mathcal{B}}
\newcommand{\cD}{\mathcal{D}}
\newcommand{\cE}{\mathcal{E}}
\newcommand{\cG}{\mathcal{G}}
\newcommand{\cH}{\mathcal{H}}
\newcommand{\cK}{\mathcal{K}}
\newcommand{\cL}{\mathcal{L}}
\newcommand{\cO}{\mathcal{O}}
\newcommand{\cP}{\mathcal{P}}
\newcommand{\cS}{\mathcal{S}}
\newcommand{\cU}{\mathcal{U}}
\newcommand{\cW}{\mathcal{W}}
\newcommand{\cX}{\mathcal{X}}
\newcommand{\ip}[2]{\langle #1,#2\rangle}
\newcommand{\norm}[1]{\lVert#1\rVert}
\newcommand{\bPhi}{\overline{\Phi}}
\newcommand{\gmeasure}{\gamma}
\newcommand{\polylog}{\operatorname{polylog}}
\newcommand{\Span}{\operatorname{span}}

\title{An $O(n^{0.1537})$-Coloring Algorithm for $3$-Colorable Graphs}
\author{
    Kaisei Deguchi\thanks{Supported by Kakenhi JP26K21777, JP25K24465, and JST ASPIRE JPMJAP2302.} \\
    The University of Tokyo, Japan \\
    deguchi-kaisei919@g.ecc.u-tokyo.ac.jp
    \and
    Ken-ichi Kawarabayashi \thanks{Supported by Kakenhi JP26K21777, JP25K24465, and JST ASPIRE JPMJAP2302.} \\
    National Institute of Informatics \& \\ The University of Tokyo, Japan \\
    k\_keniti@nii.ac.jp
    \and
    Hirotaka Yoneda\thanks{Supported by Kakenhi JP26K21777, JP25K24465, JST ASPIRE JPMJAP2302, and JST ACT-X JPMJAX25CT.} \\
    The University of Tokyo, Japan \\
    squar37@gmail.com
    \and
    Masataka Yoneda\thanks{Supported by Kakenhi JP26K21777, JP25K24465, and JST ASPIRE JPMJAP2302.} \\
    The University of Tokyo, Japan \\
    e869120@gmail.com
}
\date{}
\begin{document}
\maketitle

\pagenumbering{gobble}

\begin{abstract}
    We present a polynomial-time randomized algorithm that colors every $n$-vertex $3$-colorable graph with $O(n^{0.1537})$ colors. This improves upon the previous bound of $O(n^{0.1954})$ colors by Bansal, Huang, and Lee (SODA '27).

    To achieve this result, we introduce a new algorithmic framework for finding a large independent set from a solution to the vector coloring SDP. Combined with the (combinatorial) dense-progress algorithm of Kawarabayashi, Thorup, and Yoneda (STOC '24), this framework yields an $n^{1/6 + o(1)}$-coloring algorithm. Furthermore, we strengthen the SDP-based algorithm specifically for 3-colorable graphs to achieve $O(n^{0.1537})$ colors.
\end{abstract}

\newpage
{
    \setcounter{tocdepth}{2}
    \addtocontents{toc}{\protect\linespread{1.08}\protect\selectfont} 
    \tableofcontents
    \clearpage
}
\pagenumbering{arabic}
\setcounter{page}{1}


\section{Introduction}\label{sec:1}

The question of \emph{graph coloring}, which colors each vertex of a given graph $G = (V, E)$ with the smallest number of colors so that adjacent vertices always have different colors, is a central topic in both discrete math and theoretical computer science. 
This is not only one of the most important NP-hard problems, but even deciding whether three colors suffice is among Karp's 21 classic NP-complete problems~\cite{GJS74, Kar72}.

Indeed, the fastest known algorithm to solve graph coloring takes $O^{*}(2^{n})$ time~\cite{BHK09}, where $n$ is the number of vertices.
For determining whether the graph is $k$-colorable, the fastest known algorithm takes $O^{*}(1.3289^{n})$ time for $k = 3$~\cite{BE05}, $O^{*}(1.7159^{n})$ time for $k = 4$~\cite{WGJ+24}, and $O^{*}((2-\varepsilon_{k})^{n})$ time for $k = 5, 6$ for some very small $\varepsilon_{k} > 0$~\cite{Zam21}.
Very recently, two independent preprints that claim $O^{*}((2-\varepsilon_{k})^{n})$ time for all $k \in \mathbb{N}$ appeared~\cite{Zam26, Pra26}.
Polynomial-time algorithms are known only for limited graph classes, including perfect graphs~\cite{GLS84} and bounded-treewidth graphs~\cite{CFK+15}.
\subsection{Approximation Algorithms}\label{sec:1.1}

Approximation algorithms are well studied for NP-hard problems in general.
However, graph coloring is NP-hard even to approximate within an $n^{1-\varepsilon}$ factor for all $\varepsilon > 0$~\cite{FK98, Has99, Zuc06}.
Note that the best known approximation factor is $O(n (\log{\log{n}})^{2} / (\log{n})^{3})$~\cite{BR90, Hal93}.

Given very strong negative results for the approximation, the following question became the main line of research: \emph{``Given a $k$-colorable graph $G = (V, E)$, find a coloring in polynomial time, using a sufficiently small number of colors.''}
In particular, the case $k = 3$ is important for being the smallest nontrivial case, and improving this case typically yields a better algorithm for all subsequent $k \ge 4$~\cite{Blu94, HNZ02}.

This topic also connects deeply to a central problem in complexity theory, the \emph{Unique Games Conjecture}, and its variant, the \emph{$d$-to-$1$ Conjecture}, both introduced by Khot~\cite{Kho02}.
If the $d$-to-$1$ Conjecture is true for some constant $d$, then it is NP-hard to color $3$-colorable graphs with any constant number of colors in polynomial time~\cite{DMR09, GS20}.
Very recently, a preprint that claims to resolve the $4$-to-$1$ Conjecture appeared~\cite{FMW26}, which, if true, would resolve the impossibility for any constant number of colors.
Meanwhile, unconditional NP-hardness was already known for coloring $k$-colorable graphs with $2k-1$ colors~\cite{KLS00, GK04, BKO19}.
\subsection{Polynomial-Time Coloring of \texorpdfstring{$3$}{3}-Colorable Graphs}\label{sec:1.2}

Wigderson~\cite {Wig82} initiated research on coloring $3$-colorable graphs in 1982, presenting an $O(\sqrt{n})$-coloring algorithm. Since then, researchers have made a long line of improvements over 40 years, approaching the problem from two sides: the \textbf{large-degree} case and the \textbf{small-degree} case. We present the sequence of improvements in \autoref{tab:intro}.

\begin{table}[ht]
    \centering
    \small
    \caption{The line of research for coloring $3$-colorable graphs in polynomial time.
    Although not listed below, some very recent preprints claim the following results: $\widetilde{O}(n^{0.1751})$ colors on September 27, 2026~\cite{NT26}, $\widetilde{O}(n^{0.1780})$ colors on September 30, 2026~\cite{MM26}, $\widetilde{O}(n^{0.1740})$ colors on October 1, 2026~\cite{Ana26}, and $\widetilde{O}(n^{0.1653})$ colors on October 2, 2026~\cite{Son26}.}
    \begin{tabular}{|c|c|c|p{10.5cm}|}
        \hline
        Year & Paper & Colors & Result Description (Informal) \\
        \hline
        1982 & \cite{Wig82} & $O(n^{1/2})$ & \textbf{Combinatorial.} The first nontrivial algorithm for coloring $3$-colorable graphs. \\
        1989 & \cite{Blu89} & $\widetilde{O}(n^{2/5})$ & \textbf{Combinatorial.} Utilized the second neighborhood with a $(2-\frac{\log{\log{n}}}{\log{n}})$-approximation algorithm for vertex cover~\cite{BE85}. \\
        1990 & \cite{Blu90} & $\widetilde{O}(n^{3/8})$ & \textbf{Combinatorial.} Improved the case where the structure between the first and second neighborhoods is dense. \\
        1994 & \cite{KMS94} & $\widetilde{O}(n^{1/4})$ & \textbf{SDP.} Applied semidefinite programming (SDP) to the problem and achieved $\widetilde{O}(\Delta^{1/3})$ colors for maximum degree $\Delta$. \\
        1997 & \cite{BK97} & $\widetilde{O}(n^{3/14})$ & Devised a framework that combines small-degree and large-degree progress. \\
        2006 & \cite{ACC06} & $\widetilde{O}(n^{0.2111})$ & \textbf{SDP.} Devised a better rounding scheme that considers the second neighborhood. \\
        2007 & \cite{Chl07} & $\widetilde{O}(n^{0.2072})$ & \textbf{SDP.} Considered a higher-degree SDP using Lasserre's sum-of-squares hierarchy~\cite{Las01}. \\
        2012 & \cite{KT12} & $\widetilde{O}(n^{0.2049})$ & \textbf{Combinatorial.} Introduced the idea of ``sparse cut''; the best known purely combinatorial solution with $\widetilde{O}(n^{4/11})$ colors. \\
        2014 & \cite{KT14} & $\widetilde{O}(n^{0.2000})$ & \textbf{Combinatorial.} Introduced a new recursive subroutine that makes progress with even smaller dense substructures. \\
        2024 & \cite{KTY24} & $\widetilde{O}(n^{0.1975})$ & \textbf{Combinatorial.} Introduced the idea of ``side cut'', and finally completed a framework toward $(n/\Delta)^{1/2} n^{o(1)}$-coloring for degree $\Delta$. \\
        2026 & \cite{BHL27} & $\widetilde{O}(n^{0.1954})$ & \textbf{SDP.} Considered the third neighborhood structure and combined it with a new ``vector $\frac{5}{2}$-coloring'' idea. \\
        2026 & \textbf{Ours} & $O(n^{0.1537})$ & \textbf{SDP.} Introduced iterative local-maximum selection under Gaussian projections with Gaussian cover lifting. \\
        \hline
    \end{tabular}
    \label{tab:intro}
\end{table}

\paragraph{Two directions of the attack.}
Generally speaking, if we want to create an algorithm to color $3$-colorable graphs with $\widetilde{O}(n^{\alpha})$ colors (for $\alpha \in (0, 1)$), we consider an algorithm to make ``progress'' toward $\widetilde{O}(n^{\alpha})$-coloring. Here, the \emph{progress} includes obtaining an independent set of size $\widetilde{\Omega}(n^{1-\alpha})$, as repeating the routine of ``coloring the independent set with a common new color and removing these vertices from the graph'' until the graph becomes empty gives an $\widetilde{O}(n^{\alpha})$-coloring. Note that, while Blum~\cite{Blu94} also considers other types of progress,\footnote{Blum~\cite{Blu94} gives three types of progress: (1) finding an independent set of size $\widetilde{\Omega}(n^{1-\alpha})$, (2) finding a nonempty independent set $S \subseteq V$ such that $|N(S)| = O(n^{\alpha} \cdot |S|)$, and (3) finding two vertices $x, y \in V$ that have the same color in every $3$-coloring.} which are used in solving the large-degree case including~\cite{KTY24}, we do not take them into account in this paper.

The following framework is widely used: for an arbitrary threshold $\Delta$, it suffices to consider making progress toward $\widetilde{O}(n^{\alpha})$-coloring for graphs with \emph{maximum degree} $\le \Delta$ and for graphs with \emph{minimum degree} $\ge \Delta$ (\cite[Proposition~7.1]{KT17}). For this reason, an improvement on either the large-degree case or the small-degree case yields a better algorithm for $3$-colorable graphs.

\paragraph{The starting point: Wigderson's algorithm.}
For graphs with minimum degree $\Delta$, the induced subgraph $G[N(x)]$ on the neighborhood of $x \in V$ is bipartite, with at least $\Delta$ vertices. Hence, an independent set of size $\ge \Delta/2$ can be obtained in linear time (by bipartite coloring), making progress toward $O(n/\Delta)$-coloring. Graphs with maximum degree $\Delta$ can be colored with $\Delta + 1$ colors by greedy coloring. Wigderson combined them with threshold $\Delta = \sqrt{n}$ and yielded an $O(\sqrt{n})$-coloring algorithm.

\paragraph{Large-degree case and combinatorics.}
In 1989, Blum achieved the first improvement over Wigderson’s algorithm; his line of work~\cite{Blu89, Blu90, Blu94} yields a better algorithm for the large-degree cases, which makes progress toward $\widetilde{O}((n/\Delta)^{3/5})$-coloring.
The attack on the large-degree case is mainly done in a \emph{combinatorial} way, analyzing dense substructures of a graph.
The best known algorithm for the large-degree case makes progress toward $(n/\Delta)^{1/2} n^{o(1)}$-coloring for every fixed $\varepsilon > 0$, whenever the minimum degree $\Delta > \sqrt{n}$~\cite{KTY24}.
Therefore, we can use the following tool when combined with the small-degree case.

\begin{theorem}[\cite{KTY24}, restated]
    \label{thm:kty24} Fix $\alpha \in(0, \frac{1}{4}]$.
    Suppose that there exists an algorithm that, given a $3$-colorable graph with maximum degree $\Delta < n^{1-2\alpha}$, obtains an independent set of size $n^{1-\alpha-o(1)}$ in polynomial time.
    Then, one can color a $3$-colorable graph with $n^{\alpha+o(1)}$ colors in polynomial time.
\end{theorem}

\paragraph{Small-degree case and SDP.}

On the other hand, improvements on the small-degree case generally involve \emph{semidefinite programming (SDP)}.
Goemans and Williamson's 1994 landmark result, which achieves a $0.878$-approximation of the maximum cut problem via SDP~\cite{GW95}, sparked interest in applying SDP to other optimization problems.
In the same year, Karger, Motwani, and Sudan~\cite{KMS94, KMS98} presented an SDP-based algorithm that, given a $3$-colorable graph with maximum degree $\Delta$, computes an $\widetilde{O}(\Delta^{1/3})$-coloring in polynomial time.
Their SDP formulation is \textbf{``vector $k$-coloring''}, defined to be vectors $\{v_{x}\}_{x\in V}$ in $\R^{n}$ such that $\norm{v_x} = 1$ for all $x \in V$ and $\ip{v_x}{v_y} \le -\frac{1}{k-1}$ for all $xy \in E$.
It is not hard to see that if the graph is $k$-colorable, then a vector $k$-coloring exists and can be computed in polynomial time via a polynomial-time SDP algorithm~\cite {GLS81}.
A rounding scheme on vector $3$-coloring enabled them to obtain an independent set of size $\widetilde{\Omega}(n / \Delta^{1/3})$ with high probability.

Combining this with~\cite{KTY24}'s framework (\autoref{thm:kty24}) gives $n^{1/5+o(1)}$ colors, splitting at $\Delta = n^{3/5}$.
As further studies on the SDP side~\cite{ACC06, Chl07, BHL27} improve the bounds when $\Delta$ is large, the best known algorithm~\cite{BHL27} achieves an exponent slightly better than $\frac{1}{5}$.

\paragraph{Related works.}
Some studies consider coloring 3-colorable graphs with a specific structure. Alon and Kahale \cite{AK97} showed, using spectral graph theory, a polynomial-time algorithm to find a 3-coloring for random 3-colorable graphs. Recently, Hsieh \cite{Hsi26} presented a polynomial-time algorithm to find a large independent set for 3-colorable graphs with low threshold rank.
\subsection{Our Results}\label{sec:1.3}

In this paper, we make an improvement to the small-degree cases for the problem of coloring $3$-colorable graphs in polynomial time.
Specifically, we obtain the following results.

\begin{theorem}
    \label{thm:vector-k-coloring} Fix a real number $k>2$ and $\beta\in(0,1]$.
    Let $G=(V,E)$ be a graph on $n$ vertices with maximum degree at most $n^{\beta}$, given together with an ordinary vector $k$-coloring, that is, unit vectors $\{v_{x}\}_{x\in V}$ with $\ip{v_x}{v_y}\le-\frac{1}{k-1}$ for every edge $xy\in E$.
    For every fixed $C > 0$, there is a randomized polynomial-time algorithm that, with probability at least $1-n^{-C}$, returns an independent set $I$ with
    \begin{equation*}
        |I|
        \ge
        \begin{cases}
        n^{1-\frac{k-2}{k}\beta-o(1)}, & 0<\beta\le\frac{k-2}{k}, \\[1mm]
        n^{1-\frac{(k-2)\beta}{2+k\beta}-o(1)}, & \frac{k-2}{k}<\beta\le1.
        \end{cases}
    \end{equation*}
\end{theorem}

\begin{restatable}
    {theorem}{sparsethm} \label{thm:vector-3-coloring} Fix $\beta\in(0,1]$.
    Let $G =(V,E)$ be a $3$-colorable graph on $n$ vertices with maximum degree at most $n^{\beta}$, given together with a strict vector $3$-coloring, that is, unit vectors $\{v_{x}\}_{x\in V}$ with $\ip{v_x}{v_y}=-\frac{1}{2}$ for every edge $xy\in E$.
    For every fixed $C > 0$, there is a randomized polynomial-time algorithm that, with probability at least $1-n^{-C}$, returns an independent set $I$ with
    \begin{equation*}
        |I|
        \ge
        \begin{cases}
        n^{1-\beta/3-o(1)}, & 0<\beta\le1/3, \\[1mm]
        n^{\frac{2+2\beta}{2+3\beta}-o(1)}, & 1/3<\beta\le1.
        \end{cases}
    \end{equation*}
\end{restatable}

Note that \autoref{thm:vector-k-coloring} achieves better performance than the algorithm of Karger, Motwani, and Sudan \cite{KMS98} when $\beta > \frac{k-2}{k}$ and the same performance when $0 < \beta \leq \frac{k-2}{k}$. The main ideas behind \autoref{thm:vector-k-coloring} already appear in the case $k=3$.
We therefore first prove \autoref{thm:vector-3-coloring}, which assumes a strict vector $3$-coloring.
In \autoref{sec:A}, we extend the result from strict vector $3$-colorings to ordinary vector $k$-colorings for every real $k>2$ and thereby prove \autoref{thm:vector-k-coloring}. Note that \autoref{thm:vector-k-coloring} does not require the parameter $k$ to be an integer.

Since a strict vector $k$-coloring can be computed in polynomial time by SDP~\cite{KMS98}, by combining \autoref{thm:vector-3-coloring} with~\cite{KTY24}'s framework (\autoref{thm:kty24}), we obtain the following result.

\begin{theorem}
    \label{thm:coloring-1} 
    For every fixed $C > 0$, there is a randomized polynomial-time algorithm that, with probability at least $1-n^{-C}$, properly colors every $3$-colorable $n$-vertex graph with $n^{1/6+o(1)}$ colors.
\end{theorem}

\begin{proof}
    Given a 3-colorable graph with maximum degree at most $n^{2/3}$, we compute a strict vector $3$-coloring by SDP (\autoref{sec:2.1}), and apply \autoref{thm:vector-3-coloring} with $\beta = \frac{2}{3}$. For any fixed $C' > 0$, this returns an independent set of size $n^{5/6-o(1)}$ in polynomial time with probability at least $1 - n^{-C'}$. Thus, we can apply \autoref{thm:kty24} with $\alpha = \frac{1}{6}$ and obtain a coloring with $n^{1/6+o(1)}$ colors. Since we call the independent-set algorithm only polynomially many times, choosing a sufficiently large $C' > C$ guarantees the success probability of $1 - n^{-C}$.
\end{proof}

Moreover, we further strengthen the SDP algorithm specifically for 3-colorable graphs (beyond vector 3-colorable graphs) to obtain a larger independent set. We prove the following theorem.

\begin{restatable}{theorem}{improvedisthm}
    \label{thm:improved-is}
    Let $\beta = 0.692601$. Let $G = (V, E)$ be a $3$-colorable graph on $n$ vertices with maximum degree at most $n^\beta$. For every fixed $C > 0$, there is a randomized polynomial-time algorithm that, with probability at least $1 - n^{-C}$, returns an independent set $I$ with $|I| \geq n^{(1+\beta)/2 - o(1)}$.
\end{restatable}

Combining \autoref{thm:improved-is} with \cite{KTY24}'s framework (\autoref{thm:kty24}), we obtain the following result.

\begin{theorem}
    \label{thm:improved-coloring}
    For every fixed $C > 0$, there is a randomized polynomial-time algorithm that, with probability at least $1-n^{-C}$, properly colors every $3$-colorable $n$-vertex graph with $O(n^{0.1537})$ colors.
\end{theorem}

\begin{proof}
    By \autoref{thm:improved-is}, we can apply \autoref{thm:kty24} with $\alpha = \frac{1-\beta}{2} = \frac{1 - 0.692601}{2} < 0.1537$, and obtain a coloring with $O(n^{0.1537})$ colors. Since we call the independent-set algorithm only polynomially many times, choosing a sufficiently large $C' > C$ guarantees the success probability of $1 - n^{-C}$.
\end{proof}
\subsection{New Bounds on Vector \texorpdfstring{$k$}{k}-Colorable Graphs} \label{sec:1.4}

We recall that our result (\autoref{thm:vector-k-coloring}) can be applied not only to $k$-colorable graphs but also to \emph{vector $k$-colorable graphs}, that is, graphs that have (ordinary) vector $k$-colorings.

\paragraph{Lower bounds.}

While $k$-colorable graphs are always vector $k$-colorable, the converse may not hold, and their chromatic numbers can be arbitrarily large. Feige, Langberg, and Schechtman~\cite{FLS04} give the following lower bound. In particular, there exist vector $3$-colorable graphs with chromatic number $\Omega(n^{0.1578})$.

\begin{theorem}[\cite{FLS04}, restated] \label{thm:fls04}
    Let $k > 2$ be a real constant. For infinitely many $n$, there exists an $n$-vertex vector $k$-colorable graph with chromatic number $n^{\alpha(k)-o(1)}$, where
    \begin{equation*}
        \alpha(k)
        = \max_{0 < a < \sqrt{\frac{k-2}{2(k-1)}}} \frac{\log{(1-a^2)}}{\log{\left\{\left(1-\frac{1}{(k-1)^2}\right)\left(1-\frac{2(k-1)}{k-2}a^2\right)\right\}} - \log{(1-a^2)}}.
    \end{equation*}
\end{theorem}

\paragraph{Upper bounds.}

In search of better coloring algorithms, Karger, Motwani, and Sudan~\cite{KMS98} presented the first rounding scheme on vector $k$-coloring. For a vector $k$-colorable graph with maximum degree $n^{\beta}$, their rounding scheme finds an independent set of size $n^{1 - \frac{k-2}{k}  \beta - o(1)}$. This shows that the chromatic number of vector $k$-colorable graphs is at most $n^{1-\frac{2}{k}+o(1)}$. Moreover, they introduced the following recursive technique that improves the bound when $\beta$ is large.

\begin{lemma}[\cite{KMS98}, restated] \label{lem:kms-recursion}
    For integers $k \ge 2$, suppose that an independent set of size $n^{1-\gamma_k(\beta)-o(1)}$ can be found in polynomial time for a given vector $k$-colorable graph with maximum degree $\le n^{\beta}$. Then, an independent set of size $n^{1-a_k-o(1)}$ can be found in polynomial time for a given vector $k$-colorable graph, where
    \begin{equation*}
        a_2 = 0, \quad a_k = \min_{0 < \beta \le 1} \max \{\gamma_k(\beta), 1 - \beta  (1 - a_{k-1})\} \quad (k \ge 3).
    \end{equation*}
    We can also obtain a coloring with $n^{a_k+o(1)}$ colors.
\end{lemma}

\begin{proof}
    Let $k \ge 3$, and let $G = (V, E)$ be a vector $k$-colorable graph. Then, for $x \in V$, its neighborhood $G[N(x)]$ is vector $(k-1)$-colorable, as shown in~\cite[Lemma~4.3]{KMS98}. Then, by applying the algorithm for vector $(k-1)$-colorable graphs to $G[N(x)]$, we obtain an independent set of size $|N(x)|^{1-a_{k-1}-o(1)}$, which is $n^{\beta  (1-a_{k-1}) - o(1)}$ for the largest-degree vertex $x$. Hence, \autoref{lem:kms-recursion} follows.
\end{proof}

This improves the bound from $a_k = 1 - \frac{2}{k}$ to $a_k = 1 - \frac{3}{k+1}$, i.e., the chromatic number of vector $k$-colorable graphs is at most $n^{1-\frac{3}{k+1}+o(1)}$. Then, Arora, Chlamt\'a\v{c}, and Charikar~\cite{ACC06} improved the bound to $O(n^{0.2448})$ for the case $k = 3$. Chlamt\'a\v{c} and Haviv~\cite{CH14} further improved the bound for \emph{strict} vector $k$-coloring. However, the lower bound by~\cite{FLS04} (empirically) has the exponent of $1 - \frac{4}{k} + O(k^{-2})$, leaving a substantial $\Theta(k^{-1})$ gap, as illustrated in the graph of \autoref{fig:gap-graph}.

\paragraph{Our improvements.}

Our new rounding scheme for vector $k$-colorable graphs (\autoref{thm:vector-k-coloring}) obtains an independent set of size $n^{1 - \frac{(k-2)\beta}{2+k\beta} - o(1)}$ for a vector $k$-colorable graph with maximum degree $n^{\beta}$. Hence, substituting $\beta = 1$, the chromatic number of vector $k$-colorable graphs is at most $n^{1-\frac{4}{k+2} + o(1)}$. This reduces the gap to the lower bound to only $\approx 2k^{-2}$ (empirically), showing the asymptotic tightness (with respect to $k$) of the lower-bound exponent by~\cite{FLS04}. With recursion (\autoref{lem:kms-recursion}), we obtain an even tighter upper bound: $O(n^{0.1836})$ colors for the case $k = 3$, which reduces the gap to only $\approx \frac{11}{3} k^{-3}$ (empirically). See \autoref{fig:gap-graph} for the comparison of the exponents.

\begin{figure}[t]
    \centering
    \includegraphics[width=\linewidth]{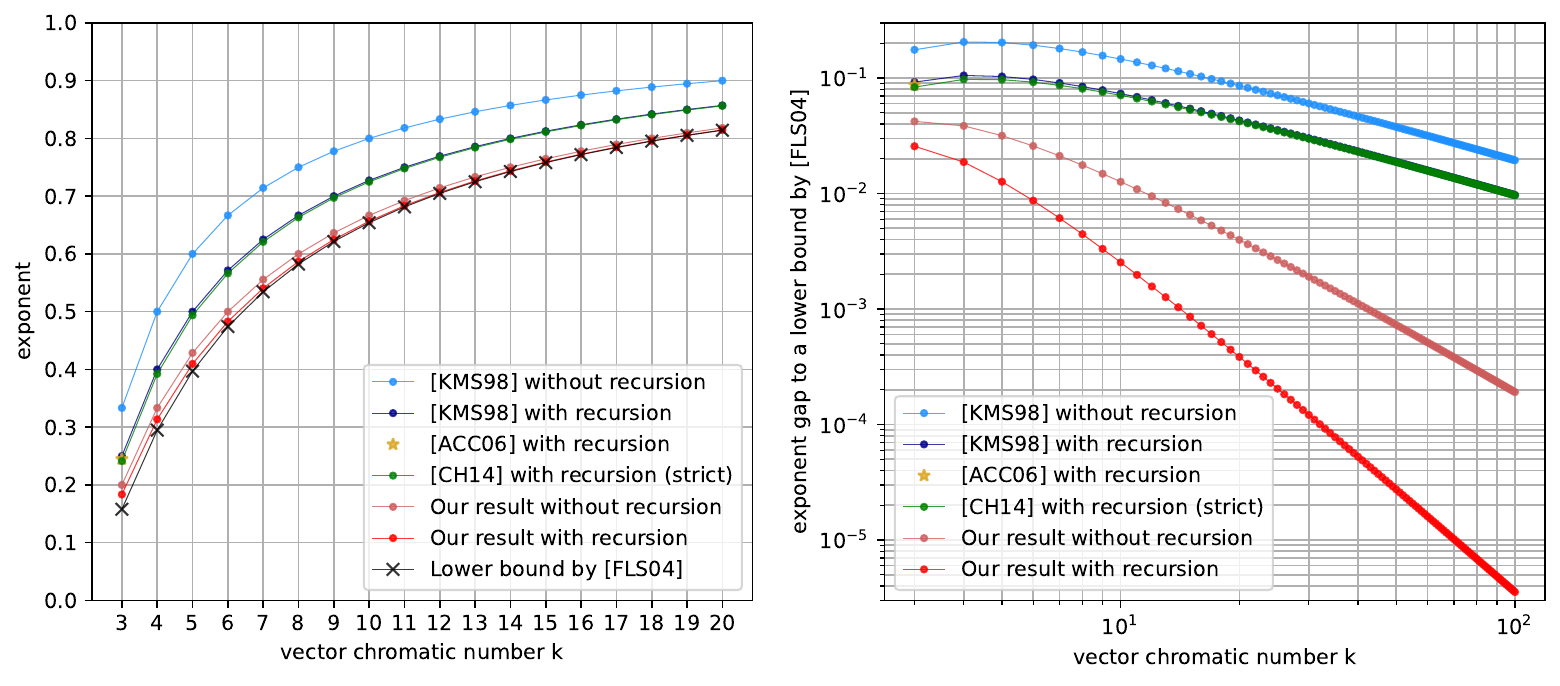}
    \caption{The comparison of the performance of rounding schemes (\cite{KMS98}, \cite{ACC06}, \cite{CH14}, and our result) and a lower bound by~\cite{FLS04}.}
    \label{fig:gap-graph}
\end{figure}

\begin{remark}
    We note that our main result (\autoref{thm:improved-coloring}) achieves $O(n^{0.1537})$ colors, which surpasses the lower bound of $\Omega(n^{0.1578})$ colors for vector $3$-colorable graphs. This is possible because we use the graph's ``$3$-colorability`` (not just ``vector $3$-colorability''), especially in the combinatorial algorithm for the large-degree case (i.e., the framework of \cite{KTY24}). Moreover, in \autoref{sec:4}, we also use a vector coloring SDP with additional constraints to improve the bound for the small-degree case; this is not possible for general vector $3$-colorable graphs. Some previous studies also use higher-degree SDP, including~\cite{Chl07, BHL27}. \label{rem:vector-3-coloring}
\end{remark}
\subsection{Organization}\label{sec:1.5}

We begin with an overview of the algorithms and proof techniques in \autoref{sec:overview}. Then, in \autoref{sec:2}, we review preliminaries on the vector coloring SDP, the geometry, and Gaussian bounds. In \autoref{sec:3}, we present an algorithm that finds a large independent set using only the strict vector $3$-coloring SDP (\autoref{thm:vector-3-coloring}), which yields an $O(n^{1/6+\varepsilon})$-coloring algorithm combined with \autoref{thm:kty24}. We also extend the result from strict vector $3$-colorings to ordinary vector $k$-colorings in \autoref{sec:A}. In \autoref{sec:4}, we develop a strengthened SDP-based algorithm and establish a stronger independent-set guarantee (\autoref{thm:improved-is}) that yields an $O(n^{0.1537})$-coloring algorithm. Finally, \autoref{sec:B} discusses polynomial-time implementation and approximate SDP constraints, while \autoref{sec:C} provides computational verification of the numerical inequalities used in our analysis.
\section{Overview of Proof}\label{sec:overview}

Here, we briefly present our two algorithms and the main ideas behind their analysis.

\subsection{The \texorpdfstring{$n^{1/6+o(1)}$}{O(n\textasciicircum{}(1/6+o(1)))}-Coloring Algorithm (\autoref{sec:3})}\label{sec:1b.1}

\paragraph{Our algorithm.} By \autoref{thm:kty24}, it suffices to design an algorithm that finds a sufficiently large independent set in graphs with maximum degree below $n^{2/3}$. First, we calculate a (strict) vector $3$-coloring of the graph that satisfies
\begin{equation*}
    \ip{v_{x}}{v_{y}}
    = -\frac{1}{2}
\end{equation*}
for all edges $xy \in E$. Then, our algorithm performs \emph{Gaussian local-maximum selection}, which is a novel technique. In Gaussian local-maximum selection, we sample a random Gaussian vector $g \sim N(0, I)$, and compute the Gaussian projection $\ip{g}{v_{x}}$ for each $x \in V$. Next, let $I$ be the set of vertices whose Gaussian projection is larger than that of any neighboring vertex (see \autoref{fig:101}). Then $I$ is always an independent set, and in many cases, $|I| \ge n^{5/6 - o(1)}$ holds. Intuitively, the independent set is large due to $\ip{v_{x}}{v_{y}} = -\frac{1}{2}$ for all $xy \in E$.

\begin{figure}[!b]
    \centering
    \includegraphics[width=\linewidth]{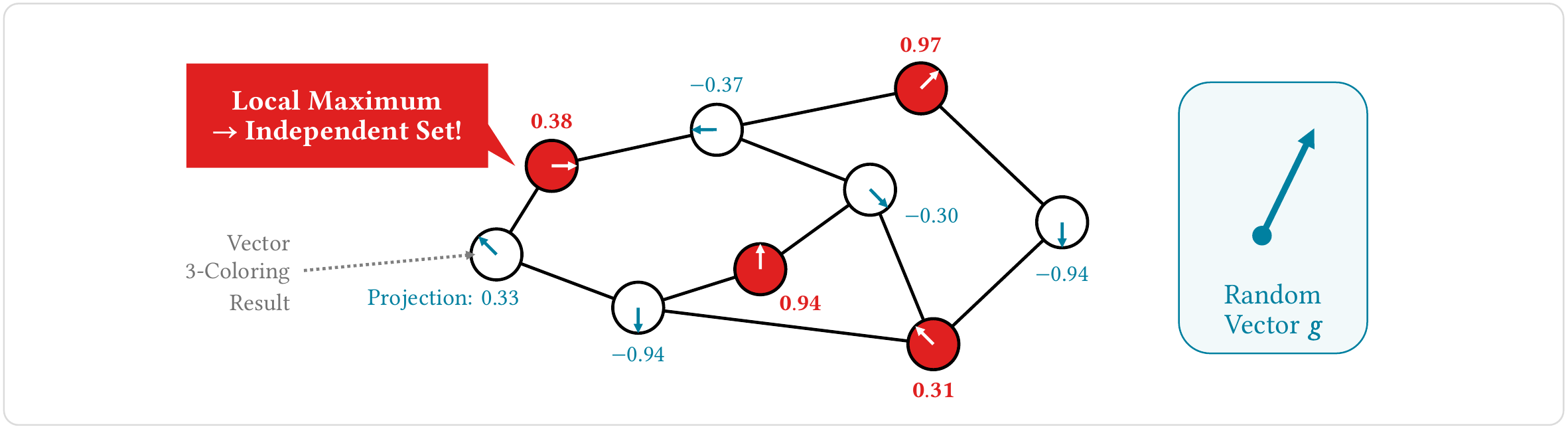}
    \caption{A sketch of the Gaussian local-maximum selection technique. The arrow inside each vertex results from the vector $3$-coloring.}
    \label{fig:101}
\end{figure}

However, only one attempt at Gaussian local-maximum selection is not enough to guarantee an independent set of size $n^{5/6 - o(1)}$ for every graph. Therefore, our algorithm constructs a sequence of nested cores $F_0 \supseteq F_1 \supseteq \dots$, with depth $\ell$ nearly $\log{n}$. We begin with $F_0 = V$, and at each level $j$ ($= 1, 2, \dots$), we sample Gaussian local maxima of $G[F_{j-1}]$ many times (e.g., $\approx n^{1/3}$ times). If all of them fail to obtain a large independent set, we delete vertices that are frequently selected (e.g., those whose empirical estimate is at least $\frac{\delta}{2}$, where $\delta \approx n^{-1/6}$) and let $F_j$ be all remaining vertices. For each remaining vertex $x$, we can estimate that the probability of selecting $x$ (by local-maximum selection) is less than $\delta$ because of the large number of samples. We call this the \emph{probability estimate}; it will be used in the analysis.

\paragraph{Proof approach.} We show that it is impossible for the algorithm to fail at all levels, assuming that the probability estimate is accurate. The key idea is to utilize the Gaussian cover~\cite{ACC06}; a \emph{Gaussian cover} is a collection of vectors $Z = \{z_1, z_2, \dots\}$ such that for a random Gaussian vector $g$, the Gaussian projection onto at least one vector in $Z$ exceeds a prescribed threshold $s$, with high probability $q \approx 1$. We prove the algorithm's guarantee by constructing Gaussian covers with progressively increasing thresholds $s$ and deriving a contradiction.

\paragraph{Constructing a Gaussian cover.} The question is how to obtain the Gaussian cover that proves contradiction. First, we pick a vertex $r \in F_{\ell}$ from the deepest level $\ell$ (note that $F_{\ell} \neq \emptyset$ due to \autoref{lem:core-nonempty}) and obtain the initial Gaussian cover with vertex set $W_0 \leftarrow N_{F_{\ell-1}}(r)$ (\autoref{cor:initial-cover}). Then, we iteratively ``lift'' the Gaussian cover to shallower levels, increasing the threshold at the cost of only a negligible loss of probability. Specifically, we let
\begin{equation*}
W_1 \leftarrow N_{F_{\ell-2}}(W_0), W_2 \leftarrow N_{F_{\ell-3}}(W_1), \dots,
\end{equation*}
where $W_j$ is the vertex set of the $j$-th Gaussian cover (see also \autoref{fig:102}). Importantly, the Gaussian projections we use in Gaussian covers are not $\ip{g}{v_y}$ for each $y \in W_j$ ($j \ge 0$), but $\ip{g}{u_y}$, where $u_y$ is the projection of $v_y$ that is orthogonal to $v_r$. This means that we show
\begin{equation*}
    \Prb_g\left[\max_{y \in W_j} \ip{g}{u_y} > s_j\right]
    = 1 - o(1)
\end{equation*}
for some $s_j$. On the other hand, by the definition of $g$ and the fact that $\norm{u_y} \le 1$,
\begin{equation*}
    \Prb_g\left[\max_{y \in W_j} \ip{g}{u_y} > s_j\right]
    < ne^{-s_j^2/2}
\end{equation*}
can be easily shown. Therefore, if $s_j$ becomes large enough before reaching $F_0$, we get a contradiction, showing that the algorithm never fails.

\begin{figure}[!b]
    \centering
    \includegraphics[width=\linewidth]{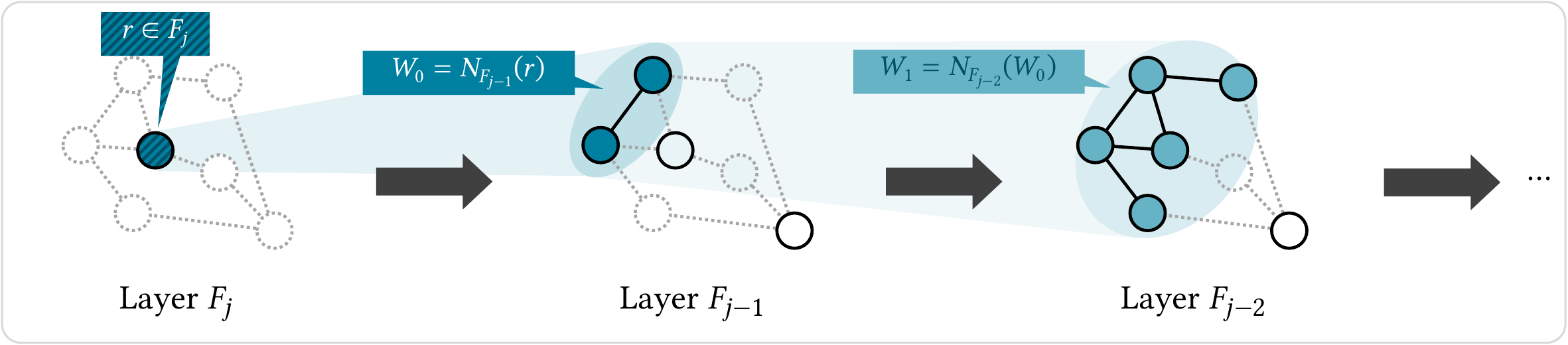}
    \caption{A sketch of lifting a Gaussian cover. Vertices with black outlines belong to $F_j$, $F_{j-1}$, and $F_{j-2}$. Vertices with gray outlines do not belong to them.}
    \label{fig:102}
\end{figure}

\paragraph{Lifting lemma.} The remaining step is to show that $\{u_y\}_{y \in W_0}, \{u_y\}_{y \in W_1}, \dots$ are actually Gaussian covers. The key ingredient is a lifting lemma (\autoref{lem:lifting}) that propagates a Gaussian cover from one core to the preceding core, increasing the threshold at the cost of losing a negligible loss of probability. To understand this propagation, fix a root $r$ in the deepest core, a vertex $x$ with $\norm{u_x}>0$, and a neighbor $y$ of $x$. Then, as written in \autoref{sec:2.2}, the strict vector coloring constraint $\langle v_x, v_y \rangle = -\frac{1}{2}$ directly implies
\begin{equation}\label{eq:first}
    \ip{g}{u_y}
    = -\left(\frac{1}{2}+\frac{\sqrt{3}}{2}\cdot\frac{\xi_y\ip{v_r}{v_x}}{\norm{u_x}}\right)\ip{g}{u_x} + \frac{\sqrt{3}}{2}\ip{g}{w_y}.
\end{equation}
Here $\xi_y$ denotes the signed in-plane component of the edge residual, and $w_y$ is orthogonal to the plane spanned by $v_r$ and $v_x$, satisfying $\norm{w_y}^2 = 1 - (\xi_y)^2$.

The difficulty for analysis is that the Gaussian projection $\ip{g}{u_y}$ in \autoref{eq:first} depends on both $\xi_y$ and $\sqrt{1-(\xi_y)^2}$. Indeed, $\ip{g}{u_y}$ consists of a contribution inherited from $x$ (the first term on the right-hand side of (\ref{eq:first})) and an independent transverse Gaussian contribution (the second term), where the former depends on $\xi_y$ and the latter depends on $\sqrt{1-(\xi_y)^2}$.

\paragraph{Resolving the difficulty.} To resolve the difficulty, we made use of Ehrhard's inequality \cite{Ehr83} (\autoref{thm:ehrhard}). First, we obtain two convex sets with contrasting Gaussian measures, using the degree bound $(\leq n^{2/3})$ and the small local-maximum probability for each vertex (see \autoref{lem:source-body} and the proof of \autoref{lem:wy-transverse-bound}). Then, we applied Ehrhard's inequality and yielded a lower bound on the probability that some neighbor attains a sufficiently large transverse Gaussian projection, with a threshold depending on both $\xi_y$ and $\sqrt{1-(\xi_y)^2}$ (\autoref{lem:wy-transverse-bound}).

This method is enough to show that if Gaussian projection at vertex $x$ is below $-s$ for the current cover threshold $s$, then with high probability, some neighbor $y$ has a Gaussian projection above a higher threshold $s' > s$, by a positive margin (\autoref{lem:U-increase}). This establishes the desired lifting lemma (\autoref{lem:lifting}). Moreover, this algorithm extends to ordinary vector $k$-coloring from the $k = 3$ case, and we use this extension in the band rounding of the improved algorithm.
\subsection{The \texorpdfstring{$O(n^{0.1537})$}{O(n\textasciicircum{}0.1537)}-Coloring Algorithm (\autoref{sec:4})}\label{sec:1b.2}

\paragraph{Our algorithm.} Our $O(n^{0.1537})$-coloring algorithm performs both Gaussian local-maximum selection and a new technique called \emph{band rounding}. The band rounding procedure is as follows. We first partition the range $[-\frac{1}{2},1]$ into $O(\log{n})$ intervals whose length is at most $\frac{1}{\log{n}}$. For each vertex $r$ and interval $P$, we define the \emph{band} $S_{r,P}$ as the set of vertices $x$ satisfying $\ip{v_r}{v_x}\in P$. For each band, we aim to obtain a large independent set; we omit the details due to the page limit, but we apply our generalized vector-coloring algorithm in \autoref{sec:A} and also use additional SDP constraints valid when $G$ is 3-colorable. See \autoref{sec:4.1}--\autoref{sec:4.3} for details.

\paragraph{Proof approach.} The goal is to show that the probability that both Gaussian local-maximum selection and band rounding fail to obtain an independent set of size $n^{1-\alpha}$ is very low for $\alpha \approx 0.1536995$. The crux of the proof is to show that if both algorithms fail, then
\begin{equation}\label{eq:magic}
    \Prb\left[\max_{x \in N_{F_j}(r)} \frac{\ip{g}{u_x}}{\norm{u_x}} > At\right] \le t^{-4}\bPhi(t/2),
\end{equation}
\noindent
holds at $A = 1.968$ for sufficiently deep core $j \geq c$ $(c = o(\ell))$. As in \autoref{sec:1b.1}, $F_j$ is the $j$-th core of the Gaussian local-maximum selection algorithm, and $u_x$ is the projection of $v_x$ orthogonal to $v_r$, where $(v_1, \dots, v_n)$ is the result of a (strict) vector 3-coloring used in the Gaussian local-maximum selection algorithm. Importantly, this tail bound is stronger than the bound used in the $n^{1/6+o(1)}$-coloring algorithm.

Assuming \autoref{eq:magic} holds, the lifting lemma can be significantly improved. Unfortunately, without the band rounding algorithm and \autoref{eq:magic}, we cannot apply the lifting lemma when the target independent-set size exceeds $\Omega(n^{5/6})$. However, thanks to \autoref{eq:magic}, we can construct more favorable convex sets to apply Ehrhard's inequality, and the lifting lemma still works even if the target independent-set size is $n^{1-\alpha}$. Therefore, by iteratively using the improved lifting lemma, starting from $F_\ell$, we ultimately reach a contradiction by  \autoref{eq:magic} before $F_{o(\ell)}$, which implies that the algorithm (almost) never fails. The sketch is shown in \autoref{fig:103}.

\begin{figure}[t]
    \centering
    \includegraphics[width=\linewidth]{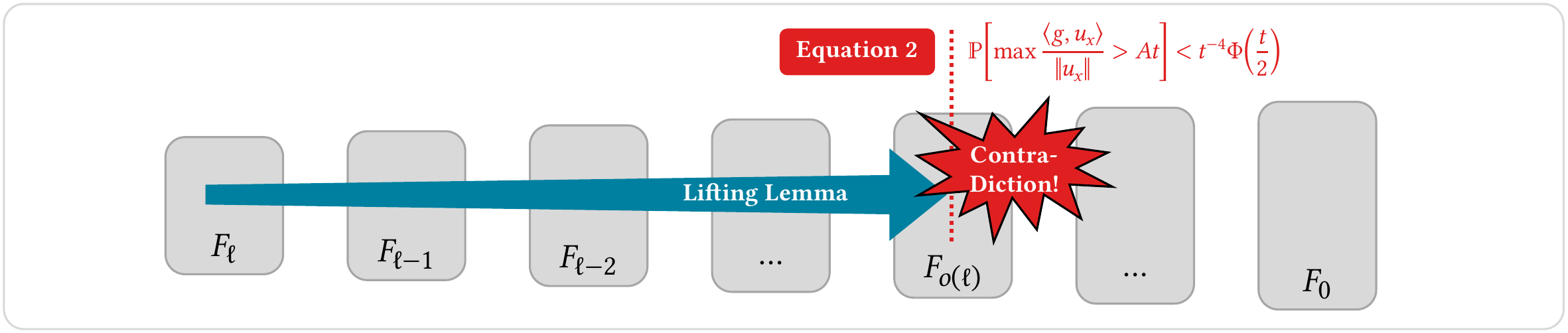}
    \caption{The sketch of our proof approach. We first prove the tail upper bound of the probability using \autoref{eq:magic}, and then derive a contradiction using an improved lifting lemma.}
    \label{fig:103}
\end{figure}

\paragraph{Proving Equation (2).} First, \autoref{eq:magic} is equivalent to bounding the probability of a \emph{local event}, where at least one neighbor of $r$ has a large Gaussian projection. Formally, we define
\begin{equation*}
    \cL_i(r)
    = \left\{\exists x \in N_{F_{i+1}}(r): \frac{\ip{g}{u_x}}{\norm{u_x}} > A_i t\right\}
\end{equation*}
\noindent
and our goal is to show $\mathbb{P}[\cL_i(r)] < (i+2) \varepsilon_t$ at $\varepsilon_t \approx e^{-t^2/(8-\varepsilon)}$. Proving this at $A_i = 1.968$ for all $i \geq c$, where $c = o(\ell)$, directly implies that \autoref{eq:magic} holds. We show this bound by induction on $i$. We start at $i = 0$ with $A_0 \approx 2.18$, and use gradually decreasing $A_i$'s in the induction step, eventually reaching $A_c = 1.968$ at $c = o(\ell)$. However, it is very difficult to show $\mathbb{P}[\cL_i(r)] < (i+2) \varepsilon_t$ directly. Therefore, we introduce the \emph{global event} (\autoref{fig:104} (b)), that at least one of the vertices in core $F_{i+1}$ has a large Gaussian projection. Formally,
\begin{equation*}
    \cG_i(r)
    = \left\{\exists x \in F_{i+1}: |\ip{g}{u_x}| > q_{A_i}(\ip{v_r}{v_x})t\right\},
\end{equation*}
\noindent
where $q_{A_i}(\rho) : [-\frac{1}{2}, 1] \to \mathbb{R}^{+}$ is a function dependent on the correlation $\rho$. Now, let $\cE_{i+1}(r) \in \{\cL_{i+1}(r), \cG_{i+1}(r)\}$, which indicates any of two events. Then, obviously
\begin{equation*}
    \Prb[\cE_{i+1}(r)] \le \Prb[\cG_{i}(r)] + \Prb[\cE_{i+1}(r) \setminus \cG_i(r)]
\end{equation*}
\noindent
holds. In the induction step on $i$, we bound the probability of $\cE_{i+1}(r)$ by showing that both the first term $\Prb[\cG_{i}(r)]$ and the second term $\Prb[\cE_{i+1}(r) \setminus \cG_i(r)]$ are small. Especially, we aim $\mathbb{P}[\cE_i(r)] \leq (i+2)\varepsilon_t$. When $\cE_{i+1} = \cL_{i+1}$, we bound the probability of the desired local events.

\begin{figure}[t]
    \centering
    \includegraphics[width=\linewidth]{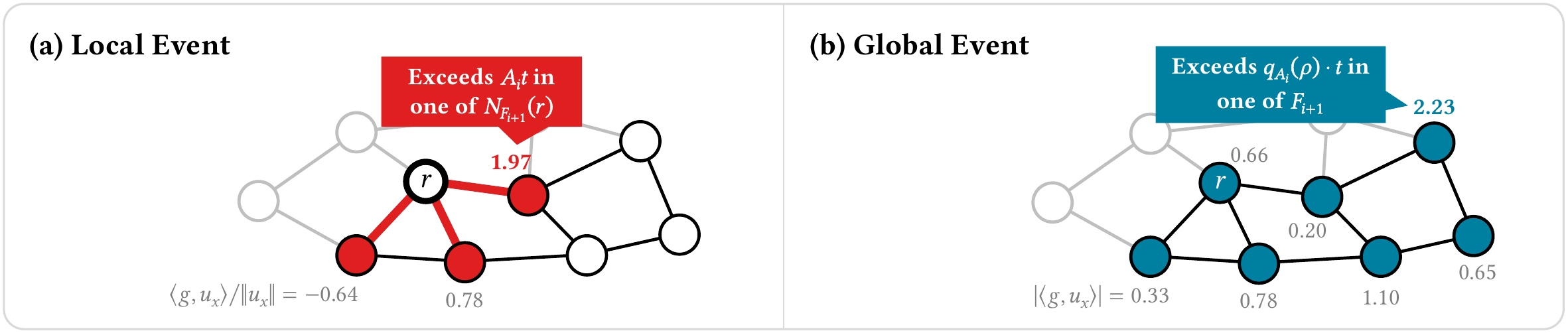}
    \caption{(a) The sketch of the local event. (b) The sketch of the global event. Vertices with black outlines belong to $F_{i+1}$. Vertices with gray outlines do not belong to them.}
    \label{fig:104}
\end{figure}

\paragraph{Bounding each term.} The first term can be bounded by induction on $i$. For the second term, when $\cE_{i+1} = \cG_{i+1}$, we use the fact that if the band rounding algorithm fails, then the size of each band is small for some correlation interval $P$ (this is exactly why the band rounding is required). By \autoref{lem:band-bound}, the band size $|S_{r, P}|$ is less than $f(\rho)$ for some $f : \rho \to \mathbb{R}^{+}$, where $\rho$ is the correlation, which provides a nontrivial improvement over the trivial bound $|S_{r, P}| \leq n$ for certain correlation ranges. Using this fact, we bound the second term by applying Ehrhard's inequality with \autoref{eq:first} (see \autoref{lem:local-bound} (i)).

When $\cE_{i+1} = \cL_{i+1}$, the intuitive reason for the small second term is that it is very rare to have a vertex $x$ with a large Gaussian projection in the neighbors of $r$, despite no vertices $y \in F_{i+1}$ that are adjacent to $x$ having a large absolute value of Gaussian projection. This is strange because if $\ip{g}{u_x}$ is large, then $\ip{g}{u_y}$ is likely to be negatively large by the negative correlation between $u_x$ and $u_y$ derived from the vector 3-coloring. In the formal proof, we apply the Gaussian H\"older inequality and Ehrhard's inequality; see \autoref{lem:neighborhood-preparation} and \autoref{lem:local-bound} (ii).
\section{Preliminaries}\label{sec:2}

For a subset $F\subseteq V$, write $N_F(x)=\{y\in F:xy\in E\}$ for $x\in V$ and $N_F(S)=\bigcup_{x\in S}N_F(x)$ for $S\subseteq V$; we write $N(x)=N_V(x)$ and $N(S)=N_V(S)$, and let $G[F]$ be the subgraph of $G$ induced by $F$, whose edge set we denote by $E[F]=\{xy\in E:x,y\in F\}$.
For a graph $G$, we write $\Delta(G)$ for its maximum degree.
Fix a total order on $V$ for breaking ties between Gaussian projections.
Define
\begin{equation*}
    I_F(g)
    =\{x\in F:\forall y\in N_F(x):(\ip{g}{v_x},x)>(\ip{g}{v_y},y)\}.
\end{equation*}
Then $I_F(g)$ is independent for every $g$.

Let $n=|V|$ be the number of vertices in the input graph.

\subsection{Vector Coloring and SDP Relaxation}\label{sec:2.1}
For an integer $k \ge 2$, we can relax the $k$-coloring problem to the following semidefinite program.
\begin{align*}
    \ip{v_x}{v_y}
    &= -\frac{1}{k-1}
    &&\text{for all } xy\in E,\\
    \norm{v_x}
    &= 1
    &&\text{for all } x\in V.
\end{align*}
We can easily verify this is a relaxation of the original problem by assigning to each vertex a vector corresponding to its color in a regular $(k-1)$-simplex, according to an unknown proper $k$-coloring.
The constraints themselves make sense for every real number $k \ge 2$, and we define vector colorings in this generality.
\begin{definition}[Strict vector $k$-coloring]
    Let $G = (V, E)$ be a graph and let $k \ge 2$ be a real number.
    A set of unit vectors $\{v_x\}_{x\in V}$ is called a \emph{strict vector $k$-coloring} if it satisfies $\ip{v_x}{v_y} = -\frac{1}{k-1}$ for every edge $xy \in E$.
\end{definition}
\begin{definition}[Ordinary vector $k$-coloring]
    \emph{Ordinary vector $k$-coloring} is the variant obtained by changing the equality in the definition of a strict vector $k$-coloring to $\le$.
    The smallest real $k \ge 2$ for which $G$ has an ordinary vector $k$-coloring is the vector chromatic number of $G$ introduced by Karger, Motwani, and Sudan~\cite{KMS98}, which is at most the chromatic number of $G$.
\end{definition}
Unless stated otherwise, a vector $k$-coloring means an ordinary vector $k$-coloring.

An SDP can be solved in polynomial time only up to an additive error.
Accordingly, when computing a strict vector $k$-coloring, we allow an additive error $\varepsilon_{\mathrm{SDP}}$ in the edge constraints.
Throughout, we fix $\varepsilon_{\mathrm{SDP}}=e^{-3n}$, so that $\log{(1/\varepsilon_{\mathrm{SDP}})}=3n$, and the computed vectors satisfy $-\frac{1}{k-1}\le\ip{v_x}{v_y}\le-\frac{1}{k-1}+\varepsilon_{\mathrm{SDP}}$ for every edge $xy\in E$.
In \autoref{sec:B.2}, we show that treating this constraint as exact is harmless for our analysis, which justifies the discussion below.

\subsection{Vector Identities}\label{sec:2.2}
Here, we introduce useful vector notation and deduce key vector identities.

Given unit vectors $\{v_x\}_{x\in V}$, we often project some $v_y$ along another $v_x$ and an orthogonal direction.
For $x,y\in V$ with $v_y\ne\pm v_x$, we write
\begin{equation*}
    v_{xy}
    =\frac{v_y-\ip{v_x}{v_y}v_x}{\norm{v_y-\ip{v_x}{v_y}v_x}}.
\end{equation*}
The vector $v_{xy}$ is a unit vector orthogonal to $v_x$, and $v_y=\ip{v_x}{v_y}v_x+\sqrt{1-\ip{v_x}{v_y}^2}\,v_{xy}$.
For $xy\in E$, we call $v_{xy}$ the \emph{edge residual} of $y$ at $x$.
Note that $v_{xy}\ne v_{yx}$ in general.
If $\{v_x\}_{x\in V}$ is a strict vector $3$-coloring and $xy\in E$, then $\ip{v_x}{v_y}=-\frac{1}{2}$, so $v_{xy}=\frac{2v_y+v_x}{\sqrt{3}}$ and
\begin{equation}\label{eq:residual}
    v_y
    =-\frac{1}{2}v_x+\frac{\sqrt{3}}{2}v_{xy}.
\end{equation}

In the rest of this subsection, let $\{v_x\}_{x\in V}$ be a strict vector $3$-coloring of $G$, and fix a vertex $r\in V$.
For $x\in V$, write $\rho_x=\ip{v_r}{v_x}$ and
\begin{equation*}
    u_x
    =v_x-\rho_xv_r,
\end{equation*}
so that $u_x\perp v_r$ and $\norm{u_x}=\sqrt{1-\rho_x^2}$; thus, in contrast to $v_x$, $u_x$ is not necessarily a unit vector.
If $\norm{u_x}>0$, then $u_x=\norm{u_x}v_{rx}$ and $v_{xr}=(v_r-\rho_xv_x)/\norm{u_x}$ by definition, where $v_{rx}\perp v_r$ and $v_{xr}\perp v_x$ are unit vectors, and $v_x$ and $v_{xr}$ form an orthonormal basis of $\Span(v_r,v_x)$.
Let $x\in V$ satisfy $\norm{u_x}>0$, and let $y\in N_V(x)$.
Since $v_{xy}\perp v_x$, we may decompose
\begin{equation}\label{eq:vxy-decomposition}
    v_{xy}
    =\xi_yv_{xr}+w_y,
\end{equation}
where $\xi_y=\ip{v_{xy}}{v_{xr}}$, $w_y\perp\Span(v_r,v_x)$ and $(\xi_y)^2+\norm{w_y}^2=1$.
Since $\ip{v_r}{v_{xr}}=(1-\rho_x^2)/\norm{u_x}=\norm{u_x}$, taking the inner product of \autoref{eq:residual} with $v_r$ gives $\rho_y=-\frac{1}{2}\rho_x+\frac{\sqrt{3}}{2}\xi_y\norm{u_x}$, that is,
\begin{equation}\label{eq:xi-rho}
    \frac{\sqrt{3}}{2}\xi_y
    =\frac{\rho_y+\rho_x/2}{\sqrt{1-\rho_x^2}}.
\end{equation}
The projection of $v_{xr}$ onto $v_r^{\perp}$ is $v_{xr}-\norm{u_x}v_r=-\rho_xu_x/\norm{u_x}$, and $w_y\perp v_r$.
Hence projecting \autoref{eq:residual} onto $v_r^{\perp}$ gives the first equality of
\begin{equation}\label{eq:uy-identity}
    u_y
    =-\left(\frac{1}{2}+\frac{\sqrt{3}}{2}\cdot\frac{\xi_y\rho_x}{\norm{u_x}}\right)u_x+\frac{\sqrt{3}}{2}w_y
    =-\frac{1/2+\rho_x\rho_y}{1-\rho_x^2}u_x+\frac{\sqrt{3}}{2}w_y,
\end{equation}
and the second equality follows from \autoref{eq:xi-rho} and $\norm{u_x}^2=1-\rho_x^2$.
In \autoref{eq:uy-identity}, $u_x\in\Span(v_r,v_x)$ is orthogonal to $v_r$, and $w_y\in\Span(v_r,v_x)^{\perp}$.

\subsection{Gaussian Notation and Bounds}\label{sec:2.3}
Throughout, $\log$ denotes the natural logarithm.
We write $\phi$, $\Phi$, and $\bPhi=1-\Phi$ for the density, the distribution function, and the upper tail of the standard normal distribution.
The vector $g\sim N(0,I)$ is a standard Gaussian vector in the ambient space of the vector coloring, and $\Prb_g$, $\E_g$ denote probability and expectation with respect to $g$; we also write $\gmeasure(K)=\Prb_g[g\in K]$ for the standard Gaussian measure of a set $K$.
We use the following two tail bounds.
\begin{fact}\label{fact:gaussian-tail}
    For every $u\ge0$,
    \begin{equation*}
        \frac{c_0}{u+1}e^{-u^2/2}
        \le \bPhi(u)
        \le e^{-u^2/2},
    \end{equation*}
    where $c_0>0$ is an absolute constant.
\end{fact}
We also use the following monotonicity property of the Gaussian tail.
\begin{fact}\label{fact:tail-monotone}
    The function $f(u)=e^{u^2/2}\bPhi(u)$ is strictly decreasing on $[0,\infty)$.
    Consequently, for all $0\le u\le v$,
    \begin{equation*}
        e^{(v^2-u^2)/2}\,\bPhi(v)
        \le\bPhi(u).
    \end{equation*}
\end{fact}
\begin{proof}
    We have $f'(u)=e^{u^2/2}\left(u\bPhi(u)-\phi(u)\right)$, and for every $u\ge0$,
    \begin{equation*}
        u\bPhi(u)
        =u\int_u^{\infty}\phi(s)\,\mathrm{d}s
        <\int_u^{\infty} s\,\phi(s)\,\mathrm{d}s
        =\int_u^{\infty}-\phi'(s)\,\mathrm{d}s
        =\phi(u),
    \end{equation*}
    so $f'(u)<0$.
    The displayed inequality is $f(v)\le f(u)$ rearranged.
\end{proof}

\begin{fact}\label{fact:gaussian-interval}
    If $Z\sim N(0,1)$, then for every $u\in\R$ and $\eta>0$,
    \begin{equation*}
        \Prb[Z\in[u,u+\eta]]
        \ge \eta\,\phi(|u|+\eta).
    \end{equation*}
\end{fact}

\begin{proof}
    Since $|z|\le |u|+\eta$ for every $z\in[u,u+\eta]$ and $\phi$ is decreasing as a function of $|z|$,
    \begin{equation*}
        \Prb[Z\in[u,u+\eta]]
        = \Phi(u+\eta) - \Phi(u)
        = \int_{u}^{u+\eta} \phi(z) \,\mathrm{d}z
        \ge \int_{u}^{u+\eta} \phi(|u|+\eta) \,\mathrm{d}z
        = \eta\,\phi(|u|+\eta).
    \end{equation*}
\end{proof}

We also use the following Gaussian H\"older inequality.
\begin{theorem}[Gaussian H\"older inequality~{\cite[Theorem~1(i)]{CDP15}}, two-block case]\label{thm:gaussian-holder}
    Let $(X_1,X_2)$ be a centered Gaussian random vector in $\R^{n_1}\times\R^{n_2}$ with covariance matrix $T=(T_{ab})_{a,b\in\{1,2\}}$, let $p_1,p_2\ge1$, and let $P=\operatorname{diag}(p_1T_{11},p_2T_{22})$.
    If $P-T\succeq0$, then $\E[f_1(X_1)f_2(X_2)]\le(\E[f_1(X_1)^{p_1}])^{1/p_1}(\E[f_2(X_2)^{p_2}])^{1/p_2}$ for all nonnegative measurable functions $f_1$ on $\R^{n_1}$ and $f_2$ on $\R^{n_2}$.
\end{theorem}

\begin{corollary}\label{cor:gaussian-holder}
    Let $\cU_1$ and $\cU_2$ be linear subspaces with $|\ip{u_1}{u_2}|\le\frac{1}{2}\norm{u_1}\norm{u_2}$ for all $u_1\in \cU_1$ and $u_2\in \cU_2$, and let $p_1,p_2>1$ satisfy $(p_1-1)(p_2-1)\ge\frac{1}{4}$.
    Then $\Prb_g[\cA_1\cap\cA_2]\le\Prb_g[\cA_1]^{1/p_1}\Prb_g[\cA_2]^{1/p_2}$ for all events $\cA_1$ and $\cA_2$ depending only on the Gaussian projections onto $\cU_1$ and $\cU_2$, respectively.
\end{corollary}

\begin{proof}
    Let $B_1$ and $B_2$ be matrices whose columns form orthonormal bases of $\cU_1$ and $\cU_2$, respectively, and put $X_1=B_1^{\mathsf{T}}g\in\R^{n_1}$ and $X_2=B_2^{\mathsf{T}}g\in\R^{n_2}$, where $n_1=\dim\cU_1$ and $n_2=\dim\cU_2$.
    Then $\cA_1=\{X_1\in A_1\}$ and $\cA_2=\{X_2\in A_2\}$ for some measurable sets $A_1$ and $A_2$, and the covariance matrix of $(X_1,X_2)$ has blocks $T_{11}=I$, $T_{22}=I$ and $T_{12}=B_1^{\mathsf{T}}B_2$.
    For $\alpha_1\in\R^{n_1}$ and $\alpha_2\in\R^{n_2}$, the assumption on $\cU_1$ and $\cU_2$ gives $|\alpha_1^{\mathsf{T}}T_{12}\alpha_2|=|\ip{B_1\alpha_1}{B_2\alpha_2}|\le\frac{1}{2}\norm{B_1\alpha_1}\norm{B_2\alpha_2}=\frac{1}{2}\norm{\alpha_1}\norm{\alpha_2}$, and hence
    \begin{align*}
        (\alpha_1,\alpha_2)^{\mathsf{T}}(P-T)(\alpha_1,\alpha_2)
        &=(p_1-1)\norm{\alpha_1}^2-2\alpha_1^{\mathsf{T}}T_{12}\alpha_2+(p_2-1)\norm{\alpha_2}^2\\
        &\ge(p_1-1)\norm{\alpha_1}^2-\norm{\alpha_1}\norm{\alpha_2}+(p_2-1)\norm{\alpha_2}^2\\
        &\ge\bigl(2\sqrt{(p_1-1)(p_2-1)}-1\bigr)\norm{\alpha_1}\norm{\alpha_2}
        &&(\because\text{AM--GM inequality})\\
        &\ge0.
        &&(\because(p_1-1)(p_2-1)\ge\tfrac{1}{4})
    \end{align*}
    Thus $P-T\succeq0$.
    Applying \autoref{thm:gaussian-holder} to $f_1=\mathbf{1}_{A_1}$ and $f_2=\mathbf{1}_{A_2}$, for which $\E[f_1(X_1)^{p_1}]=\Prb_g[\cA_1]$ and $\E[f_2(X_2)^{p_2}]=\Prb_g[\cA_2]$, gives the claim.
\end{proof}

\section{Finding a Large Independent Set from Vector \texorpdfstring{$3$}{3}-Coloring}\label{sec:3}

\autoref{tab:params-3} lists the parameters used in this section.

\begin{table}[!htbp]
    \centering
    \small
    \renewcommand{\arraystretch}{1.15}
    \setlength{\tabcolsep}{3pt}
    \caption{Parameters used in \autoref{sec:3}.}
    \label{tab:params-3}
    \begin{tabular}{llll}
        \hline
        Symbol & Low degree & High degree & Role \\
        \hline
        $\beta$ & $\in(0,1/3]$ & $\in(1/3,1]$ & exponent of the maximum degree: $\Delta(G)\le n^{\beta}$ \\
        $\ell$ & \multicolumn{2}{c}{$\lceil\log{n}\rceil$} & number of refinement levels \\
        \multicolumn{4}{l}{\textit{Parameters of \autoref{alg:sparse}}} \\
        $\delta$ & $\frac{1}{2}$ & $n^{-\beta/(2+3\beta)}\ell^{-500f(\beta)}$ & deletion threshold \\
        $T$ & $\lfloor np_*/2\rfloor$ & $\lceil n\delta/(100\ell^4)\rceil$ & target size of the independent set \\
        $R$ & $\lceil2\log{4}/p_*\rceil$ & $\lceil32\delta^{-2}\log{(8n\ell)}\rceil$ & number of runs of \autoref{alg:local-maxima} per level \\
        \multicolumn{4}{l}{\textit{Analysis of the lifting lemma}} \\
        $t$ & & $\sqrt{2\log{(1/\delta)}}$ & $\delta=e^{-t^{2}/2}$ \\
        $H_0$ & & $2\log{n}/t^{2}$ & $n=e^{H_0t^{2}/2}$ \\
        $D$ & & $2\beta\log{n}/t^{2}$ & $n^{\beta}=e^{Dt^{2}/2}$ \\
        $w$ & & $\log{t}/t^{2}$ & unit of lifting \\
        $H$ & & $H_0+40w$ & $\sqrt{H}$ is the target Gaussian cover threshold divided by $t$ \\
        $U_i$ & & $U_0+40iw$ & Gaussian cover threshold divided by $t$ after $i$ lifting steps \\
        \hline
    \end{tabular}%

\end{table}

In the low-degree regime, we write $p_*=\frac{1}{2}\bPhi\bigl(\sqrt{\tfrac{2}{3}\log{(2n^{\beta})}}\bigr)$.
In the high-degree regime, we write $f(\beta)=1/\sqrt{3\beta-1}$, which satisfies $f(\beta)\ge1/\sqrt{2}$ for $1/3<\beta\le1$.
We write $U_0=3/(2\sqrt{2})$ for the initial Gaussian cover threshold divided by $t$.

\subsection{Algorithmic Framework and Nested Cores}\label{sec:3.1}

In this subsection, we present the algorithm that obtains a large independent set with constant probability, based on Gaussian local-maximum selection.
\autoref{alg:local-maxima} always returns an independent set, since no two adjacent vertices can both be local maxima.

Fix a \emph{degree exponent} $\beta\in(0,1]$, and assume that every vertex of the input graph has degree at most $n^{\beta}$.
Besides $\ell$, \autoref{alg:sparse} takes three parameters: a deletion threshold $\delta\in(0,1)$, a target size $T$ and a sample count $R$.

\begin{algorithm}[H]
    \caption{\textsc{GaussianLocalMaxima}}
    \label{alg:local-maxima}
    \begin{algorithmic}[1]
        \Require vertex set $F$ and unit vectors $\{v_x\}_{x\in F}$
        \Ensure an independent set $I_F$
        \State Draw $g\sim N(0,I)$
        \State $I_F\gets\emptyset$
        \For{$x\in F$}
        \If{$(\ip{g}{v_x},x)>(\ip{g}{v_y},y)$ for every $y\in N_F(x)$}
        \State $I_F\gets I_F\cup\{x\}$
        \EndIf
        \EndFor
        \State \Return $I_F$
    \end{algorithmic}
\end{algorithm}

\autoref{alg:sparse} performs one trial.
First, set the initial core to be $F_0 = V$.
At level $j$, we run \autoref{alg:local-maxima} $R$ times on $F_{j-1}$ and compute, for each $x\in F_{j-1}$, the empirical estimate $\hat{p}_j(x)$ of the probability that $x$ is selected.
If any returned independent set has size at least the target size $T$, we return it with $\mathsf{SUCCESS}$.
Otherwise, we delete the vertices whose empirical estimate is at least $\delta/2$, and the remaining vertices form the next core $F_{j}$.
If all levels $j = 1, \ldots, \ell$ are exhausted, the trial exits with $\mathsf{FAIL}$ and is retried.

Below, we prove that if the estimates of the local-maximum probabilities are accurate, then \autoref{alg:sparse} returns a sufficiently large independent set.

\begin{algorithm}[H]
    \caption{\textsc{IterativeGaussianLocalMaxima}}
    \label{alg:sparse}
    \begin{algorithmic}[1]
        \Require graph $G=(V,E)$ with $n=|V|$, unit vectors $\{v_x\}_{x\in V}$, and parameters $\delta\in(0,1)$ and $T,R\in\mathbb{N}$
        \Ensure $\mathsf{SUCCESS}(I)$ with $I$ independent and $|I|\ge T$, or $\mathsf{FAIL}$
        \State $F_0\gets V$
        \For{$j=1$ to $\ell$}
        \For{$i=1$ to $R$}
        \State $I^{(i)}\gets$ \Call{GaussianLocalMaxima}{$F_{j-1},\{v_x\}_{x\in F_{j-1}}$}
        \If{$|I^{(i)}|\ge T$}
        \State \Return $\mathsf{SUCCESS}(I^{(i)})$
        \EndIf
        \EndFor
        \For{$x\in F_{j-1}$}
        \State $\hat{p}_j(x)\gets R^{-1} \sum_{i=1}^{R}\mathbf{1}[x\in I^{(i)}]$
        \EndFor
        \State $F_j\gets\{x\in F_{j-1}:\hat{p}_j(x)<\delta/2\}$
        \EndFor
        \State \Return $\mathsf{FAIL}$
    \end{algorithmic}
\end{algorithm}

The refinement step is designed so that, if no sufficiently large independent set is found, only a small number of vertices are removed at each level, while the surviving vertices have small local-maximum probabilities with respect to the preceding core.
These properties will be established in \autoref{sec:3.3} and form the basis of our Gaussian cover argument.
\subsection{Local-Maximum Probabilities}\label{sec:3.2}
This subsection collects the results about the local-maximum probabilities of \autoref{alg:local-maxima} that are shared by both the low- and high-degree regimes, and then settles the low-degree regime $0<\beta\le1/3$.

\subsubsection{A Lower Bound via Neighbor Gaussian Projections}\label{sec:3.2.1}
For a remaining subset $F \subseteq V$ and $x \in F$, define $p_F(x) = \Prb_g[x \in I_F(g)]$, the local-maximum probability of $x$.

For $F\subseteq V$ and $x\in V$, let $M_F(x)=\max_{y\in N_F(x)}\ip{g}{v_{xy}}$, where $g\sim N(0,I)$.
When $N_F(x)$ is empty, we define $M_F(x) = -\infty$.
We call $M_F(x)$ the \emph{neighborhood maximum} of $x$ in $F$.

\begin{lemma}
    \label{lem:mf-bound} Let $F\subseteq V$ and $x\in F$.
    Then, for every $a\in\R$,
    \begin{align*}
        p_{F}(x)
        &= \E_g[\bPhi(M_{F}(x)/\sqrt{3})],\\
        \Prb_g[M_{F}(x)\le a]
        &\le \frac{p_{F}(x)}{\bPhi(a/\sqrt{3})}.
    \end{align*}
\end{lemma}

\begin{proof}
    Write $g = Zv_x + g'$ with $Z \sim N(0, 1)$ and $g' \sim N(0, I - v_xv_x^{\mathsf{T}})$ independent.
    The vertex $x$ is selected as a local maximum exactly when $\ip{g}{v_x} > \ip{g}{v_y}$ for all $y \in N_{F}(x)$.
    Since $g' \perp v_x$ and $v_{xy} \perp v_x$, this condition is equivalent to $\sqrt{3} Z > \max_{y \in N_{F}(x)} \ip{g'}{v_{xy}} = M_{F}(x)$.
    Thus, the probability that $x$ is selected as a local maximum is $\E_g[\bPhi(M_{F}(x)/\sqrt{3})]$.
    On the event $M_{F}(x) \le a$, we have $\bPhi(M_{F}(x)/\sqrt{3})\ge\bPhi(a/\sqrt{3})$.
    Therefore,
    \begin{align*}
        p_{F}(x)
        &= \E_g[\bPhi(M_{F}(x)/\sqrt{3})]\\
        &\ge \E_g[\bPhi(a/\sqrt{3}) \mathbf{1}[M_{F}(x) \le a]]\\
        &=\bPhi(a/\sqrt{3})\Prb_g[M_{F}(x) \le a],
    \end{align*}
    which proves the second claim.
\end{proof}

\subsubsection{The Low-Degree Regime}\label{sec:3.2.2}
In this subsubsection, we treat the case $0<\beta\le1/3$.
Here \autoref{alg:sparse} already succeeds at the first level, and the argument uses only $p_F$, \autoref{lem:mf-bound} and the tail bounds of \autoref{fact:gaussian-tail}.

\begin{lemma}[Low-degree regime]
    \label{lem:low-degree} Fix $0<\beta\le1/3$, and let $G=(V,E)$ be a $3$-colorable graph on $n$ vertices with maximum degree at most $n^{\beta}$, given together with a strict vector $3$-coloring.
    Then there is a randomized polynomial-time algorithm that, with probability at least $\frac{3}{4}$, returns an independent set $I$ with $|I|\ge n^{1-\beta/3-o(1)}$.
\end{lemma}

\begin{proof}
    Run one trial of \autoref{alg:sparse} with the low-degree parameters of \autoref{tab:params-3}.
    We show that this trial runs in polynomial time and, with probability at least $\frac{3}{4}$, returns $\mathsf{SUCCESS}(I)$ for an independent set $I$ with $|I|\ge T\ge n^{1-\beta/3-o(1)}$.
    Fix $x\in V$, and recall that $M_V(x)=\max_{y\in N_V(x)}\ip{g}{v_{xy}}$.
    Each $v_{xy}$ is a unit vector, so $\ip{g}{v_{xy}}\sim N(0,1)$, and $x$ has at most $n^{\beta}$ neighbors.
    Hence a union bound and the upper bound of \autoref{fact:gaussian-tail} give $\Prb_g[M_V(x)>a]\le|N_V(x)|\,\bPhi(a)\le n^{\beta}e^{-a^2/2}$ for every $a\ge0$.
    Choosing $a=\sqrt{2\log{(2n^{\beta})}}$, we obtain $\Prb_g[M_V(x)\le a]\ge\frac{1}{2}$, and hence, by the second claim of \autoref{lem:mf-bound} with $F=V$, $p_V(x)\ge\Prb_g[M_V(x)\le a]\,\bPhi(a/\sqrt{3})\ge p_*$.
    By the lower bound of \autoref{fact:gaussian-tail},
    \begin{equation*}
        p_*
        \ge\frac{c_0}{2\left(1+\sqrt{\frac{2}{3}\log{(2n^{\beta})}}\right)}\exp{\left(-\frac{1}{3}\log{(2n^{\beta})}\right)}
        =\Omega\left(\frac{n^{-\beta/3}}{1+\sqrt{\log{(2n^{\beta})}}}\right)
        =n^{-\beta/3-o(1)}.
    \end{equation*}
    Hence $\E_g[|I_V(g)|]=\sum_{x\in V}p_V(x)\ge np_*$.
    Since $0\le|I_V(g)|\le n$, we have $np_*\le\E_g[|I_V(g)|]\le\frac{np_*}{2}+n\Prb_g[|I_V(g)|\ge np_*/2]$, and therefore $\Prb_g[|I_V(g)|\ge np_*/2]\ge p_*/2$.

    At the first level, \autoref{alg:sparse} runs \autoref{alg:local-maxima} $R$ times independently on $F_0=V$ and returns $\mathsf{SUCCESS}(I^{(i)})$ as soon as $|I^{(i)}|\ge T$.
    Hence the probability that no sample at the first level reaches $T$ is at most $(1-p_*/2)^{R}\le\exp{(-Rp_*/2)}\le\frac{1}{4}$.
    Thus, with probability at least $\frac{3}{4}$, the trial returns $\mathsf{SUCCESS}(I)$ with $|I|\ge T\ge np_*/2-1=n^{1-\beta/3-o(1)}$.
    Finally, one trial runs \autoref{alg:local-maxima} at most $\ell R$ times and $R=n^{\beta/3+o(1)}$, so it runs in polynomial time.
\end{proof}

\subsection{Failed Trials and Initial Gaussian Covers}\label{sec:3.3}
For the rest of this section, we treat the high-degree regime and assume $1/3<\beta\le1$.
In this subsection, we first prove two lemmas about a failed trial: (1) if \autoref{alg:sparse} fails, the deepest core $F_{\ell}$ is nonempty, and (2) with probability at least $\frac{3}{4}$, every surviving vertex $x \in F_{j}$ has small $p_{F_{j-1}}(x)$, as predicted by $\hat{p}_j(x)$.
We then convert the smallness of $p_{F_{j-1}}(x)$ into a Gaussian cover around a vertex of the deepest core.

\subsubsection{Failed Trials}\label{sec:3.3.1}
\begin{lemma}
    \label{lem:deleted-bound} If every sampled independent set $I^{(i)}$ has size less than $T$ at level $j$, then
    \begin{equation*}
        |F_{j-1} \setminus F_j|
        < \frac{2T}{\delta}.
    \end{equation*}
    In particular, if this holds at all levels $j=1,\ldots,\ell$, then
    \begin{equation*}
        \sum_{j=1}^{\ell} |F_{j-1} \setminus F_j|
        < \frac{2\ell T}{\delta}.
    \end{equation*}
\end{lemma}

\begin{proof}
    By the definition of $\hat{p}_j(x)$ and the assumption that every sampled independent set has size less than $T$, $\sum_{x \in F_{j-1}} \hat{p}_j(x) = R^{-1}\sum_{i = 1}^R|I^{(i)}| < T$.
    Since every deleted vertex satisfies $\hat{p}_j(x) \ge \delta / 2$, we obtain $|F_{j-1} \setminus F_j| < \tfrac{2T}{\delta}$.
    Summing over $j$ gives the second claim.
\end{proof}

\autoref{lem:deleted-bound} holds for arbitrary $\delta$ and $T$.
To keep the deepest core nonempty, we use their high-degree values in \autoref{tab:params-3}.

\begin{lemma}
    \label{lem:core-nonempty} If every sampled independent set has size less than $T$ at all levels $j=1,\ldots,\ell$, then $F_{\ell}$ is nonempty for sufficiently large $n$.
\end{lemma}

\begin{proof}
    Since $V = F_0 \supseteq F_1 \supseteq \cdots \supseteq F_{\ell}$, we have
    \begin{align*}
        |F_{\ell}|
        &= |F_0| - \sum_{j = 1}^{\ell}(|F_{j-1}| - |F_{j}|)\\
        &= |V| - \sum_{j = 1}^{\ell} |F_{j-1} \setminus F_{j}|\\
        &> n - \frac{2\ell T}{\delta},
    \end{align*}
    where the last inequality is \autoref{lem:deleted-bound}.
    Recall that $\delta^{-1}=n^{\beta/(2+3\beta)}\ell^{500f(\beta)}$, $\ell=\lceil\log{ n}\rceil$ and $T\le n\delta/(100\ell^4)+1$.
    Thus
    \begin{equation*}
        \frac{2\ell T}{\delta}
        \le \frac{n}{50\ell^3}+\frac{2\ell}{\delta}
        =O(n/\ell^3)+O(n^{\beta/(2+3\beta)}\ell^{500f(\beta)+1})
        =o(n).
    \end{equation*}
    Hence $|F_{\ell}|>n-o(n)\ge1$ for sufficiently large $n$, completing the proof.
\end{proof}

Next, the sample count $R$ of \autoref{tab:params-3} is chosen to make the empirical estimates accurate.

\begin{lemma}
    \label{lem:estimates} With probability at least $\frac{3}{4}$ over the samples used to construct the cores, $|\hat{p}_j(x) - p_{F_{j-1}}(x)| \le \delta/8$ for every level $j$ at which $F_j$ is constructed and every vertex $x\in F_{j-1}$; in particular, $p_{F_{j-1}}(x)<\delta$ for every such $j$ and every $x\in F_j$.
\end{lemma}

\begin{proof}
    By the definition of $F_j$, every $x \in F_j$ satisfies $\hat{p}_j(x) < \delta/2$.
    Fix any realization of all samples preceding level $j$.
    This fixes $F_{j-1}$, and the $R$ samples drawn at level $j$ are then independent.
    Hence, for each fixed $x \in F_{j-1}$, Hoeffding's inequality gives
    \begin{equation*}
        \Prb[|\hat{p}_j(x) - p_{F_{j-1}}(x)| > \delta / 8]
        < 2e^{-R\delta^2 / 32}
        \le \frac{1}{4n \ell}
    \end{equation*}
    by the choice of $R$.
    Since this bound holds for every realization of the past, a union bound over all levels and vertices gives
    \begin{equation*}
        \Prb[\exists j \in\{1,\ldots,\ell\}, \exists x \in F_{j-1} : |\hat{p}_j(x) - p_{F_{j-1}}(x)| > \delta / 8]
        < \frac{1}{4}.
    \end{equation*}
    Therefore, we have $p_{F_{j-1}}(x) \le \hat{p}_j(x) + \delta/8 < 5\delta / 8 < \delta$ for all levels $j$ and all vertices $x \in F_j$.
\end{proof}

\subsubsection{The Initial Cover}\label{sec:3.3.2}
In the Gaussian cover analysis, we regard the sampled cores as fixed and evaluate probabilities with respect to a fresh Gaussian vector $g$ independent of the samples used to construct them.
\begin{definition}[Gaussian cover]\label{def:cover}
    A finite family of vectors $X$ is called an \emph{$(h, q)$-cover} for a function $h:X\to\R$ if
    \begin{equation*}
        \Prb_g[\exists z \in X: \ip{g}{z} \ge h(z)]
        \ge q.
    \end{equation*}
    In particular, for a constant function $h\equiv s$, this condition reads
    \begin{equation*}
        \Prb_g[\max_{z \in X} \ip{g}{z} \ge s]
        \ge q,
    \end{equation*}
    and we simply call $X$ an \emph{$(s, q)$-cover} with \emph{cover threshold} $s$ and \emph{covering probability} $q$.
\end{definition}

\begin{lemma}
    \label{lem:initial-cover-threshold} Let $F\subseteq V$ and $r\in F$, and suppose that $p_F(r)\le e^{-t^2/2}$.
    Then, for every $a\in\R$,
    \begin{equation*}
        \Prb_g\left[M_F(r)>a\right]
        \ge 1-\frac{e^{-t^2/2}}{\bPhi(a/\sqrt{3})}.
    \end{equation*}
\end{lemma}

\begin{proof}
    Applying \autoref{lem:mf-bound} with $x=r$ gives
    \begin{equation*}
        \Prb_g[M_{F}(r)\le a]
        \le \frac{e^{-t^2/2}}{\bPhi(a/\sqrt{3})}.
    \end{equation*}
    Therefore, with probability at least $1-e^{-t^2/2}/\bPhi(a/\sqrt{3})$, we have $M_F(r)>a$.
\end{proof}

Fix an arbitrary vertex $r \in F_{\ell}$ of the deepest core, call it the \emph{root}, and use $\rho_x$ and $u_x$ of \autoref{sec:2.2} for this root.
Set $W_0 = N_{F_{\ell-1}}(r)$, and recall that $U_0 = 3/(2\sqrt{2})$.

\begin{corollary}
    \label{cor:initial-cover} Suppose that $p_{F_{\ell-1}}(r)\le e^{-t^2/2}$.
    Then $\{u_{y}\}_{y \in W_0}$ is a $(U_0t, 1 - O(te^{-t^2/4}))$-cover.
\end{corollary}

\begin{proof}
    Applying \autoref{lem:initial-cover-threshold} with $F=F_{\ell-1}$, $r\in F_{\ell}\subseteq F_{\ell-1}$, and $a=\sqrt{3/2}\,t$ gives
    \begin{equation*}
        \Prb_g[M_{F_{\ell-1}}(r)>\sqrt{3/2}\,t]
        \ge 1-\frac{e^{-t^2/2}}{\bPhi(t/\sqrt{2})}
        = 1-O(te^{-t^2/4}),
    \end{equation*}
    where the last step uses the lower bound of \autoref{fact:gaussian-tail} with $u=t/\sqrt{2}$.
    Therefore, with probability at least $1-O(te^{-t^2/4})$, we have $M_{F_{\ell-1}}(r)\ge\sqrt{3/2}\,t$.
    Hence $\{v_{ry}\}_{y\in W_0}$ is a $(\sqrt{3/2}\,t,1-O(te^{-t^2/4}))$-cover.
    Moreover, since each vertex $y \in W_0$ is adjacent to the root $r$, we have $u_y = \tfrac{\sqrt{3}}{2}v_{ry}$, so $\{u_y\}_{y\in W_0}$ is a $(\tfrac{\sqrt{3}}{2}\sqrt{3/2}\,t,1-O(te^{-t^2/4}))$-cover, and $\tfrac{\sqrt{3}}{2}\sqrt{3/2}=U_0$.
    This completes the proof.
\end{proof}

\subsection{Lifting to the Preceding Level}\label{sec:3.4}
The lifting argument involves two cardinality bounds: the total number of vertices $n$ and the maximum degree $n^{\beta}$.
In terms of the quantities $H_0$ and $D$ of \autoref{tab:params-3}, these bounds read $n=e^{H_{0}t^{2}/2}$ and $|N_{F}(x)|\le n^{\beta}=e^{Dt^{2}/2}$ for every $F\subseteq V$ and $x\in V$.
In this subsection, we prove the single-step lifting lemma (\autoref{lem:lifting}).
Given a $(Ut,q)$-cover $\{u_x\}_{x\in W}$ with $W$ a subset of the level-$j$ core $F_j$, we prove that $\{u_y\}_{y\in W'}$ is a $((U+\Delta U)t,q-\Delta q)$-cover for some $W'\subseteq F_{j-1}$ and appropriate $\Delta U$ and $\Delta q$.
More precisely, we can take $\Delta Ut = \Theta(wt) = \Theta(\frac{\log{t}} {t}) = \Theta(\frac{\log{\log{n}}} {\sqrt{\log{n}}})$ and $\Delta q = O(t^{-18}) = O((\log{n})^{-9})$.

The positive drift in $U$ is sufficiently large for $U$ to eventually exceed $\sqrt{H}$, while the accumulated loss in the covering probability remains negligible, leading to a contradiction with the Gaussian tail bound.

\begin{lemma}[Single-step lifting]
    \label{lem:lifting} Let $1\le j\le\ell$, and let $W\subseteq F_j$ be fixed independently of the Gaussian $g$, and suppose that $\{u_x\}_{x\in W}$ is a $(Ut,q)$-cover for some $0<U\le\sqrt{H}$.
    Assume that $p_{F_{j-1}}(x)\le e^{-t^2/2}$ for every $x\in W$, and set $W'=N_{F_{j-1}}(W)$.
    Then, for sufficiently large $t$, $\{u_y\}_{y\in W'}$ is a $((U+40w)t,q-t^{-18})$-cover.
\end{lemma}

For $x \in F_j$, we begin with an observation of the relation between the parent Gaussian projection $\ip{g}{u_x}$ and the child Gaussian projection $\ip{g}{u_y}$ for $y \in N_{F_{j-1}}(x)$.
We use the decomposition in \autoref{eq:vxy-decomposition} of \autoref{sec:2.2}.
We write $g_{\perp}$ for the orthogonal projection of $g$ onto $\Span(v_r,v_x)^{\perp}$, so that $\ip{g}{w_y}=\ip{g_{\perp}}{w_y}$.

By the first equality of \autoref{eq:uy-identity}, we have
\begin{equation*}
    \ip{g}{u_y}
    =-\left(\frac{1}{2}+\frac{\sqrt{3}}{2}\cdot\frac{\xi_y\rho_x}{\norm{u_x}}\right)\ip{g}{u_x}+\frac{\sqrt{3}}{2}\ip{g}{w_y}.
\end{equation*}

Since $t^2=\tfrac{2\beta}{2+3\beta}\log{n}+1000f(\beta)\log{\ell}$, the parameters $D$, $H_0$ and $H$ satisfy $D\to2+3\beta$ and $H_0,H\to(2+3\beta)/\beta$ as $n\to\infty$.
Since $\beta>1/3$, this gives $D>3$ and $H>1$ for sufficiently large $n$.

Define $\nu=\frac{1}{2}\sqrt{(D-3)/D}$, and let
\begin{equation*}
    \psi(s)
    =
    \begin{cases}
    \sqrt{3}\sqrt{1-s^2}, & s\le\nu,\\[2mm]
    \displaystyle\frac{\sqrt{3}(1-\nu s)}{\sqrt{1-\nu^2}}, & s\ge\nu.
    \end{cases}
\end{equation*}
Thus, $\psi$ agrees with the curve $s\mapsto\sqrt{3}\sqrt{1-s^2}$ on $[0,\nu]$ and with its tangent line at $s=\nu$ on $[\nu,\infty)$.

Here, assuming \autoref{lem:wy-transverse-bound} and \autoref{lem:U-increase} below, we prove \autoref{lem:lifting}.
The first lemma provides the probabilistic guarantee for transmission to a neighboring vertex, while the second ensures a uniform positive increase in the cover threshold.
\begin{lemma}
    \label{lem:wy-transverse-bound} Let $1\le j\le\ell$, and let $x\in F_j$ satisfy $\norm{u_x}>0$ and $p_{F_{j-1}}(x)\le e^{-t^2/2}$.
    Let $0<L<\sqrt{H}$ and $\rho\in[-1,1]$.
    Then, for sufficiently large $t$, over the Gaussian coordinates orthogonal to $\Span(v_r,v_x)$,
    \begin{align*}
        &\Prb_g\left[\exists y\in N_{F_{j-1}}(x): \ip{g}{w_y}+L\rho\xi_y t \ge \left(\psi\left(\sqrt{H-L^2(1-\rho^2)}\right)-40(1+20f(\beta))w\right)t\right]\\
        &\qquad\ge \Phi\left(\sqrt{H-L^2}\,t\right).
    \end{align*}
\end{lemma}

\begin{lemma}
    \label{lem:U-increase} For sufficiently large $n$,
    \begin{equation*}
        \frac{\sqrt{3}}{2}\psi\left(\sqrt{H-U^2}\right)
        \ge \frac{U}{2}+800f(\beta)w
    \end{equation*}
    for every $0\le U\le\sqrt{H}$.
\end{lemma}

\begin{proof}[Proof of \autoref{lem:lifting}]
    By Gaussian symmetry, the $(Ut,q)$-cover assumption also gives
    \begin{equation*}
        \Prb_g\left[\exists x\in W:\ip{g}{u_x}\le-Ut\right]
        \ge q.
    \end{equation*}
    Fix $x\in W$ and let $\cE_x$ denote the event $\ip{g}{u_x}\le-Ut$.
    If $\norm{u_x}=0$, then $\cE_x$ cannot occur, so assume $\norm{u_x}>0$.
    Set $Z_x=-\ip{g}{u_x/\norm{u_x}}$, which is a standard Gaussian, and let
    \begin{equation*}
        \cE^{\mathrm{bad}}_x
        = \cE_x\cap \left\{\forall y\in N_{F_{j-1}}(x):\ip{g}{u_y} <(U+40w)t\right\}.
    \end{equation*}
    On $\cE_x$, we have $Z_x\ge Ut/\norm{u_x}$.
    Hence $\Prb_g[\cE^{\mathrm{bad}}_x\mid Z_x=z]=0$ for $z<Ut/\norm{u_x}$.

    Now fix $z$ with $Ut/\norm{u_x}\le z<\sqrt{H}\,t$ and condition on $Z_x=z$.
    Since $u_x/\norm{u_x}\in\Span(v_r,v_x)$, the Gaussian coordinates orthogonal to $\Span(v_r,v_x)$ are unaffected by this conditioning.

    Since $\norm{u_x}\le1$, we have $U\le\norm{u_x}z/t<\sqrt{H}$.

    Applying \autoref{lem:wy-transverse-bound} with $L=z/t$ and $\rho=\rho_x$, where $L\sqrt{1-\rho^2}=\norm{u_x}z/t$, gives
    \begin{align*}
        &\Prb_g\left[\exists y\in N_{F_{j-1}}(x): \ip{g}{w_y}+\rho\xi_y z \ge \left(\psi\left(\sqrt{ H-\frac{\norm{u_x}^2z^2}{t^2} }\right)-40(1+20f(\beta))w\right)t \ \middle|\ Z_x=z\right]\\
        &\qquad\ge \Phi\left(\sqrt{Ht^2-z^2}\right).
    \end{align*}

    When the event above happens,
    \begin{align*}
        \ip{g}{u_y}
        &= \frac{\norm{u_x}z}{2} +\frac{\sqrt{3}}{2} \left(\rho\xi_y z+\ip{g}{w_y}\right)\\
        &\ge \frac{Ut}{2} + \frac{\sqrt{3}}{2} \psi\left(\sqrt{ H-\frac{\norm{u_x}^2z^2}{t^2} }\right)t -20\sqrt{3}(1+20f(\beta))wt\\
        &\ge \left(U+800f(\beta)w-20\sqrt{3}(1+20f(\beta))w\right)t\\
        &\ge(U+40w)t,
    \end{align*}
    where the second inequality is \autoref{lem:U-increase} applied at $\norm{u_x}z/t$, and the last one follows from $800f(\beta)-20\sqrt{3}(1+20f(\beta))=(800-400\sqrt{3})f(\beta)-20\sqrt{3}\ge40$, which holds since $f(\beta)\ge1/\sqrt{2}$.

    Therefore
    \begin{align*}
        &\Prb_g[\cE^{\mathrm{bad}}_x\mid Z_x=z]\\
        &\qquad\le \Prb_g\left[\forall y\in N_{F_{j-1}}(x): \ip{g}{w_y}+\rho\xi_y z < \left(\psi\left(\sqrt{ H-\frac{\norm{u_x}^2z^2}{t^2} }\right)-40(1+20f(\beta))w\right)t \ \middle|\ Z_x=z\right]\\
        &\qquad\le \bPhi\left(\sqrt{Ht^2-z^2}\right).
    \end{align*}
    Integrating over $z$ and bounding $\Prb_g[Z_x\ge\sqrt{H}\,t]$ by $\bPhi(\sqrt{H}\,t)$,
    \begin{align*}
        \Prb_g[\cE^{\mathrm{bad}}_x]
        &\le \bPhi(\sqrt{H}\,t) + \int_0^{\sqrt{H}\,t} \phi(z) \bPhi\left(\sqrt{Ht^2-z^2}\right) \,\mathrm{d}z\\
        &\le e^{-Ht^2/2} + \int_0^{\sqrt{H}\,t} \frac{1}{\sqrt{2\pi}} e^{-z^2/2} e^{-(Ht^2-z^2)/2} \,\mathrm{d}z\\
        &= \left(1+\frac{\sqrt{H}\,t}{\sqrt{2\pi}}\right) e^{-Ht^2/2},
    \end{align*}
    where we used the Gaussian tail bound $\bPhi(u)\le e^{-u^2/2}$ of \autoref{fact:gaussian-tail}.

    Taking a union bound over $x\in W$ and using $|W|\le n=e^{H_0t^2/2}$,
    \begin{equation*}
        \Prb_g\left[\exists x\in W:\cE^{\mathrm{bad}}_x\right]
        \le \left(1+\frac{\sqrt{H}\,t}{\sqrt{2\pi}}\right) e^{-(H-H_0)t^2/2}
        = \left(1+\frac{\sqrt{H}\,t}{\sqrt{2\pi}}\right) t^{-20}
        \le t^{-18}
    \end{equation*}
    for sufficiently large $t$, where we used $H-H_0=40w$, $wt^2=\log{t}$, and $H=O(1)$.

    Let
    \begin{equation*}
        \cE
        =\bigcup_{x\in W}\cE_x, \qquad \cE'
        = \left\{\max_{y\in W'}\ip{g}{u_y} \ge(U+40w)t\right\}.
    \end{equation*}
    Then
    \begin{equation*}
        \cE\setminus\cE'
        \subseteq \bigcup_{x\in W}\cE^{\mathrm{bad}}_x.
    \end{equation*}
    Therefore
    \begin{equation*}
        \Prb_g[\cE']
        \ge \Prb_g[\cE] - \Prb_g\left[\bigcup_{x\in W}\cE^{\mathrm{bad}}_x\right]
        \ge q-t^{-18},
    \end{equation*}
    which shows that $\{u_y\}_{y\in W'}$ is a $((U+40w)t,q-t^{-18})$-cover.
\end{proof}

First, we prove \autoref{lem:U-increase}.
\begin{proof}[Proof of \autoref{lem:U-increase}]
    Set
    \begin{equation*}
        \zeta
        = \sqrt{\frac{3D}{D+1}}-\sqrt{\frac{H(D-2)}{D+1}}.
    \end{equation*}
    We first show that
    \begin{equation*}
        \frac{\sqrt{3}}{2}\psi\left(\sqrt{H-U^2}\right) -\frac{U}{2}
        \ge\zeta
    \end{equation*}
    for every $0\le U\le\sqrt{H}$.
    Put $s=\sqrt{H-U^2}$, so that $U^2+s^2=H$.

    First suppose that $s\le\nu$.
    Then
    \begin{equation*}
        \frac{\sqrt{3}}{2}\psi(s)- \frac{U}{2}
        = \frac{3}{2}\sqrt{1-s^2} - \frac{1}{2}\sqrt{H-s^2}.
    \end{equation*}
    The right-hand side is decreasing in $s\in[0,\nu]$, since $H>1$ gives
    \begin{equation*}
        -\frac{3s}{2\sqrt{1-s^2}} + \frac{s}{2\sqrt{H-s^2}}
        \le0.
    \end{equation*}
    Hence
    \begin{equation*}
        \frac{\sqrt{3}}{2}\psi(s)-\frac{U}{2}
        \ge\frac{\sqrt{3}}{2}\psi(\nu)-\frac{\sqrt{H-\nu^2}}{2}.
    \end{equation*}
    The right-hand side is the value of $\frac{\sqrt{3}}{2}\psi(s)-\frac{U}{2}$ at $s=\nu$.
    Since the two branches of $\psi$ agree at $s=\nu$, this point is also covered by the second case below.
    Therefore, the first case follows from the second.

    Next, suppose that $s\ge\nu$.
    Then
    \begin{equation*}
        \frac{\sqrt{3}}{2}\psi(s)-\frac{U}{2}
        = \frac{3}{2\sqrt{1-\nu^2}} - \left(\frac{U}{2}+\frac{3\nu}{2\sqrt{1-\nu^2}}s\right),
    \end{equation*}
    so by Cauchy--Schwarz and $U^2+s^2=H$,
    \begin{equation*}
        \frac{\sqrt{3}}{2}\psi(s)-\frac{U}{2}
        \ge \frac{3}{2\sqrt{1-\nu^2}} - \sqrt{ H\left(\frac{1}{4}+\frac{9\nu^2}{4(1-\nu^2)}\right) }
        = \zeta,
    \end{equation*}
    where we used $\nu^2=\tfrac{D-3}{4D}$, hence $1-\nu^2=\tfrac{3(D+1)}{4D}$, to get
    \begin{equation*}
        \frac{3}{2\sqrt{1-\nu^2}}
        =\sqrt{\frac{3D}{D+1}}, \qquad \frac{1}{4}+\frac{9\nu^2}{4(1-\nu^2)}
        = \frac{D-2}{D+1}.
    \end{equation*}

    Thus, $\frac{\sqrt{3}}{2}\psi\left(\sqrt{H-U^2}\right)-\frac{U}{2}\ge\zeta$ holds.

    It remains to lower-bound $\zeta$.
    By the definitions of $D,H_0,t$ and $t^2=\tfrac{2\beta}{2+3\beta}\log{n}+1000f(\beta)\log{\ell}$,
    \begin{equation*}
        3D-H_0(D-2)
        = \frac{2000f(\beta)(2+3\beta)\log{n}\log{\ell}}{t^4}.
    \end{equation*}
    Hence
    \begin{equation*}
        \frac{3D-H_0(D-2)}{(D+1)w}
        = \frac{2000f(\beta)(2+3\beta)\log{n}\log{\ell}} {(2\beta\log{n}+t^2)\log{t}}
        \to \frac{2000f(\beta)(2+3\beta)^2}{3\beta(1+\beta)},
    \end{equation*}
    where we used $\tfrac{t^2}{\log{n}}\to\tfrac{2\beta}{2+3\beta}$ and $\tfrac{\log{\ell}}{\log{t}}\to2$.

    Since $H=H_0+40w$ and $\tfrac{D-2}{D+1}\to\tfrac{\beta}{1+\beta}$,
    \begin{equation*}
        \frac{3D-H(D-2)}{(D+1)w}
        \to \frac{2000f(\beta)(2+3\beta)^2}{3\beta(1+\beta)}-\frac{40\beta}{1+\beta}.
    \end{equation*}
    Therefore, rationalizing $\zeta$,
    \begin{equation*}
        \frac{\zeta}{w}
        = \frac{\tfrac{3D-H(D-2)}{(D+1)w}}{\sqrt{\tfrac{3D}{D+1}}+\sqrt{\tfrac{H(D-2)}{D+1}}}
        \to \frac{1}{2}\sqrt{\frac{1+\beta}{2+3\beta}}\left(\frac{2000f(\beta)(2+3\beta)^2}{3\beta(1+\beta)}-\frac{40\beta}{1+\beta}\right)
        >800f(\beta),
    \end{equation*}
    where we used $\tfrac{3D}{D+1}\to\tfrac{2+3\beta}{1+\beta}$ and $\tfrac{H(D-2)}{D+1}\to\tfrac{2+3\beta}{1+\beta}$.
    For the last inequality, note that for $\beta\in(1/3,1]$ we have $\tfrac{2+3\beta}{\beta}\ge5$, $2+3\beta\ge2$ and $1+\beta\le2$, so
    \begin{equation*}
        \frac{1}{2}\sqrt{\frac{1+\beta}{2+3\beta}}\cdot\frac{2000f(\beta)(2+3\beta)^2}{3\beta(1+\beta)}
        =\frac{1000f(\beta)}{3}\cdot\frac{2+3\beta}{\beta}\cdot\sqrt{\frac{2+3\beta}{1+\beta}}
        \ge\frac{5000f(\beta)}{3},
    \end{equation*}
    while $\frac{1}{2}\sqrt{\tfrac{1+\beta}{2+3\beta}}\cdot\tfrac{40\beta}{1+\beta}\le20$.
    Since $f(\beta)\ge1/\sqrt{2}$, the limit is therefore at least
    \begin{equation*}
        \frac{5000f(\beta)}{3}-20
        =800f(\beta)+\frac{2600f(\beta)}{3}-20
        >800f(\beta).
    \end{equation*}
    Thus $\zeta\ge800f(\beta)w$ for sufficiently large $n$.
    Hence
    \begin{equation*}
        \frac{\sqrt{3}}{2}\psi\left(\sqrt{H-U^2}\right)
        \ge \frac{U}{2}+800f(\beta)w,
    \end{equation*}
    which proves the claim.
\end{proof}

The remaining part of this subsection is devoted to proving \autoref{lem:wy-transverse-bound}.

\subsubsection{Ehrhard Extrapolation}\label{sec:3.4.1}
To prove \autoref{lem:wy-transverse-bound}, we use Ehrhard's inequality (\autoref{thm:ehrhard}).
Given a convex set $K_0$ of small Gaussian measure and a convex set $K_1$ of large Gaussian measure, it bounds the Gaussian measure of the extrapolated set $K_{-\theta}$, which is even smaller (\autoref{lem:conv-bound}).
This allows us to show that the transverse Gaussian projection of at least one neighbor of the parent is sufficiently large.
\begin{theorem}[Ehrhard's inequality~\cite{Ehr83}]
    \label{thm:ehrhard} If $K,K'$ are convex subsets of a Euclidean space and $0\le\lambda\le1$, then
    \begin{equation*}
        \Phi^{-1}\gmeasure((1-\lambda)K+\lambda K')
        \ge(1-\lambda)\Phi^{-1}\gmeasure(K)+\lambda\Phi^{-1}\gmeasure(K').
    \end{equation*}
    Here $(1-\lambda)K+\lambda K'=\{(1-\lambda)z+\lambda z':z\in K,\ z'\in K'\}$.
    Endpoint and zero-measure cases are interpreted by limits.
\end{theorem}

\begin{lemma}
    \label{lem:conv-bound} For a finite set $Y$ and its corresponding vectors $\{z_y\}_{y \in Y}$ and real thresholds $\{A_y\}_{y \in Y},\{B_y\}_{y \in Y}$, let
    \begin{align*}
        K_0 &=\{g:\ip{g}{z_y} \le A_y \ \forall y \in Y\},\\
        K_1 &=\{g:\ip{g}{z_y} \le B_y \ \forall y \in Y\},\\
        K_{-\theta} &=\{g:\ip{g}{z_y} \le(1+\theta)A_y -\theta B_y \ \forall y \in Y\}.
    \end{align*}
    If $\gmeasure(K_0)\le\bPhi(\mu_0)$ and $\gmeasure(K_1) \ge \Phi(\kappa_0)$, where $\mu_0,\kappa_0 \ge 0$, then for every $\theta\ge0$,
    \begin{equation*}
        \gmeasure(K_{-\theta})
        \le \bPhi(\mu_0+\theta(\mu_0+\kappa_0)).
    \end{equation*}
\end{lemma}
\begin{proof}
    By definition, $\tfrac{1}{1 + \theta}K_{-\theta} + \tfrac{\theta}{1 + \theta}K_1 \subseteq K_0$.
    Applying \autoref{thm:ehrhard} with $\lambda = \tfrac{\theta}{1+\theta}$, $K = K_{-\theta}$ and $K' = K_1$, we have
    \begin{align*}
        \frac{1}{1 + \theta}\Phi^{-1}\gmeasure(K_{-\theta}) + \frac{\theta}{1 + \theta}\Phi^{-1}\gmeasure(K_1)
        &\le \Phi^{-1}\gmeasure\left(\frac{1}{1 + \theta}K_{-\theta} + \frac{\theta}{1 + \theta}K_1\right)\\
        &\le \Phi^{-1}\gmeasure(K_0).
        &&(\because\Phi^{-1}\text{ is monotone})
    \end{align*}
    By rearranging,
    \begin{align*}
        \Phi^{-1}\gmeasure(K_{-\theta})
        &\le(1 + \theta)\Phi^{-1}\gmeasure(K_0) - \theta\Phi^{-1}\gmeasure(K_1)\\
        &\le -(1 + \theta)\mu_0 - \theta \kappa_0.
        &&(\because\gmeasure(K_0)\le\bPhi(\mu_0)=\Phi(-\mu_0),\ \gmeasure(K_1)\ge\Phi(\kappa_0))
    \end{align*}
    The monotonicity of $\Phi$ gives $\gmeasure(K_{-\theta}) \le \Phi(-\mu_0-\theta(\mu_0+\kappa_0))=\bPhi(\mu_0+\theta(\mu_0+\kappa_0))$.
\end{proof}

\begin{lemma}
    \label{lem:source-body} Let $1\le j\le\ell$, $x\in F_j$ with $\norm{u_x}>0$ and $a,b\in\R$, and let
    \begin{equation*}
        K_0(a,b)
        = \left\{g:\ip{g_{\perp}}{w_y}\le(\sqrt{3}a-b\xi_y)t-\frac{1}{t} \ \forall y\in N_{F_{j-1}}(x)\right\}.
    \end{equation*}
    Then
    \begin{equation*}
        \gmeasure(K_0(a,b))
        \le \frac{p_{F_{j-1}}(x)}{\Prb[Z_1\in[at,at+\tfrac{1}{t}]]\Prb[Z_2\in[bt,bt+\tfrac{1}{t}]]},
    \end{equation*}
    where $Z_1,Z_2\sim N(0,1)$.
    In particular, if $a\ge0$ and $p_{F_{j-1}}(x)\le e^{-t^2/2}$, then
    \begin{equation*}
        \gmeasure(K_0(a,b))
        \le 2\pi t^2 e^{-(1-a^2-b^2)t^2/2+a+|b|+t^{-2}}.
    \end{equation*}
\end{lemma}

\begin{proof}
    For the standard Gaussian vector $g \sim N(0, I)$, write $g = Z_1v_x + Z_2v_{xr} + g_{\perp}$.
    Here, $Z_1,Z_2\sim N(0,1)$, and $Z_1,Z_2,g_{\perp}$ are mutually independent.
    The condition that $x$ is selected as a local maximum is equivalent to
    \begin{align*}
        \ip{g}{v_x} > \ip{g}{v_y} \ (\forall y \in N_{F_{j-1}}(x))
        &\iff \ip{g}{v_x} > \ip{g}{-\frac{1}{2}v_x + \frac{\sqrt{3}}{2}v_{xy}} \ (\forall y \in N_{F_{j-1}}(x))\\
        &\iff \ip{g}{v_x} > \ip{g}{-\frac{1}{2}v_x + \frac{\sqrt{3}}{2}\left(\xi_yv_{xr} + w_y\right)} \ (\forall y \in N_{F_{j-1}}(x))\\
        &\iff Z_1 > -\frac{1}{2}Z_1 + \frac{\sqrt{3}}{2}\left(\xi_yZ_2 + \ip{g}{w_y}\right) \ (\forall y \in N_{F_{j-1}}(x))\\
        &\iff \sqrt{3} Z_1 - \xi_yZ_2 >\ip{g_{\perp}}{w_{y}} \ (\forall y \in N_{F_{j-1}}(x)).
    \end{align*}
    We focus on the event $(Z_1, Z_2) \in[at, at + \tfrac{1}{t}] \times[bt, bt + \tfrac{1}{t}]$.
    If the event happens and $g \in K_0(a,b)$, then
    \begin{equation*}
        \ip{g_{\perp}}{w_y}
        \le(\sqrt{3}a-b\xi_y)t-\frac{1}{t}
        \le\sqrt{3}Z_1 - \xi_yZ_2.
    \end{equation*}
    Since $K_0(a,b)$ depends only on $g_{\perp}$ and $Z_1, Z_2, g_{\perp}$ are all independent,
    \begin{align*}
        &\Prb[Z_1\in[at,at+\tfrac{1}{t}]]\Prb[Z_2\in[bt,bt+\tfrac{1}{t}]]\gmeasure(K_0(a,b))\\
        &\qquad= \Prb_g[(Z_1, Z_2) \in[at, at + \tfrac{1}{t}] \times[bt, bt + \tfrac{1}{t}] \text{ and }g \in K_0(a,b)]\\
        &\qquad\le \Prb_g[\sqrt{3} Z_1 - \xi_yZ_2>\ip{g_{\perp}}{w_{y}} \ (\forall y \in N_{F_{j-1}}(x))]\\
        &\qquad= \Prb_g[x\text{ is selected as a local maximum}]\\
        &\qquad\le p_{F_{j-1}}(x).
    \end{align*}
    Dividing by $\Prb[Z_1\in[at,at+\tfrac{1}{t}]]\Prb[Z_2\in[bt,bt+\tfrac{1}{t}]]>0$ gives the claim.
    If moreover $a\ge0$ and $p_{F_{j-1}}(x)\le e^{-t^2/2}$, then by the bound above,
    \begin{align*}
        \gmeasure(K_0(a,b))
        &\le \frac{e^{-t^2/2}}{\tfrac{1}{t^2}\phi(at+\tfrac{1}{t})\phi(|b|t+\tfrac{1}{t})}
        &&(\because p_{F_{j-1}}(x)\le e^{-t^2/2},\ \text{\autoref{fact:gaussian-interval}},\ a\ge0)\\
        &= 2\pi t^2 e^{-(1-a^2-b^2)t^2/2+a+|b|+t^{-2}}. && \qedhere
    \end{align*}
\end{proof}
The next lemma chooses the parameters so that Ehrhard extrapolation provides both the required probability bound and the target threshold $\psi$.
\begin{lemma}
    \label{lem:exact-parameters} Let $0<L<\sqrt{H}$, $\rho\in[-1,1]$ and $s = \sqrt{H-L^2(1-\rho^2)}$.
    For sufficiently large $n$, there exist $a_0 \ge \tfrac{\sqrt{3}}{2}$, $\kappa\ge0$, $0 \le \theta \le 20f(\beta)$, $|b|\le \frac{1}{2}$ and $0 \le \mu \le \frac{1}{2}$ such that
    \begin{align*}
        a_0^2 + b^2 + \mu^2
        &= 1,\\
        \mu+\theta(\mu+\kappa)
        &=\sqrt{H-L^2},
    \end{align*}
    and
    \begin{equation*}
        (1+\theta)\sqrt{3}a_0 -\sqrt{\bigl(L\rho-(1+\theta)b\bigr)^2+\theta^2(D+\kappa^2)}
        = \psi(s).
    \end{equation*}
\end{lemma}

\begin{proof}
    Since $s=\sqrt{H-L^2(1-\rho^2)}$, we have $L^2\rho^2+(H-L^2)=s^2$.

    First, suppose that $s\le\nu$.
    Choose
    \begin{align*}
        a_0&=\sqrt{1-s^2}, &b&=L\rho, &\mu&=\sqrt{H-L^2}, &\kappa&=\theta=0.
    \end{align*}
    By $L^2\rho^2+(H-L^2)=s^2$, we have $a_0^2+b^2+\mu^2=1$.
    Moreover,
    \begin{equation*}
        |b|,\mu
        \le s
        \le\nu
        <\frac{1}{2}, \qquad a_0
        \ge\sqrt{1-\nu^2}
        \ge\frac{\sqrt{3}}{2}.
    \end{equation*}
    Also, $\mu+\theta(\mu+\kappa)=\mu=\sqrt{H-L^2}$.

    Since $\theta=0$ and $b=L\rho$,
    \begin{equation*}
        (1+\theta)\sqrt{3}a_0 -\sqrt{\bigl(L\rho-(1+\theta)b\bigr)^2+\theta^2(D+\kappa^2)}
        = \sqrt{3}\sqrt{1-s^2}
        =\psi(s).
    \end{equation*}

    Now suppose that $s\ge\nu$.
    Choose
    \begin{align*}
        a_0&=\sqrt{1-\nu^2}, &b&=\frac{\nu L\rho}{s}, &\mu&=\frac{\nu\sqrt{H-L^2}}{s},\\
        \kappa&=D\mu, &\theta&=\frac{s/\nu-1}{D+1}.
    \end{align*}
    Then $a_0^2+b^2+\mu^2=1$ and $\mu+\theta(\mu+\kappa)=\sqrt{H-L^2}$.
    Since $\nu^2=\tfrac{D-3}{4D}<\tfrac{1}{4}$, we have $|b|,\mu\le\nu<\frac{1}{2}$ and $a_0\ge\tfrac{\sqrt{3}}{2}$.
    Since $s\le\sqrt{H}$, we also have
    \begin{equation*}
        0
        \le\theta
        \le\frac{\sqrt{H}/\nu-1}{D+1}
        \to\frac{1}{3(1+\beta)}\left(\frac{2(2+3\beta)}{\sqrt{\beta(3\beta-1)}}-1\right)
        <20f(\beta),
    \end{equation*}
    since $\sqrt{H}/\nu\to\tfrac{2(2+3\beta)}{\sqrt{\beta(3\beta-1)}}\le10\sqrt{3}\,f(\beta)<20f(\beta)$ (using $2+3\beta\le5$ and $\beta>1/3$) and $D+1\to3(1+\beta)>1$.
    Hence $\theta\le20f(\beta)$ for sufficiently large $n$.
    
    Since $L\rho=(1+(D+1)\theta)b$, $\kappa=D\mu$, and $b^2+\mu^2=\nu^2$,
    \begin{equation*}
        \sqrt{\bigl(L\rho-(1+\theta)b\bigr)^2+\theta^2(D+\kappa^2)}
        =\frac{\theta}{2}\sqrt{D(D+1)}.
    \end{equation*}
    Also, since $a_0=\sqrt{1-\nu^2}$,
    \begin{equation*}
        (1+\theta)\sqrt{3}a_0 -\sqrt{\bigl(L\rho-(1+\theta)b\bigr)^2+\theta^2(D+\kappa^2)}
        = \frac{1}{2}\sqrt{\frac{D+1}{D}} \bigl(3-(D-3)\theta\bigr).
    \end{equation*}
    By the definition of $\theta$ and $\nu^2=\tfrac{D-3}{4D}$,
    \begin{equation*}
        3-(D-3)\theta
        = \frac{4D}{D+1}(1-\nu s).
    \end{equation*}
    Moreover,
    \begin{equation*}
        \frac{1}{2}\sqrt{\frac{D+1}{D}} \frac{4D}{D+1}
        = \frac{\sqrt{3}}{\sqrt{1-\nu^2}}.
    \end{equation*}
    Hence
    \begin{equation*}
        (1+\theta)\sqrt{3}a_0 -\sqrt{\bigl(L\rho-(1+\theta)b\bigr)^2+\theta^2(D+\kappa^2)}
        = \frac{\sqrt{3}(1-\nu s)}{\sqrt{1-\nu^2}}
        = \psi(s).
    \end{equation*}

    This completes the proof.
\end{proof}

Finally, we prove \autoref{lem:wy-transverse-bound}.

\begin{proof}[Proof of \autoref{lem:wy-transverse-bound}]
    Set $s=\sqrt{H-L^2(1-\rho^2)}$.
    Let $a_0$, $b$, $\mu$, $\kappa$ and $\theta$ be as in \autoref{lem:exact-parameters}.

    They satisfy
    \begin{align*}
        a_0^2+b^2+\mu^2
        &=1,\\
        \mu+\theta(\mu+\kappa)
        &=\sqrt{H-L^2},
    \end{align*}
    and
    \begin{equation*}
        (1+\theta)\sqrt{3}a_0 -\sqrt{\bigl(L\rho-(1+\theta)b\bigr)^2+\theta^2(D+\kappa^2)}
        =\psi(s).
    \end{equation*}
    Moreover, $a_0\ge\tfrac{\sqrt{3}}{2}$, $|b|\le\frac{1}{2}$, $0\le \mu\le\frac{1}{2}$, and $0\le\theta\le20f(\beta)$.

    Set $a=a_0-20w$.
    Then $a\ge0$ for sufficiently large $t$, and $1-a^2-b^2=\mu^2+40a_0w-400w^2\ge \mu^2+30w$.
    Consequently,
    \begin{equation*}
        2\pi t^2 e^{-(1-a^2-b^2)t^2/2+a+|b|+t^{-2}}
        \le O(t^{-13})e^{-\mu^2t^2/2}
        \le\bPhi(\mu t).
    \end{equation*}
    The last inequality follows from the lower bound of \autoref{fact:gaussian-tail}, $\bPhi(\mu t)\ge c_0(\mu t+1)^{-1}e^{-\mu^2t^2/2}$, together with $\mu\le\frac{1}{2}$.

    Let
    \begin{equation*}
        K_0(a,b)
        = \left\{g:\ip{g_{\perp}}{w_y}\le(\sqrt{3}a-b\xi_y)t-\frac{1}{t} \ \forall y\in N_{F_{j-1}}(x)\right\}
    \end{equation*}
    be the set in \autoref{lem:source-body}.
    Then
    \begin{align*}
        \gmeasure(K_0(a,b))
        &\le 2\pi t^2 e^{-(1-a^2-b^2)t^2/2+a+|b|+t^{-2}}
        &&(\because\text{\autoref{lem:source-body}},\ a\ge0,\ p_{F_{j-1}}(x)\le e^{-t^2/2})\\
        &\le \bPhi(\mu t).
        &&(\because\text{the inequality above})
    \end{align*}

    Next, let
    \begin{equation*}
        K_1
        = \left\{g:\ip{g_{\perp}}{w_y}\le\sqrt{D+\kappa^2}\sqrt{1-(\xi_y)^2}\,t \ \forall y\in N_{F_{j-1}}(x)\right\}.
    \end{equation*}
    For a fixed $y\in N_{F_{j-1}}(x)$, since $\norm{w_y}=\sqrt{1-(\xi_y)^2}$, we have $\Prb_g[\ip{g}{w_y}>\sqrt{D+\kappa^2}\sqrt{1-(\xi_y)^2}\,t]\le\bPhi(\sqrt{D+\kappa^2}\,t)$.
    Since $|N_{F_{j-1}}(x)|\le n^{\beta}=e^{Dt^2/2}$, by a union bound, $1-\gmeasure(K_1) \le e^{Dt^2/2}\bPhi(\sqrt{D+\kappa^2}\,t) \le \bPhi(\kappa t)$, where the last inequality is \autoref{fact:tail-monotone} applied with $u=\kappa t$ and $v=\sqrt{D+\kappa^2}\,t$.
    Hence $\gmeasure(K_1)\ge\Phi(\kappa t)$.

    Since $\ip{g}{w_y}=\ip{g_{\perp}}{w_y}$, the sets $K_0$ and $K_1$ in \autoref{lem:conv-bound} for $Y=N_{F_{j-1}}(x)$, $z_y=w_y$, $A_y=(\sqrt{3}a-b\xi_y)t-\tfrac{1}{t}$ and $B_y=\sqrt{D+\kappa^2}\sqrt{1-(\xi_y)^2}\,t$ are $K_0(a,b)$ and $K_1$ above.
    Hence \autoref{lem:conv-bound} with $\mu_0=\mu t$ and $\kappa_0=\kappa t$ gives $\gmeasure(K_{-\theta})\le\bPhi((\mu+\theta(\mu+\kappa))t)$, and the complement of $K_{-\theta}$ is the event below.
    Together with $\mu+\theta(\mu+\kappa)=\sqrt{H-L^2}$,
    \begin{align*}
        &\Prb_g\left[\exists y \in N_{F_{j-1}}(x): \ip{g}{w_y} > (1+\theta)(\sqrt{3}a-b\xi_y)t -\theta\sqrt{D+\kappa^2}\sqrt{1-(\xi_y)^2}\,t -\frac{1+\theta}{t}\right]\\
        &\qquad\ge \Phi(\sqrt{H-L^2}\,t).
    \end{align*}

    When the event above happens,
    \begin{equation*}
        \ip{g}{w_y}+L\rho\xi_y t
        > \Bigl((1+\theta)\sqrt{3}a +\bigl(L\rho-(1+\theta)b\bigr)\xi_y -\theta\sqrt{D+\kappa^2}\sqrt{1-(\xi_y)^2} \Bigr)t -\frac{1+\theta}{t}.
    \end{equation*}
    Since $A\xi-B\sqrt{1-\xi^2}\ge-\sqrt{A^2+B^2}$ for $|\xi|\le1$ by Cauchy--Schwarz, we obtain
    \begin{align*}
        \ip{g}{w_y}+L\rho\xi_y t
        &\ge \left((1+\theta)\sqrt{3}a -\sqrt{\bigl(L\rho-(1+\theta)b\bigr)^2+\theta^2(D+\kappa^2)}\right)t -\frac{1+\theta}{t}\\
        &= \left((1+\theta)\sqrt{3}(a_0-20w) -\sqrt{\bigl(L\rho-(1+\theta)b\bigr)^2+\theta^2(D+\kappa^2)}\right)t -\frac{1+\theta}{t}\\
        &\ge \bigl(\psi(s) - 20\sqrt{3}(1+20f(\beta))w\bigr)t -\frac{1+20f(\beta)}{t}\\
        &\ge \bigl(\psi(s) - 40(1+20f(\beta))w\bigr)t,
    \end{align*}
    where the second inequality uses $\theta\le20f(\beta)$, and the last one uses $(40-20\sqrt{3})(1+20f(\beta))wt\ge(1+20f(\beta))/t$ for sufficiently large $t$.
    This completes the proof.
\end{proof}
\subsection{Iterated Lifting}\label{sec:3.5}
In this subsection, we repeatedly apply the single-step lifting lemma (\autoref{lem:lifting}), starting from the initial Gaussian cover around a vertex in the deepest core.
After $o(\ell)$ lifting steps, the cover threshold exceeds $\sqrt{H}\,t$, while the covering probability decreases by only $o(1)$.

\begin{lemma}
    \label{lem:sparse-progress-one} Suppose that every empirical estimate used in \autoref{alg:sparse} satisfies $|\hat{p}_j(x)-p_{F_{j-1}}(x)|\le\delta/8$ for every level $j$ at which $F_j$ is constructed and every $x\in F_{j-1}$.
    Then, for sufficiently large $n$, \autoref{alg:sparse} returns $\mathsf{SUCCESS}(I)$ for some independent set $I$ with $|I|\ge T$.
\end{lemma}

\begin{proof}
    Suppose for contradiction that \autoref{alg:sparse} returns $\mathsf{FAIL}$.
    Then every independent set sampled at every level has size less than $T$, all cores $F_1,\ldots,F_{\ell}$ are constructed, and $F_{\ell}$ is nonempty by \autoref{lem:core-nonempty}.
    Moreover, by the definition of $F_j$, every $x\in F_j$ satisfies $\hat{p}_j(x)<\delta/2$, and hence
    \begin{equation*}
        p_{F_{j-1}}(x)
        \le\hat{p}_j(x)+\frac{\delta}{8}
        < e^{-t^2/2} \qquad(\forall j,\,\forall x\in F_j).
    \end{equation*}

    Choose a root $r\in F_{\ell}$, and let $W_0$ and $U_0$ be as in \autoref{sec:3.3}.
    Since $p_{F_{\ell-1}}(r)<\delta=e^{-t^2/2}$, \autoref{cor:initial-cover} shows that $\{u_x\}_{x\in W_0}$ is a $(U_0t,1-O(te^{-t^2/4}))$-cover.
    Let $m$ be the smallest integer with $U_m>\sqrt{H}$.
    Since $H=O(1)$, we have $m=O(1/w)=O(t^2/\log{t})$.
    As $t^2=\Theta(\log{n})$ and $\ell=\Theta(\log{n})$, this gives $m=o(\ell)$, so $m\le\ell-1$ for sufficiently large $n$.
    For $0\le i<m$, define $W_{i+1}=N_{F_{\ell-i-2}}(W_i)$.
    By construction, $W_i\subseteq F_{\ell-i-1}$ and $W_i$ is determined by the samples used to construct the cores, independently of $g$.

    We claim that $\{u_x\}_{x\in W_i}$ is a $(U_it,1-O(te^{-t^2/4})-it^{-18})$-cover for every $0\le i\le m$.
    The case $i=0$ is \autoref{cor:initial-cover}, as noted above.
    For $i<m$, every $x\in W_i\subseteq F_{\ell-i-1}$ satisfies $p_{F_{\ell-i-2}}(x)\le e^{-t^2/2}$ by the bound above.
    Since $0<U_i\le\sqrt{H}$, \autoref{lem:lifting} applied to $W=W_i$ and $U=U_i$ shows that $\{u_x\}_{x\in W_{i+1}}$ is a $(U_{i+1}t,1-O(te^{-t^2/4})-(i+1)t^{-18})$-cover, completing the induction.

    Since $m=O(t^2/\log{t})$, we have $mt^{-18}=o(1)$, and also $te^{-t^2/4}=o(1)$.
    Thus, for sufficiently large $n$, $\{u_x\}_{x\in W_m}$ is a $(U_mt,\frac{1}{2})$-cover.

    On the other hand, $\norm{u_x}\le1$ for every $x$, so $\Prb_g[\ip{g}{u_x}\ge U_mt]\le\bPhi(U_mt)$, and $|W_m|\le n=e^{H_0t^2/2}$.
    Therefore, by a union bound and the Gaussian tail bound $\bPhi(u)\le e^{-u^2/2}$ of \autoref{fact:gaussian-tail},
    \begin{align*}
        \Prb_g\left[\max_{x\in W_m}\ip{g}{u_x}\ge U_mt\right]
        &\le n\bPhi(U_mt)\\
        &\le e^{H_0t^2/2}e^{-U_m^2t^2/2}\\
        &< e^{-(H-H_0)t^2/2}
        = t^{-20},
    \end{align*}
    where we used $U_m^2>H$ and $(H-H_0)t^2/2=20\log{t}$.
    For sufficiently large $n$, this is less than $\frac{1}{2}$, contradicting the $(U_mt,\frac{1}{2})$-cover property.
    Hence \autoref{alg:sparse} cannot return $\mathsf{FAIL}$ and must return $\mathsf{SUCCESS}(I)$ for an independent set $I$ with $|I|\ge T$.
\end{proof}

\begin{lemma}[High-degree regime]
    \label{lem:high-degree} Fix $1/3<\beta\le1$, and let $G=(V,E)$ be a $3$-colorable graph on $n$ vertices with maximum degree at most $n^{\beta}$, given together with a strict vector $3$-coloring.
    Then, there is a randomized polynomial-time algorithm that, with probability at least $\frac{3}{4}$, returns an independent set $I$ with $|I|\ge n^{\frac{2+2\beta}{2+3\beta}-o(1)}$.
\end{lemma}

\begin{proof}
    Run one trial of \autoref{alg:sparse} with the high-degree parameters of \autoref{tab:params-3}.
    Since $\beta$ is fixed and $\ell=O(\log{n})$, we have $\delta^{-1}=n^{\beta/(2+3\beta)}\polylog(n)$.
    Each call to \autoref{alg:local-maxima} takes polynomial time, and the number of calls is at most $\ell R$.
    Since $R=O(\delta^{-2}\log{n})=n^{2\beta/(2+3\beta)+o(1)}$, one trial runs in polynomial time.

    By \autoref{lem:estimates}, with probability at least $\frac{3}{4}$, every empirical estimate used in the construction of the cores satisfies $|\hat{p}_j(x)-p_{F_{j-1}}(x)|\le\delta/8$.
    Conditioned on this event, \autoref{lem:sparse-progress-one} shows that \autoref{alg:sparse} returns $\mathsf{SUCCESS}(I)$ for an independent set $I$ with
    \begin{equation*}
        |I| \geq T \geq \frac{n\delta}{100\ell^4} = n^{1-\frac{\beta}{(2+3\beta)}-o(1)} = n^{\frac{2+2\beta}{2+3\beta}-o(1)},
    \end{equation*}
    when $n$ is sufficiently large. The ``$-o(1)$'' part allows us to ignore the case of small $n$, and \autoref{lem:high-degree} follows.
\end{proof}

\sparsethm*

\begin{proof}[Proof of \autoref{thm:vector-3-coloring}]
    Fix $C>0$. Run $\lceil C\log_4{n}\rceil$ independent trials of \autoref{alg:sparse}, using the parameters from \autoref{lem:low-degree} if $0<\beta\le1/3$ and from \autoref{lem:high-degree} if $1/3<\beta\le1$.
    Return the largest independent set obtained.
    By \autoref{lem:low-degree} or \autoref{lem:high-degree}, according to the regime, each trial returns an independent set of the required size with probability at least $\frac{3}{4}$.
    Hence the probability that no trial returns such a set is at most $4^{-\lceil C\log_4{n}\rceil}\le n^{-C}$.
    Since each trial runs in polynomial time and there are $O(\log{n})$ trials, the total running time remains polynomial.
\end{proof}

\section{Improved Sparse Progress}\label{sec:4}
In this section, we prove \autoref{thm:improved-is} using the preceding vector coloring algorithm as a subroutine.
\autoref{tab:params-4} lists the parameters used in this section.

\begin{table}[!htbp]
    \centering
    \small
    \renewcommand{\arraystretch}{1.15}
    \setlength{\tabcolsep}{3pt}
    \caption{Parameters used in \autoref{sec:4}.}
    \label{tab:params-4}
    \begin{tabular}{lll}
        \hline
        Symbol & Value & Role \\
        \hline
        $\beta$ & $0.692601$ & exponent of the maximum degree: $\Delta(G)\le n^{\beta}$ \\
        $\ell$ & $\lceil\log{n}\rceil$ & number of refinement levels \\
        \multicolumn{3}{l}{\textit{Parameters of \autoref{alg:sparse}}} \\
        $\delta$ & $n^{-(1-\beta)/2}$ & deletion threshold \\
        $T$ & $\lceil n\delta/(100\ell^4)\rceil$ & target size of the independent set \\
        $R$ & $\lceil32\delta^{-2}\log{(8n\ell)}\rceil$ & number of runs of \autoref{alg:local-maxima} per level \\
        \multicolumn{3}{l}{\textit{Analysis of the lifting lemma}} \\
        $t$ & $\sqrt{2\log{(1/\delta)}}\ (=\sqrt{(1-\beta)\log{n}})$ & $\delta=e^{-t^{2}/2}$ \\
        $H_0$ & $2\log{n}/t^{2}\ (=2/(1-\beta))$ & $n=e^{H_0t^{2}/2}$ \\
        $D$ & $2\beta\log{n}/t^{2}\ (=2\beta/(1-\beta))$ & $n^{\beta}=e^{Dt^{2}/2}$ \\
        $w$ & $\log{t}/t^{2}$ & unit of the preparation steps \\
        $H$ & $H_0+3\eta\approx6.506205$ & $\sqrt{H}$ is the target Gaussian cover threshold divided by $t$ \\
        $U_i$ & $U_0+\frac{i}{3000}$ & Gaussian cover threshold divided by $t$ after $i$ lifting steps \\
        \multicolumn{3}{l}{\textit{Analysis of the tail bound on the neighborhood maximum}} \\
        $\eta$ & $10^{-6}$ & slack in the band-size bounds \\
        $D_{\mathrm{a}}$ & $D+3\eta\approx4.506205$ & $D$ enlarged by the slack $3\eta$ \\
        $A_i$ & $\max\{1.968,A_0-iw\}$ & tail-bound threshold divided by $t$ after $i$ preparation steps \\
        $m_0$ & $\lceil(A_0-1.968)/w\rceil$ & number of preparation steps \\
        \hline
    \end{tabular}%

\end{table}

As in \autoref{sec:3}, $U_0=3/(2\sqrt{2})$ is the initial Gaussian cover threshold divided by $t$.
We write $A_0=\sqrt{D_{\mathrm{a}}+\frac{1}{4}}$ for the initial tail-bound threshold divided by $t$.

\subsection{Preliminaries for \textsc{SparseProgress}}\label{sec:4.1}

We consider the following semidefinite programming relaxation: for $m\in\{1,2\}$ and $S\subseteq V$, let $\Theta_m(S)$ be the degree-2 relaxation of the problem of finding a largest induced subgraph of $G[S]$ that is $m$-colorable.
For each vertex $i\in S$ and color $c\in\{1,\ldots,m\}$ we have a vector $x_{i,c}$, together with an anchor vector $x_{\emptyset}$, and we write
\begin{equation*}
    x_i
    =\sum_{c=1}^m x_{i,c}.
\end{equation*}
Then
\begin{equation*}
    \Theta_m(S)
    =\max\sum_{i\in S}\sum_{c=1}^m\norm{x_{i,c}}^2
\end{equation*}
subject to
\begin{align*}
    \norm{x_{\emptyset}}^2
    &=1, \\
    \ip{x_{\emptyset}}{x_{i,c}}
    &=\norm{x_{i,c}}^2
    &&\text{for all } i\in S \text{ and } c\in\{1,\ldots,m\}, \\
    \ip{x_{i,c}}{x_{i,c'}}
    &=0
    &&\text{for all } i\in S \text{ and } c\ne c'\in\{1,\ldots,m\}, \\
    \ip{x_{i,c}}{x_{j,c}}
    &=0
    &&\text{for all } ij\in E[S] \text{ and } c\in\{1,\ldots,m\}, \\
    \ip{x_{\emptyset}-x_i}{x_{\emptyset}-x_j}
    &\ge0
    &&\text{for all } ij\in E[S].
\end{align*}

In the intended solution, for an induced subgraph of $G[S]$ with a proper $m$-coloring, we set $x_{i,c}=x_{\emptyset}$ if vertex $i$ belongs to the subgraph and receives color $c$, and $x_{i,c}=0$ otherwise.
For $m=1$, without the last family of constraints, this is the Lov\'asz theta function~\cite{Lov79}.

Due to the third equality, we have $\norm{x_i}^2 = \sum_{c=1}^m\norm{x_{i,c}}^2$.
\subsection{The Algorithm \textsc{SparseProgress}}\label{sec:4.2}
In this subsection, we describe our improved algorithm for finding a large independent set.
The algorithm uses a vector coloring SDP strengthened with additional constraints.
We begin by defining the correlation bands used in the algorithm.

\paragraph{Bands.}
Fix a partition $\cP$ of $[-\frac{1}{2},1]$ into $O(\log{n})$ intervals of length at most $1/\log{n}$, each contained in $[-\frac{1}{2},0]$ or in $[0,1]$.
For a vertex $r\in V$ and $P\in\cP$, let $S_{r,P}=\{x\in V:\ip{v_r}{v_x}\in P\}$. We call $P$, and the band $S_{r,P}$, \emph{positive} if $P\subseteq[0,1]$ and \emph{negative} otherwise; we set $\rho_P=\inf P$ for a positive band and $\rho_P=\sup P$ for a negative band.

The choice of auxiliary SDPs is motivated by the simplex representation of a proper $3$-coloring (\autoref{sec:2.1}).
In this representation, $\frac{1+2\ip{v_r}{v_x}}{3}$ is $1$ if $x$ has the same color as $r$ and $0$ otherwise, whereas $\frac{2-2\ip{v_r}{v_x}}{3}$ is $1$ if $x$ has one of the other two colors and $0$ otherwise.
The vertices of the first kind form an independent set, and those of the second kind induce a bipartite subgraph.
This motivates using $\Theta_1$ for positive bands and $\Theta_2$ for negative bands.

For each band, we use an auxiliary SDP witness either to extract an independent set directly or to construct an ordinary vector $k$-coloring, to which we apply the algorithm of \autoref{thm:vector-k-coloring}.
If the resulting independent set has size at least $T$, we return it.
Otherwise, provided that the internal vector coloring algorithm meets its guarantee, we obtain an upper bound on the size of the band.
These bounds are then used to strengthen the tail bound on the neighborhood maximum and the Gaussian cover lifting argument.

We use \textsc{BuildVector} to obtain two types of guarantees.
First, for each correlation band, it provides an auxiliary SDP witness whose total mass is sufficiently large.
These witnesses are used in the band-rounding procedures.
Second, the strengthened SDP imposes additional constraints that control correlations between Gaussian projections associated with adjacent vertices.
These constraints will be used in the tail analysis of the neighborhood maximum in \autoref{sec:4.4}.

\paragraph{Certified vectors and band witnesses.}
Let $\varepsilon_{\mathrm{SDP}}$ be the accuracy parameter of \autoref{sec:2.1}.
To motivate the additional semidefinite constraint, we first assume that the vector coloring constraints are exact.
For an edge $rx \in E$, consider the symmetric matrix
\begin{equation*}
    \Gamma'_{rx}
    =\begin{pmatrix}
    (\ip{v_{ry}}{v_{ry'}})_{y,y'} & (2\ip{v_{ry}}{v_{xz}})_{y,z}\\
    (2\ip{v_{xz}}{v_{ry}})_{z,y} & (\ip{v_{xz}}{v_{xz'}})_{z,z'}
    \end{pmatrix},
\end{equation*}
where $y, y' \in N_V(r)$ and $z, z' \in N_V(x)$.

Let $X\in\R^{V\times V}$ be a symmetric matrix with unit diagonal and define
\begin{equation*}
    (\Gamma_{rx}(X))_{(i,j),(i',j')}
    =\frac{1+\mathbf{1}_{i\ne i'}}{3}(4X_{jj'}+2X_{ji'}+2X_{ij'}+X_{ii'})
\end{equation*}
for $i,i'\in\{r,x\}$, $j\in N_V(i)$, and $j'\in N_V(i')$.
If $X$ is the Gram matrix of $\{v_i\}_{i\in V}$, then $\Gamma_{rx}(X)=\Gamma'_{rx}$, since $v_{ij}=(2v_j+v_i)/\sqrt{3}$ for every $ij\in E$.
In the exact setting, positive semidefiniteness of this matrix bounds the correlation between two spans of edge residuals by $\frac{1}{2}$, as shown in \autoref{lem:edge-neighborhood}.

We use \textsc{BuildVector} through the following output guarantee.
\begin{lemma}[Certified vectors and band witnesses]
    \label{lem:build-vector} There is a subroutine \textsc{BuildVector} that, given a $3$-colorable graph $G=(V,E)$ on $n\ge 2$ vertices, returns $(\{v_i\}_{i\in V},\cS = \{(\rho_P, S_{r,P}, \cW_{r,P})\}_{r \in V, P \in \cP})$ in polynomial time in $n$ and $\log{(1/\varepsilon_{\mathrm{SDP}})}$ and with the following properties.

    The vectors $\{v_i\}_{i\in V}$ are unit vectors satisfying
    \begin{align*}
        -\frac{1}{2}
        &\le \ip{v_i}{v_j}
        \le 1
        &&\text{for all } i,j\in V,\\
        \ip{v_i}{v_j}
        &\le -\frac{1}{2} + \varepsilon_{\mathrm{SDP}}
        &&\text{for all } ij\in E,\\
        \ip{v_i}{v_j}+\ip{v_i}{v_{j'}}
        &\le \frac{1}{2}
        &&\text{for all } jj'\in E \text{ and } i\in V\setminus\{j,j'\},\\
        \Gamma_{rx}(X)
        &\succeq -\varepsilon_{\mathrm{SDP}}I
        &&\text{for all } rx\in E,
    \end{align*}
    where $X_{ij}=\ip{v_i}{v_j}$.

    For every $r\in V$ and $P\in\cP$, writing $S=S_{r,P}$, the triple $(\rho_P,S_{r,P},\cW_{r,P})$ has the following properties.
    \begin{enumerate}
        \item[(i)] If $P$ is positive, then $\cW_{r,P}=(x_{\emptyset},\{x_j\}_{j\in S})$ is feasible for $\Theta_1(S)$ and
            \begin{equation*}
                \sum_{j\in S}\norm{x_j}^2
                \ge\sum_{j\in S}\frac{1+2\ip{v_r}{v_j}}{3}-\varepsilon_{\mathrm{SDP}}
                \ge\frac{1+2\rho_P}{3}|S|-\varepsilon_{\mathrm{SDP}}.
            \end{equation*}

        \item[(ii)] If $P$ is negative, then $\cW_{r,P}=(x_{\emptyset},\{(x_{j,1},x_{j,2})\}_{j\in S})$ is feasible for $\Theta_2(S)$ and
            \begin{equation*}
                \sum_{j\in S}\left(\norm{x_{j,1}}^2+\norm{x_{j,2}}^2\right)
                \ge\sum_{j\in S}\frac{2-2\ip{v_r}{v_j}}{3}-\varepsilon_{\mathrm{SDP}}
                \ge\frac{2-2\rho_P}{3}|S|-\varepsilon_{\mathrm{SDP}}.
            \end{equation*}
    \end{enumerate}
\end{lemma}

The implementation of \textsc{BuildVector} and the proof of \autoref{lem:build-vector} are given in \autoref{sec:B}.
We use the algorithm of \autoref{thm:vector-k-coloring} as a subroutine on the vector colorings obtained from the bands.
We set a rounding slack $\varepsilon_{\mathrm{thr}}\in(0,1/30]$ for \autoref{alg:round-pos} and \autoref{alg:round-neg}, and fix a grid step $h\in(0,1/8]$ with $1/h\in\mathbb{N}$ for \textsc{VectorColoring}.
These parameters are determined so that they are sufficiently small depending only on $\beta$ and $\eta$ of \autoref{tab:params-4}, and not on $n$.
In particular, they are chosen so that the term $O(h+\varepsilon_{\mathrm{thr}})$ in the proof of \autoref{lem:band-bound}, whose implied constant depends only on $\beta$, is at most $\eta/4$.
Since the vector coloring parameter and the degree exponent may vary across bands, we discretize them to ensure that the asymptotic guarantee of \autoref{thm:vector-k-coloring} applies for all choices of these parameters.
\textsc{VectorColoring}$(G',\{z_i\}_{i\in V(G')},k)$ runs the algorithm of \autoref{thm:vector-k-coloring} with $k_h=\min\{k'\in\cK_h:k'\ge k\}$ in place of $k$ and $\beta_h(G')=\min\{\beta'\in\cD_h:\Delta(G')\le|V(G')|^{\beta'}\}$ in place of $\beta$, where $\cK_h=\{2+h,2+2h,\ldots,4\}$, $\cD_h=\{h,2h,\ldots,1\}$, and $h$ is the grid step fixed above.
Each invocation uses the size and maximum degree of its own input graph $G'$, rather than those of the original graph.

\begin{algorithm}[H]
    \caption{\textsc{SparseProgress}}
    \label{alg:improved}
    \begin{algorithmic}[1]
        \Require graph $G=(V,E)$
        \Ensure an independent set $I$ with $|I|\ge T$, or $\mathsf{FAIL}$

        \State $(\{v_x\}_{x\in V},\cS)\gets\Call{BuildVector}{G}$
        \For{$(\rho_P,S,\cW)\in\cS$}
        \If{$S$ is a positive band}
        \State $\mathsf{res}\gets\Call{RoundPositive}{S,\rho_P,\cW}$
        \Else
        \State $\mathsf{res}\gets\Call{RoundNegative}{S,\rho_P,\cW}$
        \EndIf
        \If{$\mathsf{res}=\mathsf{SUCCESS}(I)$}
        \State \Return $I$
        \EndIf
        \EndFor
        \State $\mathsf{res}\gets$ \Call{IterativeGaussianLocalMaxima}{$G,\{v_x\}_{x\in V},\delta,T,R$}
        \If{$\mathsf{res}=\mathsf{SUCCESS}(I)$}
        \State \Return $I$
        \EndIf
        \State \Return $\mathsf{FAIL}$
    \end{algorithmic}
\end{algorithm}

\paragraph{Band rounding.}
For each band $S=S_{r,P}$ with witness $\cW=\cW_{r,P}$, \autoref{alg:improved} calls \autoref{alg:round-pos} if $S$ is positive and \autoref{alg:round-neg} if $S$ is negative.
Both procedures keep the vertices of $S$ with large witness mass and round their vectors with \textsc{VectorColoring}.

If $S$ is positive, then $\cW=(x_{\emptyset},\{x_i\}_{i\in S})$, and \textsc{RoundPositive} sets the threshold $\tau=\frac{1+2\rho_P}{3}-\varepsilon_{\mathrm{thr}}$ and the set $S'=\{i\in S:\norm{x_i}^{2}\ge\tau\}$.
It splits $S'$ into the set $I_0$ of isolated vertices of $G[S']$ and the rest $S''=S'\setminus I_0$.
The vectors $z_i = \frac{x_{i}-\norm{x_{i}}^2x_{\emptyset}}{\sqrt{\norm{x_i}^2(1-\norm{x_i}^2)}}$ for $i \in S''$ form an ordinary vector $\frac{1}{\tau}$-coloring, as shown in \autoref{lem:band-rounding}.

If $S$ is negative, then $\cW=(x_{\emptyset},\{(x_{i,1},x_{i,2})\}_{i\in S})$, and \textsc{RoundNegative} sets $\tau=\frac{2-2\rho_P}{3}-\varepsilon_{\mathrm{thr}}$ and $S'=\{i\in S:\norm{x_{i,1}}^{2}+\norm{x_{i,2}}^{2}\ge\tau\}$.
The vectors $z_i = \frac{x_{i,1}-x_{i,2}}{\norm{x_i}}$ for $i\in S'$ form an ordinary vector $\frac{3\tau-1}{2\tau-1}$-coloring, as shown in \autoref{lem:band-rounding}.

In both cases, the procedure returns $\mathsf{SUCCESS}$ with the resulting independent set if its size is at least $T$ and $\mathsf{FAIL}$ otherwise.

\begin{algorithm}[H]
    \caption{\textsc{RoundPositive}}
    \label{alg:round-pos}
    \begin{algorithmic}[1]
        \Require positive band $S$, lower endpoint $\rho_P\ge0$,
        witness $\cW=(x_{\emptyset},\{x_i\}_{i\in S})$
        \Ensure $\mathsf{SUCCESS}(I)$ or $\mathsf{FAIL}$

        \State $\tau\gets\frac{1+2\rho_P}{3}-\varepsilon_{\mathrm{thr}}$
        \State $S' \gets \{i\in S:\norm{x_i}^2\ge \tau\}$
        \State $I_0 \gets \{i\in S':N_{S'}(i)=\emptyset\}$
        \State $S''\gets S'\setminus I_0$
        \State $z_i\gets \frac{x_i-\norm{x_i}^2x_{\emptyset}}{\sqrt{\norm{x_i}^2(1-\norm{x_i}^2)}}$
        for every $i\in S''$
        \State $I'\gets\Call{VectorColoring}{G[S''],\{z_i\}_{i\in S''},\,\frac{1}{\tau}}$
        \State $I\gets I_0\cup I'$
        \State \Return $\mathsf{SUCCESS}(I)$ if $|I|\ge T$,
        and $\mathsf{FAIL}$ otherwise
    \end{algorithmic}
\end{algorithm}

\begin{algorithm}[H]
    \caption{\textsc{RoundNegative}}
    \label{alg:round-neg}
    \begin{algorithmic}[1]
        \Require negative band $S$, upper endpoint $\rho_P\le0$,
        witness $\cW=(x_{\emptyset},\{(x_{i,1},x_{i,2})\}_{i\in S})$
        \Ensure $\mathsf{SUCCESS}(I)$ or $\mathsf{FAIL}$

        \State $\tau\gets\frac{2-2\rho_P}{3}-\varepsilon_{\mathrm{thr}}$
        \State $S' \gets \{i\in S:\norm{x_{i,1}}^2+\norm{x_{i,2}}^2\ge \tau\}$
        \State $z_i\gets \frac{x_{i,1}-x_{i,2}}{\norm{x_{i}}}$
        for every $i\in S'$
        \State $I\gets\Call{VectorColoring}{G[S'],\{z_i\}_{i\in S'},\,\frac{3\tau-1}{2\tau-1}}$
        \State \Return $\mathsf{SUCCESS}(I)$ if $|I|\ge T$,
        and $\mathsf{FAIL}$ otherwise
    \end{algorithmic}
\end{algorithm}

For simplicity, in the analysis below, we treat the second and fourth constraints of \autoref{lem:build-vector} as exact, namely
$\ip{v_i}{v_j}= -\frac{1}{2}$ for every $ij \in E$
and
$\Gamma_{rx}(X) \succeq 0$ for every $rx \in E$.
The actual implementation satisfies these constraints only up to an additive error of $\varepsilon_{\mathrm{SDP}}$, as stated in \autoref{lem:build-vector}.
We justify the use of the exact constraints by showing that these errors do not affect our analysis.
Both issues are handled in \autoref{sec:B.2}.
\subsection{Guarantees of Failed Band Rounding}
\label{sec:4.3}

In this subsection, we show that a band whose rounding finds no large independent set is small.
The argument has three steps: \textsc{BuildVector} stores a witness of large mass for every band (\autoref{lem:build-vector}), the witness yields a large subset of the band with an ordinary vector coloring (\autoref{lem:band-rounding}), and if \textsc{VectorColoring} finds no large independent set in this subset, then the band is small (\autoref{lem:band-bound}).
These bounds are used in \autoref{sec:4.4} in order to obtain a sharper tail bound on the neighborhood maximum.

We use the parameters of \autoref{tab:params-4}; note that $H_0=D+2$ and $\delta=n^{-1/(D+2)}$.
Throughout, every vertex of $G$ has degree at most $n^{\beta}$.

\begin{lemma}[Guarantee of band rounding]
    \label{lem:band-rounding} Fix a band $S=S_{r,P}$.
    Let $0<\varepsilon_{\mathrm{thr}}\le1/30$.
    Then $S'$ satisfies $|S'|\ge\frac{\varepsilon_{\mathrm{thr}}}{2}|S|$.
    \begin{enumerate}
        \item[(i)] Suppose that $S$ is positive and $\rho_P \le \frac{1}{4}$, and let $\beta_{+}\in(0,1]$ satisfy $\Delta(G[S''])\le|S''|^{\beta_+}$.
            Whenever the internal call to \textsc{VectorColoring} meets its guarantee, \textsc{RoundPositive} constructs an independent set $I$ satisfying $|I|\ge(\varepsilon_{\mathrm{thr}}|S|)^{\frac{2\tau(1+\beta_{+})}{2\tau+\beta_{+}}-O(h)-o(1)}$.

        \item[(ii)] Suppose that $S$ is positive and $\rho_P>\frac{1}{4}$.
            Then $G[S]$ is edgeless, and hence we have $I=I_0=S'$.

        \item[(iii)] Suppose that $S$ is negative, and let $\beta_{-}\in(0,1]$ satisfy $\Delta(G[S'])\le|S'|^{\beta_-}$.
            Whenever the internal call to \textsc{VectorColoring} meets its guarantee, \textsc{RoundNegative} constructs an independent set $I$ satisfying $|I|\ge(\varepsilon_{\mathrm{thr}}|S|)^{\frac{2(1+\beta_{-})(2\tau-1)}{(4+3\beta_{-})\tau-(2+\beta_{-})}-O(h)-o(1)}$.
    \end{enumerate}
\end{lemma}

\begin{proof}
    First suppose that $S$ is positive.
    From the definitions $S'=\{j\in S:\norm{x_j}^{2}\ge\tau\}$ and $\tau=\frac{1+2\rho_P}{3}-\varepsilon_{\mathrm{thr}}$, we have
    \begin{equation*}
        \sum_{j\in S}\norm{x_j}^{2}
        \le|S'|+\tau|S\setminus S'|
        \le|S'|+\tau|S|.
    \end{equation*}
    Since the output of \textsc{BuildVector} satisfies
    \begin{equation*}
        \sum_{j\in S}\norm{x_j}^{2}
        \ge\frac{1+2\rho_P}{3}|S|-\varepsilon_{\mathrm{SDP}}
        =(\tau +\varepsilon_{\mathrm{thr}})|S| -\varepsilon_{\mathrm{SDP}},
    \end{equation*}
    we have $(\tau+\varepsilon_{\mathrm{thr}})|S| -\varepsilon_{\mathrm{SDP}}\le|S'|+\tau|S|$, which gives $|S'|\ge\varepsilon_{\mathrm{thr}}|S| -\varepsilon_{\mathrm{SDP}}\ge\frac{\varepsilon_{\mathrm{thr}}}{2}|S|$ for nonempty $S$, since $\varepsilon_{\mathrm{SDP}}=e^{-3n}\le\frac{\varepsilon_{\mathrm{thr}}}{2}$ for sufficiently large $n$.
    In the negative case, we can similarly show that $|S'|\ge\frac{\varepsilon_{\mathrm{thr}}}{2}|S|$.

    Suppose that $S$ is positive and $\rho_P > \frac{1}{4}$.
    We have $\ip{v_r}{v_j}+ \ip{v_r}{v_{j'}} \ge 2\rho_P > \frac{1}{2}$ for every edge $jj'\in E[S]$ with $r\notin\{j,j'\}$, contradicting the third constraint of \autoref{lem:build-vector}.
    Moreover, for the fixed vertex $r$ and every $j\in N_S(r)$, $\ip{v_r}{v_j}\le -\frac{1}{2}+ \varepsilon_{\mathrm{SDP}}< \frac{1}{4}$ holds, which contradicts the assumption $\rho_P > \frac{1}{4}$.

    Thus, $G[S]$ is edgeless and we can take an independent set $I$ as $I_0 = S'$.

    \paragraph{Construction of a vector coloring for a positive band ($\rho_P\le\frac{1}{4}$).}
    Within this band, $\tau = \frac{1+2\rho_P}{3}- \varepsilon_{\mathrm{thr}}\le \frac{1}{2} -\varepsilon_{\mathrm{thr}}< \frac{1}{2}$.
    From the constraints of $\Theta_{1}(S)$, we have $\ip{x_i}{x_j} =0$ and $\ip{x_{\emptyset}-x_i}{x_{\emptyset}-x_j}\ge0$ for $ij\in E[S]$.
    Hence $\norm{x_i}^{2}+\norm{x_j}^{2}\le1$ holds.
    Every vertex of $S''$ has a neighbor in $S'$ (and hence in $S''$).
    From the definition of $S'$, we have $\tau\le\norm{x_i}^{2},\norm{x_j}^{2}$, which gives $\tau\le\norm{x_i}^{2}\le1-\tau$ for every $i\in S''$.

    For $i\in S''$, set $z_{i}=\frac{x_{i}-\norm{x_i}^{2}x_{\emptyset}}{\sqrt{\norm{x_i}^{2}(1-\norm{x_i}^{2})}}$, which is a unit vector.
    For $ij\in E[S'']$,
    \begin{align*}
        \ip{z_i}{z_j}
        & =-\sqrt{\frac{\norm{x_i}^{2}}{1-\norm{x_i}^{2}}}\sqrt{\frac{\norm{x_j}^{2}}{1-\norm{x_j}^{2}}} \\
        & \le-\frac{\tau}{1-\tau} \\
        & =-\frac{1}{\frac{1}{\tau}-1}.
    \end{align*}
    Therefore, $\{z_{i}\}_{i\in S''}$ is a vector $\frac{1}{\tau}$-coloring, and $\tau<\frac{1}{2}$ gives $\frac{1}{\tau} > 2$.

    \paragraph{Construction of a vector coloring for a negative band.}
    Recall that $x_{i}=x_{i,1}+x_{i,2}$.
    From the constraints of $\Theta_{2}(S)$, we have $\ip{x_{i,1}}{x_{i,2}}=0$ and $\ip{x_{\emptyset}}{x_{i,c}}=\norm{x_{i,c}}^{2}$.
    Hence $\norm{x_{i,1}-x_{i,2}}^{2}=\norm{x_{i,1}}^{2}+\norm{x_{i,2}}^{2}=\norm{x_i}^{2}$.
    From the definition of $S'$, we have $\norm{x_i}^{2}\ge\tau$ for every $i\in S'$.
    Moreover, for a negative band we have $\tau=\frac{2-2\rho_P}{3}-\varepsilon_{\mathrm{thr}}\ge\frac{2}{3}-\varepsilon_{\mathrm{thr}}>\frac{1}{2}$.

    For $i\in S'$, set $z_{i}=\frac{x_{i,1}-x_{i,2}}{\norm{x_i}}$, which is a unit vector.
    For $ij\in E[S']$, the constraints of $\Theta_{2}(S)$ give $\ip{x_{i,1}-x_{i,2}}{x_{j,1}-x_{j,2}}=-\ip{x_i}{x_j}$ and $0\le\ip{x_{\emptyset}-x_i}{x_{\emptyset}-x_j}=1-\norm{x_i}^{2}-\norm{x_j}^{2}+\ip{x_i}{x_j}$, and hence $\ip{x_i}{x_j}\ge\norm{x_i}^{2}+\norm{x_j}^{2}-1$.
    Therefore,
    \begin{align*}
        \ip{z_i}{z_j}
        & =-\frac{\ip{x_i}{x_j}}{\norm{x_i}\norm{x_j}} \\
        & \le\frac{1-\norm{x_i}^{2}-\norm{x_j}^{2}}{\norm{x_i}\norm{x_j}} \\
        & \le\frac{1-2\norm{x_i}\norm{x_j}}{\norm{x_i}\norm{x_j}}
        & & (\because\text{AM--GM inequality}) \\
        & \le\frac{1}{\tau}-2
        & & (\because\norm{x_i}\norm{x_j}\ge\tau) \\
        & =-\frac{1}{\frac{3\tau-1}{2\tau-1}-1}.
    \end{align*}
    Therefore, $\{z_{i}\}_{i\in S'}$ is a vector $\frac{3\tau-1}{2\tau-1}$-coloring.

    Now, we apply \textsc{VectorColoring} to the constructed vectors $\{z_{i}\}_{i\in S''}$ with $k=\frac{1}{\tau}$ for a positive band and $\{z_{i}\}_{i\in S'}$ with $k=\frac{3\tau-1}{2\tau-1}$ for a negative band.
    For a positive band, $\frac{1}{3}-\varepsilon_{\mathrm{thr}}\le\tau<\frac{1}{2}$ gives $2<\frac{1}{\tau}\le\frac{10}{3}$, and for a negative band, $\frac{2}{3}-\varepsilon_{\mathrm{thr}}\le\tau<1$ gives $2<\frac{3\tau-1}{2\tau-1}\le\frac{27}{8}$, where we used $\varepsilon_{\mathrm{thr}}\le\frac{1}{30}$.
    Hence, $k \in[2,4]$ in both cases.
    Let $G'$ be the graph to which it is applied.
    \textsc{VectorColoring} runs the algorithm of \autoref{thm:vector-k-coloring} with $k_{h}$ and $\beta_{h}(G')$, and in both regimes of \autoref{thm:vector-k-coloring} it returns, with probability at least $\frac{3}{4}$, an independent set of size at least $|V(G')|^{1-\frac{(k_{h}-2)\beta_{h}(G')}{2+k_{h}\beta_{h}(G')}-o(1)}$.
    Since the pair $(k_{h},\beta_{h}(G'))$ lies in the finite set $\cK_{h}\times\cD_{h}$, the threshold for $|V(G')|$ and the $o(1)$ term are uniform over all bands.
    The function $(k,\beta')\mapsto\frac{(k-2)\beta'}{2+k\beta'}$ is increasing in $\beta'$, and since its numerator is linear in each argument and its denominator is at least $2$, changing $k\in[2,4]$ and $\beta'\in[0,2]$ by at most $h$ changes it by $O(h)$.
    Since $k\le k_{h}\le k+h$ and $\beta_{h}(G')\le\beta_{\pm}+h$, we obtain
    \begin{equation*}
        \frac{(k_{h}-2)\beta_{h}(G')}{2+k_{h}\beta_{h}(G')}
        \le\frac{(k-2)\beta_{\pm}}{2+k\beta_{\pm}} +O(h).
    \end{equation*}
    For $k=\frac{1}{\tau}$ and $k=\frac{3\tau-1}{2\tau-1}$, the quantity $1-\frac{(k-2)\beta'}{2+k\beta'}$ equals $\frac{2\tau(1+\beta')}{2\tau+\beta'}$ and $\frac{2(1+\beta')(2\tau-1)}{(4+3\beta')\tau-(2+\beta')}$, respectively.

    For a positive band, applying \textsc{VectorColoring} to $G[S'']$, where $S''=S'\setminus I_0$, we obtain an independent set $I'$ with $|I'|\ge|S''|^{\frac{2\tau(1+\beta_{+})}{2\tau+\beta_{+}}-O(h)-o(1)}$.
    We return $I=I'\cup I_0$, and its size is bounded as follows:
    \begin{align*}
        |I|
        & \ge|I_0|+|S''|^{\frac{2\tau(1+\beta_{+})}{2\tau+\beta_{+}}-O(h)-o(1)} \\
        & \ge|I_0|^{\frac{2\tau(1+\beta_{+})}{2\tau+\beta_{+}}-O(h)-o(1)}+|S''|^{\frac{2\tau(1+\beta_{+})}{2\tau+\beta_{+}}-O(h)-o(1)} \\
        & \ge|S'|^{\frac{2\tau(1+\beta_{+})}{2\tau+\beta_{+}}-O(h)-o(1)} \\
        & \ge\left(\frac{\varepsilon_{\mathrm{thr}}}{2}|S|\right)^{\frac{2\tau(1+\beta_{+})}{2\tau+\beta_{+}}-O(h)-o(1)} \\
        & \ge(\varepsilon_{\mathrm{thr}}|S|)^{\frac{2\tau(1+\beta_{+})}{2\tau+\beta_{+}}-O(h)-o(1)}.
    \end{align*}
    For a negative band, applying \textsc{VectorColoring} to $G[S']$, we obtain an independent set $I$ of size at least $(\varepsilon_{\mathrm{thr}}|S|)^{\frac{2(1+\beta_{-})(2\tau-1)}{(4+3\beta_{-})\tau-(2+\beta_{-})}-O(h)-o(1)}$.

    This completes the proof.
\end{proof}

We express the band-size bound on the scale $\frac{2\log{|S|}}{t^2}$.
The capacity $C_D(\rho)$ defined below gives the resulting upper bound, with an additive slack, when band rounding fails under its stated guarantees.
\begin{definition}[Capacity]
    \label{def:capacity} For $D>0$ and $\rho\in[-\frac{1}{2},1]$, define
    \begin{equation*}
        C_{D}(\rho)
        =
        \begin{cases}
        \displaystyle\min\left\{D+2,\ \frac{1}{2}\left(1+\sqrt{1+6D(D+1)\frac{1-2\rho}{1-4\rho}}\right)\right\}, & -\frac{1}{2}\le\rho\le0, \\[4mm]
        \displaystyle\min\left\{D+2,\ \frac{1}{2}\left(1+\sqrt{1+\frac{6D(D+1)}{1+2\rho}}\right)\right\}, & 0\le\rho\le\frac{1}{4}, \\[4mm]
        D+1, & \frac{1}{4}<\rho\le1.
        \end{cases}
    \end{equation*}
\end{definition}

\begin{lemma}[Band bound after failed rounding]
    \label{lem:band-bound} Fix a band $S=S_{r,P}$.
    Suppose that the corresponding band-rounding procedure returns $\mathsf{FAIL}$ and that, whenever it makes its internal call to \textsc{VectorColoring}, that call meets its guarantee.
    Then, for every $\rho\in P$,
    \begin{equation*}
        |S|
        \le\exp{\left((C_D(\rho)+\eta)\frac{t^{2}}{2}\right)}.
    \end{equation*}
\end{lemma}

\begin{proof}
    Let $s = \frac{2\log{|S|}}{t^{2}}$.
    When $s < D + 1 + \eta/2$, we have $s < C_D(\rho) + \eta/2$.
    Thus, we assume that $s \ge D + 1 + \eta/2$ below.

    First suppose that $S$ is positive and $\rho_P>\frac{1}{4}$.
    By \autoref{lem:band-rounding}, $|I|=|S'|\ge \frac{\varepsilon_{\mathrm{thr}}}{2}|S|$.
    Since $|I| < T$, we have $|S| < 2T/\varepsilon_{\mathrm{thr}}\le e^{(D+1+o(1))t^2/2}$, which gives $s < D + 1 + o(1) = C_D(\rho) + o(1)$.

    Below we consider bands with $\rho_P \in[-\frac{1}{2}, \frac{1}{4}]$.
    By \autoref{lem:band-rounding}, $|S'| \ge \frac{\varepsilon_{\mathrm{thr}}}{2}|S|$, so $|S'| = e^{(s+o(1))t^2/2}$.
    For a positive band, since the independent set returned does not reach the target size, we have $|I_0| < T \le e^{(D+1)t^2/2}= |S|e^{(D+1-s)t^2/2}$.
    Thus, $|S''| =|S'| - |I_0| \ge\frac{\varepsilon_{\mathrm{thr}}}{2}|S|-|S|e^{(D+1-s)t^2/2}=(\frac{\varepsilon_{\mathrm{thr}}}{2}-o(1))|S| = e^{(s+o(1))t^2/2}$.

    Since $\Delta(G)\le e^{Dt^2/2}$, and $|S''|=e^{(s+o(1))t^2/2}$ for a positive band and $|S'|=e^{(s+o(1))t^2/2}$ for a negative band, we have $\Delta(G[S''])\le e^{Dt^2/2}=|S''|^{D/s+o(1)}$ and $\Delta(G[S'])\le e^{Dt^2/2}=|S'|^{D/s+o(1)}$, respectively.
    Hence we may take $\beta_{\pm}=\frac{D}{s} + O(h) + o(1)$ in \autoref{lem:band-rounding}, which lies in $(0,1]$ since $s\ge D+1$.

    \paragraph{Positive case.}
    For a fixed $\tau \le \frac{1}{2}$, $\frac{2\tau(1+\beta')}{2\tau + \beta'}$ is decreasing in $\beta'$.
    By \autoref{lem:band-rounding},
    \begin{equation*}
        |I|
        \ge(\varepsilon_{\mathrm{thr}}|S|)^{\frac{2\tau(1+\beta_{+})}{2\tau+\beta_{+}}-O(h)-o(1)}
        \ge\exp{\left(\left(\frac{2\tau s(s+D)}{2\tau s+D}-O(h)-o(1)\right)\frac{t^{2}}{2} \right)}.
    \end{equation*}
    Failure of the band-rounding procedure gives
    \begin{equation*}
        \exp{\left(\left(\frac{2\tau s(s+D)}{2\tau s+D}-O(h)-o(1)\right)\frac{t^{2}}{2} \right)}
        \le|I|
        \le T
        \le e^{(D+1)t^2/2}.
    \end{equation*}
    Hence, we have $s\le\frac{1}{2}\bigl(1+\sqrt{1+\frac{2D(D+1)}{\tau}}\bigr)+O(h)+o(1)$.
    Substituting $\tau = \frac{1 + 2\rho_P}{3}- \varepsilon_{\mathrm{thr}}$,
    \begin{equation*}
        s
        \le\frac{1}{2}\left(1+\sqrt{1+\frac{6D(D+1)}{1+2\rho_P}}\right)+O(h+\varepsilon_{\mathrm{thr}})+o(1).
    \end{equation*}
    \paragraph{Negative case.}
    For a fixed $\frac{1}{2}<\tau<1$, $\frac{2(1+\beta')(2\tau-1)}{(4+3\beta')\tau-(2+\beta')}$ is decreasing in $\beta'$.
    By \autoref{lem:band-rounding} with $\beta_{-}$ as above, and since the band-rounding procedure fails, we obtain
    \begin{equation*}
        \exp{\left(\left(\frac{2(2\tau-1)s(s+D)}{(4s+3D)\tau-(2s+D)}-O(h)-o(1)\right)\frac{t^{2}}{2}\right)}
        \le|I|
        \le T
        \le e^{(D+1)t^2/2}.
    \end{equation*}
    Thus, we have $s \le \frac{1}{2}\bigl(1+\sqrt{1+\frac{2D(D+1)(3\tau - 1)}{2\tau - 1}}\bigr)+O(h)+o(1)$.
    Substituting $\tau = \frac{2 - 2\rho_P}{3}- \varepsilon_{\mathrm{thr}}$,
    \begin{equation*}
        s
        \le\frac{1}{2}\left(1+\sqrt{1+\frac{6D(D+1)(1-2\rho_P)}{1-4\rho_P}}\right)+O(h+\varepsilon_{\mathrm{thr}})+o(1).
    \end{equation*}
    Moreover, in either case, $|S| \le n$ gives $s \le D + 2$.
    Therefore, we obtain $s \le C_D(\rho_P) + O(h + \varepsilon_{\mathrm{thr}}) + o(1) \le C_D(\rho_P) + \eta/2$ by the choice of $\varepsilon_{\mathrm{thr}}$ and $h$ in \autoref{sec:4.2} and for sufficiently large $n$.
    For any $\rho \in P$, we have $|\rho-\rho_P|=o(1)$, and hence $|C_D(\rho) - C_D(\rho_P)|=o(1)\le\eta/2$ for sufficiently large $n$.
    Thus, we have
    \begin{equation*}
        |S|
        \le \exp{\left((C_D(\rho) + \eta)\frac{t^{2}}{2}\right)}.
    \end{equation*}
\end{proof}

\subsection{Tail Bound on the Neighborhood Maximum}
\label{sec:4.4} With $D_{\mathrm{a}}$ and $H$ as in \autoref{tab:params-4}, we have $H=H_0+3\eta=D_{\mathrm{a}}+2$.

In this subsection, we prove \autoref{lem:neighborhood-preparation}, which strengthens the tail bound on the neighborhood maximum used to construct the large-measure convex set $K_1$.
Starting from the threshold that follows from the degree bound alone, we lower the neighborhood threshold step by step along the nested cores.
The induction simultaneously controls projection thresholds over the entire core (the global bound) and within the neighborhood of each vertex (the local bound).
At each step, the previous local bound, applied at each candidate vertex, supplies the large-measure convex set for Ehrhard extrapolation, while the previous global bound ensures that the Gaussian projections of the possible transmission targets stay below the previous thresholds.
After the preparation, the resulting tail bound is used in the improved lifting argument in \autoref{sec:4.5}.

For a vertex $r\in V$, we use $\rho_{x}$ and $u_{x}$ of \autoref{sec:2.2}.

Put $\varepsilon_{t}=\exp{(-(\frac{1}{4}+\eta)\frac{t^{2}}{2})}$.
For $i\ge0$, define the global and local events
\begin{align*}
    \cG_{i} & =\left\{\exists x\in F_{i+1}:|\ip{g}{u_x}|>q_{A_i}(\rho_{x})t\right\}, & \cL_{i} & =\left\{M_{F_{i+1}}(r)>A_{i}t\right\},
\end{align*}
where the threshold function $q_{A_i}(\rho)$ is defined in \autoref{sec:4.4.1}.
Note that these events depend on $r$.

\begin{lemma}
    \label{lem:neighborhood-preparation} Suppose, for every vertex $r\in V$, every $P\in\cP$ and every $\rho\in P$, $|S_{r,P}|\le\exp{\left((C_D(\rho)+\eta)\frac{t^{2}}{2}\right)}$.
    Suppose that $p_{F_{j-1}}(x)\le e^{-t^2/2}$ for every $1\le j\le\ell$ and every $x\in F_{j}$.
    Then, for sufficiently large $n$, $m_{0}=o(\ell)$, and for every $0\le i\le m_{0}$ and every vertex $r\in V$, we have
    \begin{equation*}
        \Prb_{g}[\cG_{i}]\le(i+2)\varepsilon_{t}\qquad\text{and}\qquad\Prb_{g}[\cL_{i}]\le(i+2)\varepsilon_{t}.
    \end{equation*}
\end{lemma}

We first show how a tail bound on the neighborhood maximum yields a large-measure convex set for Ehrhard extrapolation.
This conversion lemma is used both within the preparatory induction for the tail bound on the neighborhood maximum and in the improved lifting argument.

\begin{lemma}
    \label{lem:profile-upper-body} Let $F\subseteq V$ and $x\in V$ with $\norm{u_x} >0$.
    Let $A\in\R$ and $c\ge0$, and let
    \begin{equation*}
        K_{1}(A,c)
        = \left\{g:\ip{g_{\perp}}{w_y}\le(A-c|\xi_{y}|)t \ \forall y\in N_{F}(x)\right\}.
    \end{equation*}
    Then
    \begin{equation*}
        1-\gmeasure(K_{1}(A,c))
        \le\frac{\Prb_{g}[M_{F}(x)>At]}{\Prb[Z\in[ct,ct+\tfrac{1}{t}]]},
    \end{equation*}
    where $Z\sim N(0,1)$.
    In particular, if $\kappa\ge0$, $c^{2}+\kappa^{2}<\frac{1}{4}$ and $\Prb_{g}[M_{F}(x)>At]\le t^{-4}\bPhi(t/2)$, then, for sufficiently large $t$,
    \begin{equation*}
        \gmeasure(K_{1}(A,c))
        \ge\Phi(\kappa t).
    \end{equation*}
\end{lemma}

\begin{proof}[Proof of \autoref{lem:profile-upper-body}]
    Recall the decomposition of $v_{xy}$ in \autoref{eq:vxy-decomposition} for $y\in N_{F}(x)$, and that $\ip{g}{w_y}=\ip{g_{\perp}}{w_y}$.
    Set
    \begin{equation*}
        Z
        =\ip{g}{v_{xr}}.
    \end{equation*}
    Then $\ip{g}{v_{xy}}=\xi_{y}Z+\ip{g_{\perp}}{w_y}$ for every $y\in N_{F}(x)$, and $Z$ is a standard Gaussian independent of $g_{\perp}$.

    For each $g\notin K_{1}(A,c)$, choose $y\in N_{F}(x)$ with $\ip{g_{\perp}}{w_y}>(A-c|\xi_{y}|)t$ using the fixed total order on $V$, and choose $s_{y}\in\{-1,1\}$ so that $s_{y}\xi_{y}=|\xi_{y}|$.
    Both choices depend only on $g_{\perp}$.
    If $s_{y}Z\in[ct,ct+\tfrac{1}{t}]$, then
    \begin{align*}
        \ip{g}{v_{xy}}
        & =\xi_{y}Z+\ip{g_{\perp}}{w_y} \\
        & >c|\xi_{y}|t+(A-c|\xi_{y}|)t
        & & (\because\xi_{y}Z=|\xi_{y}|s_{y}Z\ge c|\xi_{y}|t) \\
        & =At.
    \end{align*}
    Since $y$ and $s_{y}$ depend only on $g_{\perp}$, $Z$ is independent of $g_{\perp}$, and $s_{y}Z$ is a standard Gaussian for each fixed sign $s_{y}$, we obtain
    \begin{align*}
        \bigl(1-\gmeasure(K_{1}(A,c))\bigr)\,\Prb[Z\in[ct,ct+\tfrac{1}{t}]]
        & =\Prb_{g}[g\notin K_{1}(A,c)\text{ and }s_{y}Z\in[ct,ct+\tfrac{1}{t}]] \\
        & \le\Prb_{g}[M_{F}(x)>At].
    \end{align*}
    Dividing by $\Prb[Z\in[ct,ct+\tfrac{1}{t}]]>0$ proves the claim.

    If moreover $\kappa\ge0$, $c^{2}+\kappa^{2}<\frac{1}{4}$ and $\Prb_{g}[M_{F}(x)>At]\le t^{-4}\bPhi(t/2)$, then $0\le c<\frac{1}{2}$ and $\kappa<\frac{1}{2}$, and hence, for $t\ge1$, \autoref{fact:gaussian-interval} gives
    \begin{equation*}
        \Prb[Z\in[ct,ct+\tfrac{1}{t}]]
        \ge\frac{1}{t}\phi(ct+\tfrac{1}{t})
        =\frac{1}{\sqrt{2\pi}\,t}e^{-c^2t^2/2-c-1/(2t^2)}
        \ge\frac{e^{-1}}{\sqrt{2\pi}\,t}e^{-c^2t^2/2}.
    \end{equation*}
    Therefore, for sufficiently large $t$,
    \begin{align*}
        1-\gmeasure(K_{1}(A,c))
        & \le\frac{t^{-4}\bPhi(t/2)}{\Prb[Z\in[ct,ct+\tfrac{1}{t}]]} \\
        & \le e\sqrt{2\pi}\,t^{-3}e^{c^2t^2/2}\bPhi(t/2) \\
        & \le e\sqrt{2\pi}\,t^{-3}e^{(1/4-\kappa^2)t^2/2}\bPhi(t/2)
        & & (\because c^{2}<\tfrac{1}{4}-\kappa^{2}) \\
        & \le e\sqrt{2\pi}\,t^{-3}\bPhi(\kappa t)
        & & (\because\text{\autoref{fact:tail-monotone}},\ \kappa<\tfrac{1}{2}) \\
        & \le\bPhi(\kappa t).
    \end{align*}
\end{proof}

\subsubsection{Thresholds and Parameters}\label{sec:4.4.1}
The threshold $q_{\mathrm{c}}$ below is the one obtained directly from the band-size bounds and a Gaussian union bound (see \autoref{eq:tail-sum}).
The profile $f_A$ is a stronger threshold that we realize inductively, and $q_A=\min\{q_{\mathrm{c}},f_A\}$ combines the region where the direct estimate suffices with the region where the propagation analysis is needed.
Property (P1) below controls the change of the threshold in one step.
For $1.968\le A\le A_{0}$ and $\rho\in[-\frac{1}{2},1]$, we define $q_{\mathrm{c}}(\rho)= \sqrt{(1-\rho^{2})(C_{D_{\mathrm{a}}}(\rho)+\frac{1}{4})}$ and $q_{A}(\rho)=\min\{ q_{\mathrm{c}}(\rho),f_{A}(\rho)\}$, where
\begin{equation*}
    f_{A}(\rho)
    =\tilde{f}(A)\sqrt{1-\rho^{2}}-0.714\zeta_{1}+0.826\zeta_{1}^{2}-0.8\zeta_{2}+0.45\zeta_{2}^{2}+0.39\zeta_{3}+0.65\zeta_{3}^{2},
\end{equation*}
where $\tilde{f}(A)=2.466+0.52(A-2)+2.46(A-2)^{2}$, $\zeta_{1}=(-\frac{1}{40}-\rho)_{+}$, $\zeta_{2}=(\rho-\frac{1}{50})_{+}$, $\zeta_{3}=(\rho-\frac{1}{4})_{+}$ and $x_{+} =\max\{x,0\}$.
We use the following two properties of $f_{A}$ and $q_{A}$.
\begin{itemize}
    \item[(P1)] 
        For $1.968\le A\le A'\le A_{0}$ and $\rho\in[-\frac{1}{2},1]$,
        \begin{equation*}
            0
            \le q_{A'}(\rho)-q_{A}(\rho)
            \le2(A'-A).
        \end{equation*}
        Indeed, $\partial_{A}f_{A}(\rho)=\tilde{f}'(A)\sqrt{1-\rho^{2}}$ with $\tilde{f}'(A)=0.52+4.92(A-2)\in[0,2]$ for $1.968\le A\le A_{0}$, and $q_{\mathrm{c}}$ does not depend on $A$.

    \item[(P2)] 
        For $A=A_0$, $q_{\mathrm{c}} \le f_{A}$, that is, $q_A = q_{\mathrm{c}}$ (proved in \autoref{sec:C}).
\end{itemize}

Furthermore, we use the following lemma, verified by computation; its proof is given in \autoref{sec:C}.
Cases 1 and 2 of this lemma are used for the global and the local update, respectively.
\begin{restatable}
    {lemma}{magiccondition} \label{lem:source-bounds} Let $1.968\le A\le A_{0}$, and let $\rho\in[-\frac{1}{2},1)$ satisfy the assumption on $(A,\rho)$ in one of the two cases in the table below.
    Then there exist $(a,b,\mu,c,\kappa,\theta)$, independent of $\rho'$, such that
    \begin{enumerate}
        \item[(i)] $a,\mu,c,\kappa\ge0$, $b\in\R$ and $0<\theta<14$;
        \item[(ii)] $a^{2}+b^{2}+\mu^{2}\le1-10^{-12}$ and $c^{2}+\kappa^{2}\le\frac{1}{4}-10^{-20}$;
        \item[(iii)] $L^{2}+\Lambda(\mu+\theta(\mu+\kappa))^{2}\ge J+10^{-18}$;
        \item[(iv)] for every $\rho'$ satisfying $-\frac{1}{2}\le \rho'\le\frac{1}{2}-\rho$ and $\frac{3}{4}-\rho^{2}-(\rho')^{2}-\rho \rho'\ge0$,
            \begin{equation*}
                \frac{L(1/2+\rho \rho')-(1+\theta)b(\rho'+\rho/2)+\theta c|\rho'+\rho/2|}{\sqrt{1-\rho^{2}}} +\frac{3}{2}(1+\theta)a-\frac{\sqrt{3}}{2}\theta A
                \ge q_{A}(\rho')+3\cdot10^{-6}.
            \end{equation*}
    \end{enumerate}
    Here $(L,J,\Lambda)$ are as follows:
    \begin{equation*}
        \begin{array}{c|c|c|c|c}
        & \text{assumption} & L & J & \Lambda \\
        \hline \text{Case 1} & q_A(\rho)<q_{\mathrm{c}}(\rho) & \dfrac{q_A(\rho)}{\sqrt{1-\rho^{2}}} & C_{D_{\mathrm{a}}}(\rho)+\dfrac{1}{4} & 1 \\[4mm]
        \text{Case 2} & \rho=-\dfrac{1}{2} & A & D_{\mathrm{a}}+\dfrac{1}{4} & \dfrac{3999}{1000}
        \end{array}
    \end{equation*}
\end{restatable}
Condition (ii) supplies the measure estimates needed for Ehrhard extrapolation, condition (iii) provides sufficient exponential decay for the probability bounds, and condition (iv) ensures that the transmitted Gaussian projection exceeds the required threshold.
The same parameter choice must apply to every admissible neighboring correlation $\rho'$, since the convex sets in the Ehrhard argument use one parameter choice for all neighbors of $x$.
The computational verification certifies only the existence of such parameters.

In the rest of this subsection, for $x\in V$ with $\norm{u_x}>0$, we assume that $(a,b,\mu,c,\kappa,\theta)$ are parameters given by \autoref{lem:source-bounds} with $\rho=\rho_{x}$, in one of the two cases and for some value $A$ of the parameter there, and we let $(L,J,\Lambda)$ be the corresponding entries of its table.

Let $x\in V$ satisfy $\norm{u_x}>0$.
Then every $y\in N_{V}(x)$ satisfies
\begin{equation*}
    -\frac{1}{2}\le\rho_{y}\le\frac{1}{2}-\rho_{x}, \qquad \frac{3}{4}-\rho_{x}^{2} -\rho_{y}^{2}-\rho_{x}\rho_{y}\ge0,
\end{equation*}
that is, $\rho'=\rho_{y}$ satisfies the condition on $\rho'$ in \autoref{lem:source-bounds} with $\rho=\rho_{x}$.
Indeed, $\rho_{y}\ge-\frac{1}{2}$ and $\rho_{x}+\rho_{y}\le\frac{1}{2}$ by \autoref{lem:build-vector}, where if $y=r$ we have $\rho_{x}=-\frac{1}{2}$ and $\rho_{y}=1$ directly, and $\frac{3}{4}-\rho_{x}^{2}-\rho_{y}^{2}-\rho_{x}\rho_{y}\ge0$ since \autoref{eq:xi-rho} and $(\xi_y)^2\le1$ give $(\rho_{y}+\rho_{x}/2)^{2}\le\frac{3}{4}(1-\rho_{x}^{2} )$.

\subsubsection{Gaussian Bounds}\label{sec:4.4.2}
We first prove \autoref{lem:edge-neighborhood}, and then \autoref{lem:local-bound}, whose two assertions are used in the induction steps for the global and the local bound in the proof of \autoref{lem:neighborhood-preparation}.
The two assertions differ in how the vertex $x$ is chosen.
In the first assertion, $x$ is fixed, so the parent Gaussian projection $\ip{g}{u_x}$ is independent of the transverse Gaussian coordinates, and Ehrhard extrapolation applies.
In the second assertion, $x$ is the maximizer of $M_{F_j}(r)$ and depends on $g$; we keep the event that $x$ attains the maximum, and bound its correlation with the event that the transmission from $x$ to its children fails, using the additional SDP constraint and the Gaussian H\"older inequality.
For $x\in V$, we define $\cU_{x}=\Span\{v_{xy} :y\in N_{V}(x)\}$.
Since $v_{xy}\perp v_{x}$ for every $y\in N_{V}(x)$, we have $\cU_{x}\subseteq v_{x}^{\perp}$.

\begin{lemma}
    \label{lem:edge-neighborhood} Let $X$ be the Gram matrix of $\{v_{x}\}_{x\in V}$, and let $rx\in E$.
    Suppose that $\Gamma_{rx}(X)\succeq0$.
    Then $|\ip{u_1}{u_2}|\le\frac{1}{2}\norm{u_1}\norm{u_2}$ for all $u_1\in \cU_{r}$ and $u_2\in \cU_{x}$.
\end{lemma}

\begin{proof}
    Write $u_1=\sum_{y\in N_V(r)}(\alpha_1)_{y}v_{ry}$ and $u_2=\sum_{z\in N_V(x)}(\alpha_2)_{z}v_{xz}$.
    By the definition of $\Gamma_{rx}(X)$, for every $\lambda\in\R$,
    \begin{equation*}
        0
        \le(\alpha_1,\lambda\alpha_2)^{\mathsf{T}}\Gamma_{rx}(X)(\alpha_1,\lambda\alpha_2)
        =\norm{u_1}^{2}+4\lambda\ip{u_1}{u_2}+\lambda^{2}\norm{u_2}^{2}.
    \end{equation*}
    Hence the discriminant of this quadratic in $\lambda$ is nonpositive, giving $4\ip{u_1}{u_2}^{2}\le\norm{u_1}^{2}\norm{u_2}^{2}$, as claimed.
\end{proof}

\begin{lemma}
    \label{lem:local-bound} Fix a vertex $r\in V$.
    Let $1 \le j \le \ell$ and $1.968 \le A \le A' \le A_0$, with $A' - A \le w$.
    Assume that $p_{F_{j-1}}(x)\le e^{-t^2/2}$ and $\Prb_{g}[M_{F_{j-1}}(x)>A't]\le t^{-4}\bPhi(t/2)$ for every $x\in F_{j}$.
    In each of the following assertions, take the parameters $(a, b, \mu, c, \kappa, \theta)$ from the corresponding case of \autoref{lem:source-bounds}.
    Then, for sufficiently large $t$, the following hold.
    \begin{enumerate}
        \item[(i)] Let $x\in F_{j}$ satisfy $\norm{u_x}>0$ and $q_{A}(\rho_{x})<q_{\mathrm{c}}(\rho_{x})$.
        Let $L$ be as in Case~1 of \autoref{lem:source-bounds} with $\rho=\rho_{x}$.
        Then
        \begin{equation*}
            \Prb_g[|\ip{g}{u_x}| > L \norm{u_x}t, |\ip{g}{u_y}| \le q_{A'}(\rho_y)t\ \forall y \in N_{F_{j-1}}(x)]
            \le 2\bPhi(Lt)\bPhi((\mu + \theta(\mu + \kappa))t).
        \end{equation*}
        \item[(ii)]
        Let $X$ be the Gram matrix of $\{v_{x}\}_{x\in V}$, and suppose that $\Gamma_{rx}(X) \succeq 0$ for every $x \in N_{F_{j}}(r)$, and take the parameters from Case 2 of \autoref{lem:source-bounds}.
        Put $\cL=\{M_{F_{j}}(r)>At\}$.
        On $\cL$, let $x$ denote the maximizer of $M_{F_{j}}(r)$.
        Then, for all $p_{1},p_{2}>1$ with $(p_{1}-1)(p_{2}-1)\ge\frac{1}{4}$,
        \begin{equation*}
            \Prb_g[\cL, -\ip{g}{u_y} \le q_{A'}(\rho_y)t\ \forall y \in N_{F_{j-1}}(x)]
            \le (|N_{F_{j}}(r)|\bPhi(At))^{1-1/p_1}\bPhi((\mu + \theta(\mu + \kappa))t)^{1/p_2}\Prb_{g}[\cL]^{1/p_1}.
        \end{equation*}
    \end{enumerate}
\end{lemma}

\begin{proof}\leavevmode
    \paragraph{First assertion.}
    Let $x\in F_{j}$ satisfy $\norm{u_x}>0$.
    By the second equality of \autoref{eq:uy-identity}, every $y\in N_{F_{j-1}}(x)$ satisfies
    \begin{equation*}
        u_{y}
        =-\frac{1/2+\rho_{x}\rho_{y}}{1-\rho_{x}^{2}}\,u_{x}+\frac{\sqrt{3}}{2}w_{y},
    \end{equation*}
    where $u_{x}\in\Span(v_{r},v_{x})$ and $w_{y}\in\Span(v_{r},v_{x})^{\perp}$.
    Taking the inner product with $g$ and using $\ip{g}{w_{y}}=\ip{g_{\perp}}{w_{y}}$, we obtain
    \begin{equation*}
        \ip{g}{u_{y}}
        =-\frac{1/2+\rho_{x}\rho_{y}}{1-\rho_{x}^{2}}\,\ip{g}{u_x}+\frac{\sqrt{3}}{2}\ip{g_{\perp}}{w_{y}},
    \end{equation*}
    where $\ip{g}{u_x}$ is independent of $g_{\perp}$.
    Take $(a,b,\mu,c,\kappa,\theta)$ and $L$ from either case of \autoref{lem:source-bounds} with $\rho=\rho_{x}$, and put
    \begin{equation*}
        \cB_{x}
        =\left\{\forall y\in N_{F_{j-1}}(x):\ip{g_{\perp}}{w_{y}} \le(1+\theta)(\sqrt{3}a-b\xi_{y})t-\theta(A'-c|\xi_{y}|)t-\frac{1+\theta}{t} \right\}.
    \end{equation*}
    Since $\cB_{x}$ is defined through $\ip{g_{\perp}}{w_{y}}$ only, it depends only on $g_{\perp}$.
    Moreover, if $rx\in E$, then $v_{xr}\in \cU_{x}$, and hence $w_{y}=v_{xy}-\xi_{y}v_{xr}\in \cU_{x}$ for every $y\in N_{F_{j-1}}(x)$ by \autoref{eq:vxy-decomposition}.
    Since also $w_{y}\perp\Span(v_{r},v_{x})$ and $\ip{g_{\perp}}{w_{y}}=\ip{g}{w_{y}}$, in this case $\cB_{x}$ depends only on the Gaussian projection onto $\cU_{x}\cap\Span(v_{r},v_{x})^{\perp}$.
    We first bound $\Prb_{g}[\cB_{x}]$.
    Since $c^{2}+\kappa^{2}<\frac{1}{4}$ by (ii) of \autoref{lem:source-bounds}, and $\Prb_{g}[M_{F_{j-1}}(x)>A't]\le t^{-4}\bPhi(t/2)$, \autoref{lem:profile-upper-body} with $F=F_{j-1}$ and $A'$ in place of $A$ gives, for sufficiently large $t$,
    \begin{equation*}
        \Prb_{g}[\forall y\in N_{F_{j-1}}(x):\ip{g_{\perp}}{w_y}\le(A'-c |\xi_{y}|)t]
        \ge\Phi(\kappa t).
    \end{equation*}
    Moreover, since $a\ge0$ and $p_{F_{j-1}}(x)\le e^{-t^2/2}$, we have, for sufficiently large $t$,
    \begin{align*}
        \gmeasure(K_{0}(a,b))
        & \le2\pi t^{2}e^{-(1-a^2-b^2)t^2/2+a+|b|+t^{-2}}
        & & (\because\text{\autoref{lem:source-body}}) \\
        & \le2\pi e^{3}t^{2}e^{-(\mu^2+10^{-12})t^2/2}
        & & (\because\text{(ii) of \autoref{lem:source-bounds}}) \\
        & \le\bPhi(\mu t).
        & & (\because\text{\autoref{fact:gaussian-tail}})
    \end{align*}
    Since $\ip{g}{w_y}=\ip{g_{\perp}}{w_y}$, the sets $K_{0}$ and $K_{1}$ in \autoref{lem:conv-bound} for $Y=N_{F_{j-1}}(x)$, $z_{y}=w_{y}$, $A_{y}=(\sqrt{3}a-b\xi_{y})t-\tfrac{1}{t}$ and $B_{y}=(A'-c|\xi_{y}|)t$ are $K_{0}(a,b)$ and $K_{1}(A',c)$.
    Hence \autoref{lem:conv-bound} with $\mu_0=\mu t$ and $\kappa_0=\kappa t$ gives $\gmeasure(K_{-\theta})\le\bPhi((\mu+\theta(\mu+\kappa))t)$.
    Since $K_{-\theta}=\cB_{x}$ by definition, this means $\Prb_{g}[\cB_{x}]\le\bPhi((\mu+\theta(\mu+\kappa))t)$ for sufficiently large $t$, that is,
    \begin{equation*}
        \Prb_{g}\left[\exists y\in N_{F_{j-1}}(x):\ip{g_{\perp}}{w_{y}}>(1+\theta)(\sqrt{3}a-b\xi_{y})t-\theta(A'-c|\xi_{y}|)t-\frac{1+\theta}{t}\right]
        \ge \Phi((\mu+\theta(\mu+\kappa))t).
    \end{equation*}

    Next, we show that $\cB_{x}$ occurs whenever $-\ip{g}{u_x}>L\norm{u_x}t$ and $\ip{g}{u_y}\le q_{A'}(\rho_{y})t$ for every $y\in N_{F_{j-1}}(x)$.
    Suppose that $-\ip{g}{u_x}>L\norm{u_x}t$.
    Then, for every $y$ in the event above, we have
    \begin{align*}
        \frac{\ip{g}{u_{y}}}{t}
        & \ge\frac{L(1/2+\rho_{x}\rho_{y})}{\sqrt{1-\rho_{x}^{2}}}+\frac{\sqrt{3}}{2}\cdot\frac{\ip{g_{\perp}}{w_{y}}}{t} \\
        & >\frac{L(1/2+\rho_{x}\rho_{y})}{\sqrt{1-\rho_{x}^{2}}}+\frac{\sqrt{3}}{2}\left((1+\theta)(\sqrt{3}a-b\xi_{y})-\theta(A'-c|\xi_{y}|)\right)-\frac{\sqrt{3}(1+\theta)}{2t^{2}} \\
        & =\frac{L(1/2+\rho_{x}\rho_{y})-(1+\theta)b(\rho_{y}+\rho_{x}/2)+\theta c|\rho_{y}+\rho_{x}/2|}{\sqrt{1-\rho_{x}^{2}}}+\frac{3}{2}(1+\theta)a-\frac{\sqrt{3}}{2}\theta A'-\frac{\sqrt{3}(1+\theta)}{2t^{2}} \\
        & \ge q_{A}(\rho_{y})+3\cdot10^{-6}-\frac{\sqrt{3}}{2}\theta(A'-A)-\frac{\sqrt{3}(1+\theta)}{2t^{2}} \\
        & \ge q_{A'}(\rho_{y})+3\cdot10^{-6}-\left(2+7\sqrt{3}\right)w-\frac{15\sqrt{3}}{2t^{2}} \\
        & >q_{A'}(\rho_{y}),
    \end{align*}
    where the first inequality uses $\frac{1}{2}+\rho_{x}\rho_{y}\ge0$, which holds since $\rho_{x},\rho_{y}\in[-\frac{1}{2},1]$, the second holds since $y$ is in the event above, the equality follows from \autoref{eq:xi-rho}, the third inequality follows from (iv) of \autoref{lem:source-bounds}, which applies since $\rho'=\rho_{y}$ satisfies the condition on $\rho'$ there by the bounds on $\rho_{y}$ above, the fourth uses (P1), $A'-A\le w$ and $\theta<14$, and the last inequality holds for sufficiently large $t$ since $w=\log{t}/t^{2}$.
    
    Hence $-\ip{g}{u_x}>L\norm{u_x}t$ and $\ip{g}{u_y}\le q_{A'}(\rho_{y})t$ for every $y\in N_{F_{j-1}}(x)$ imply that the event above fails, that is, $\cB_{x}$ occurs.

    Since $\cB_{x}$ depends only on $g_{\perp}$, which is independent of $\ip{g}{u_x}$, we obtain
    \begin{align*}
        \Prb_{g}[-\ip{g}{u_x}>L\norm{u_x}t,\,\forall y\in N_{F_{j-1}}(x):\ip{g}{u_y}\le q_{A'}(\rho_{y})t]
        & \le\Prb_{g}[-\ip{g}{u_x}>L\norm{u_x}t]\,\Prb_{g}[\cB_{x}] \\
        & \le\bPhi(Lt)\bPhi((\mu+\theta(\mu+\kappa))t),
    \end{align*}
    where the second inequality follows from the bound on $\Prb_{g}[\cB_{x}]$ above.
    Applying this also to $-g$ gives the first assertion.

    \paragraph{Second assertion.}
    Let $x\in N_{F_{j}}(r)$.
    Then $rx\in E$ and $\rho_{x}=-\frac{1}{2}$, so $\norm{u_x}>0$, and $\ip{g}{v_{rx}}>At$ means $-\ip{-g}{u_x}>A\norm{u_x}t$, where $L=A$ in Case 2.
    Let $\cB_{x}^{-}$ be the event $\cB_{x}$ above for the parameters of Case 2, with $-g$ in place of $g$.
    By the bound on $\Prb_{g}[\cB_{x}]$ above and the symmetry of $g$, we have $\Prb_{g}[\cB_{x}^{-}]\le\bPhi((\mu+\theta(\mu+\kappa))t)$.
    Recall that, since $rx\in E$, the event $\cB_{x}^{-}$ depends only on the Gaussian projection onto $\cU_{x}\cap\Span(v_{r},v_{x})^{\perp}$.
    Let $\cL^{(x)}$ be the event that $\cL$ occurs and $x$ is the maximizer of $M_{F_{j}}(r)$ (ties are broken by the fixed total order on $V$).
    By what we have shown above for $\cB_{x}$, applied to $-g$, if $\cL^{(x)}$ occurs and $-\ip{g}{u_y}\le q_{A'}(\rho_{y})t$ for every $y\in N_{F_{j-1}}(x)$, then $\cB_{x}^{-}$ occurs.

    Since $v_{rx}\in \cU_{r}\subseteq v_r^{\perp}$ and $v_{rx}\in\Span(v_r,v_x)$, for every $z\in N_{F_j}(r)$, we may write
    \begin{equation*}
        v_{rz}=\ip{v_{rz}}{v_{rx}}v_{rx}+h_z, \qquad h_z\in \cU_{r}\cap\Span(v_r,v_x)^{\perp}.
    \end{equation*}
    Thus, conditional on $\ip{g}{v_{rx}}=s$, we have
    \begin{equation*}
        \ip{g}{v_{rz}}
        =\ip{v_{rz}}{v_{rx}}s+\ip{g}{h_z},
    \end{equation*}
    so $\cL^{(x)}$ depends only on the Gaussian projection onto $\cU_{r}\cap\Span(v_r,v_x)^{\perp}$.
    On the other hand, $\cB_x^-$ depends only on the Gaussian projection onto $\cU_{x}\cap\Span(v_r,v_x)^{\perp}$.
    Since both subspaces are orthogonal to $v_{rx}$, the Gaussian projections onto them are jointly independent of $\ip{g}{v_{rx}}$.
    Consequently, conditioning on $\ip{g}{v_{rx}}=s$ does not change their joint Gaussian distribution.
    By \autoref{lem:edge-neighborhood}, the correlation between these two subspaces is at most $\frac{1}{2}$.
    Thus, \autoref{cor:gaussian-holder} gives $\Prb_{g}[\cL^{(x)}\cap\cB_{x}^{-}\mid \ip{g}{v_{rx}}=s]\le\Prb_{g} [\cL^{(x)}\mid \ip{g}{v_{rx}}=s]^{1/p_1}\bPhi((\mu+\theta(\mu+\kappa))t)^{1/p_2}$.
    Since $\cL^{(x)}\subseteq\{\ip{g}{v_{rx}}>At\}$, integration and H\"older's inequality give
    \begin{align*}
        \Prb_{g}[\cL^{(x)}\cap\cB_{x}^{-}]
        & \le\bPhi((\mu+\theta(\mu+\kappa))t)^{1/p_2}\int_{At}^{\infty}\Prb_{g}[\cL^{(x)}\mid \ip{g}{v_{rx}}=s]^{1/p_1}\phi(s)\,\mathrm{d}s \\
        & \le\bPhi((\mu+\theta(\mu+\kappa))t)^{1/p_2}\bPhi(At)^{1-1/p_1}\Prb_{g}[\cL^{(x)}]^{1/p_1}.
    \end{align*}
    Finally, the events $\cL^{(x)}$ are pairwise disjoint and $\sum_{x\in N_{F_j}(r)}\Prb_{g}[\cL^{(x)}]=\Prb_{g}[\cL]$.
    Hence, by the inclusion above and H\"older's inequality once more, the probability in the second assertion is at most
    \begin{align*}
        \sum_{x\in N_{F_j}(r)}\Prb_{g}[\cL^{(x)}\cap\cB_{x}^{-}]
        & \le\bPhi((\mu+\theta(\mu+\kappa))t)^{1/p_2}\bPhi(At)^{1-1/p_1}\sum_{x\in N_{F_j}(r)}\Prb_{g}[\cL^{(x)}]^{1/p_1} \\
        & \le(|N_{F_j}(r)|\bPhi(At))^{1-1/p_1}\bPhi((\mu+\theta(\mu+\kappa))t)^{1/p_2}\Prb_{g}[\cL]^{1/p_1}.
    \end{align*}
\end{proof}

\subsubsection{Proof of the Tail Bound}\label{sec:4.4.3}
\begin{proof}[Proof of \autoref{lem:neighborhood-preparation}]
    \allowdisplaybreaks We prove the two bounds by induction on $i$.
    We call the bounds on $\cG_i$ and $\cL_i$ the \emph{global bound} and the \emph{local bound}, respectively.
    We take $t$ sufficiently large that $(m_{0}+2)\varepsilon_{t}\le t^{-4}\bPhi(t/2)$.
    Since $m_{0}=O(1/w)=O(t^{2}/\log{t})$ and $\ell=\Theta(\log{n})=\Theta(t^{2})$, we have $m_{0}=o(\ell)$, and hence $m_{0}+1\le\ell$ for sufficiently large $n$.
    Below, $r\in V$ is an arbitrary vertex, and the induction hypothesis is used for every vertex.

    The following estimate is used to sum the correlation-dependent failure probabilities over all vertices.
    For $P\in\cP$, choose $\rho_{P}^{*}\in\{\rho_{x}:x\in S_{r,P}\}$ minimizing $C_{D_{\mathrm{a}}}$ if $S_{r,P}\ne\emptyset$, and $\rho_{P}^{*}\in P$ arbitrarily otherwise.
    Every branch of $C_{D}(\rho)$ has $D$-derivative at least $1$, and hence $D_{\mathrm{a}}-D=3\eta$ gives $C_{D_{\mathrm{a}}}(\rho)-C_{D}(\rho)\ge3\eta$.
    Since every $x\in V$ lies in $S_{r,P}$ for some $P\in\cP$, we have, for sufficiently large $t$,
    \begin{align}
        \sum_{x\in V}2\exp{\left(-\left(C_{D_{\mathrm{a}}}(\rho_{x})+\frac{1}{4}\right)\frac{t^{2}}{2}\right)}
        & \le2\sum_{P\in\cP}|S_{r,P}|\exp{\left(-\left(C_{D_{\mathrm{a}}}(\rho_{P}^{*})+\frac{1}{4}\right)\frac{t^{2}}{2}\right)}\nonumber \\
        & \le2\sum_{P\in\cP}\exp{\left(-\left(\frac{1}{4}+C_{D_{\mathrm{a}}}(\rho_{P}^{*})-C_{D}(\rho_{P}^{*})-\eta\right)\frac{t^{2}}{2}\right)}\nonumber \\
        & \le2|\cP|\exp{\left(-\left(\frac{1}{4}+2\eta\right)\frac{t^{2}}{2}\right)}\nonumber \\
        & \le\varepsilon_{t}, \label{eq:tail-sum}
    \end{align}
    where the second inequality uses $|S_{r,P}|\le e^{(C_D(\rho_P^*)+\eta)t^2/2}$, which is the hypothesis of \autoref{lem:neighborhood-preparation}, that is, the band bound of \autoref{lem:band-bound} after failed band rounding, the third inequality uses $C_{D_{\mathrm{a}}}(\rho_{P}^{*})-C_{D}(\rho_{P}^{*})\ge3\eta$, and the last inequality uses $|\cP|=O(\log{n})=O(t^{2})$.

    \paragraph{1. Base case.}
    For every fixed vertex $x$, since $|N_{F_0}(x)|\le e^{Dt^2/2}$, $A_{0}^{2}=D_{\mathrm{a}} +\frac{1}{4}$ and $D_{\mathrm{a}}-D=3\eta$, by a union bound,
    \begin{equation*}
        \Prb_{g}[M_{F_0}(x)>A_{0}t]
        \le e^{Dt^2/2}\bPhi(A_{0}t)
        \le e^{-(A_0^2-D)t^2/2}
        =e^{-(1/4+3\eta)t^2/2}
        \le\varepsilon_{t}.
    \end{equation*}
    Since $F_{1}\subseteq F_{0}$, taking $x=r$ gives $\Prb_{g}[\cL_{0}]\le\varepsilon_{t}\le2\varepsilon_{t}$, which is the bound on $\cL_{0}$.

    Since $q_{\mathrm{c}}(\rho_{x})=\norm{u_x}\sqrt{C_{D_{\mathrm{a}}}(\rho_{x})+\frac{1}{4}}$, a union bound, \autoref{fact:gaussian-tail} and \autoref{eq:tail-sum} give
    \begin{equation*}
        \Prb_{g}[\exists x\in F_{0}:|\ip{g}{u_x}|>q_{\mathrm{c}}(\rho_{x})t]
        \le\sum_{x\in V}2\exp{\left(-\left(C_{D_{\mathrm{a}}}(\rho_{x})+\frac{1}{4}\right)\frac{t^{2}}{2}\right)}
        \le\varepsilon_{t}.
    \end{equation*}
    Since $q_{A_0}=q_{\mathrm{c}}$ by (P2) and $F_{1}\subseteq F_{0}$, this gives $\Prb_{g}[\cG_{0}] \le\varepsilon_{t}\le2\varepsilon_{t}$, which is the bound on $\cG_{0}$.

    \paragraph{2. Induction step for the global bound.}
    Assume that both bounds hold for some $i<m_{0}$.
    For every fixed $x\in F_{i+2}$, applying the bound on $\cL_{i}$ with $x$ in place of $r$ gives
    \begin{equation*}
        \Prb_{g}[M_{F_{i+1}}(x)>A_{i}t]
        \le(i+2)\varepsilon_{t}
        \le t^{-4}\bPhi(t /2).
    \end{equation*}
    For $x\in F_{i+2}$ with $f_{A_{i+1}}(\rho_{x})<q_{\mathrm{c}}(\rho_{x})$, take $(a,b,\mu,c,\kappa,\theta)$ of Case 1 with $A=A_{i+1}$, so that $L=\frac{f_{A_{i+1}}(\rho_{x})}{\sqrt{1-\rho_{x}^{2}}}$.
    Since $0\le A_{i}-A_{i+1}\le w$, the bound above and the assumption $p_{F_{i+1}}(x)\le e^{-t^2/2}$ for $x\in F_{i+2}$ show that the assumptions of \autoref{lem:local-bound} hold with $j=i+2$, $A=A_{i+1}$ and $A'=A_{i}$.
    Hence (i) of \autoref{lem:local-bound} applies to such $x$, and (iii) of \autoref{lem:source-bounds} gives $L^{2}+(\mu+\theta(\mu+\kappa ))^{2}\ge C_{D_{\mathrm{a}}}(\rho_{x})+\frac{1}{4}+10^{-18}$.
    Since $q_{A_{i+1}}=\min\{q_{\mathrm{c}},f_{A_{i+1}}\}$ and $f_{A}$ is increasing in $A$, we have
    \begin{align*}
        \Prb_{g}[\cG_{i+1}] 
        & \le\Prb_{g}[\cG_{i}]+\Prb_{g}[\cG_{i+1}\cap\cG_{i}^{c}] \\
        & \le(i+2)\varepsilon_{t}+\sum_{\substack{x\in F_{i+2}\\f_{A_{i+1}}(\rho_x)<q_{\mathrm{c}}(\rho_x)}}\Prb_{g}[|\ip{g}{u_x}|>f_{A_{i+1}}(\rho_{x})t,\,\cG_{i}^{c}] \\
        & \le(i+2)\varepsilon_{t}+\sum_{x\in V}2\exp{\left(-\left(C_{D_{\mathrm{a}}}(\rho_{x})+\frac{1}{4}\right)\frac{t^{2}}{2}\right)} \\
        & \le(i+3)\varepsilon_{t},
    \end{align*}
    where the second inequality uses the bound on $\cG_{i}$, the third uses \autoref{lem:local-bound}, \autoref{fact:gaussian-tail} and (iii) of \autoref{lem:source-bounds}, and the last uses \autoref{eq:tail-sum}; this is the bound on $\cG_{i+1}$.

    \paragraph{3. Induction step for the local bound.}
    Take $(a,b,\mu,c,\kappa,\theta)$ of Case 2 with $A=A_{i+1}$, so that $L=A_{i+1}$.
    Then (iii) of \autoref{lem:source-bounds} gives $A_{i+1}^{2}+\frac{3999}{1000}(\mu+\theta(\mu+\kappa))^{2}\ge D_{\mathrm{a}}+\frac{1}{4} +10^{-18}$.

    As in the global step, the assumptions of \autoref{lem:local-bound} hold with $j=i+2$, $A=A_{i+1}$ and $A'=A_{i}$, and the event $\cL$ in (ii) is $\cL_{i+1}$.
    If $\cL_{i+1}$ occurs with the maximum attained by $x$ and $\cG_{i}$ does not occur, then $-\ip{g}{u_{y}}\le q_{A_i}(\rho_{y})t$ for every $y\in N_{F_{i+1}}(x)$; that is, $\cL_{i+1}\cap\cG_{i}^{c}$ is contained in the event in (ii) of \autoref{lem:local-bound}.

    Then (ii) of \autoref{lem:local-bound} with $p_{1}=\frac{11997}{11996}$ and $p_{2}=3000$ gives
    \begin{align*}
        \Prb_{g}[\cL_{i+1}]
        & \le\Prb_{g}[\cG_{i}] + \Prb_{g}[\cL_{i+1}\cap\cG_{i}^{c}] \\
        & \le(i+2)\varepsilon_{t}+\left(|N_{F_{i+2}}(r)|\,\bPhi(A_{i+1}t)\,\bPhi((\mu+\theta(\mu+\kappa))t)^{3999/1000}\right)^{1/11997}\Prb_{g}[\cL_{i+1}]^{11996/11997} \\
        & \le(i+2)\varepsilon_{t}+\left(e^{Dt^2/2}e^{-(A_{i+1}^2+\frac{3999}{1000}(\mu+\theta(\mu+\kappa))^2)t^2/2}\right)^{1/11997}\Prb_{g}[\cL_{i+1}]^{11996/11997} \\
        & \le(i+2)\varepsilon_{t}+\left(\varepsilon_{t}e^{-\eta t^2}\right)^{1/11997}\Prb_{g}[\cL_{i+1}]^{11996/11997},
    \end{align*}
    where the second inequality uses the bound on $\cG_{i}$ and \autoref{lem:local-bound}, the third uses \autoref{fact:gaussian-tail}, and the last uses $A_{i+1}^{2}+\frac{3999}{1000}(\mu+\theta(\mu+\kappa))^{2}\ge D_{\mathrm{a}}+\frac{1}{4}$ and $D_{\mathrm{a}}-D=3\eta$.
    Thus, it remains to show that the last term is at most $\varepsilon_{t}$.

    Since $x^{11996/11997}\le1+x$ for $x\ge0$, applying this with $x=\Prb_{g}[\cL_{i+1}]/\varepsilon_{t}$ shows that the last term is at most $e^{-\eta t^2/11997}(\varepsilon_{t}+\Prb_{g}[\cL_{i+1}])$.
    Since $e^{-\eta t^2/11997}\le\frac{1}{2}$ for sufficiently large $t$, the display above then gives $\Prb_{g}[\cL_{i+1}]\le2(i+3)\varepsilon_{t}\le2(m_{0}+2)\varepsilon_{t}$.
    Hence the last term is at most $(2m_{0}+5)e^{-\eta t^2/11997}\varepsilon_{t}$, which is at most $\varepsilon_{t}$ for sufficiently large $t$ since $m_{0}=O(t^{2}/\log{t})$.
    Therefore $\Prb_{g}[\cL_{i+1}]\le(i+3)\varepsilon_{t}$, which is the bound on $\cL_{i+1}$.
    This completes the induction.
\end{proof}

\subsection{Improved Lifting to the Preceding Level}
\label{sec:4.5}
We now return to the Gaussian cover lifting argument.
The convex set $K_0$ of small Gaussian measure is constructed from the local-maximum probability as before.
For the convex set $K_1$ of large Gaussian measure, we use the prepared tail bound on the neighborhood maximum through \autoref{lem:profile-upper-body} in place of the degree-based construction.
This refinement allows us to maintain positive threshold drift for a larger target independent set.
We keep the parameters and conventions of \autoref{sec:4.4}.

\begin{lemma}[Improved single-step lifting]
    \label{lem:lifting-improved} Let $1\le j\le\ell$, and let $W\subseteq F_{j}$ be fixed independently of the Gaussian $g$.
    Suppose that $\{u_{x}\}_{x\in W}$ is a $(Ut,q)$-cover for some $0<U\le\sqrt{H}$.
    Assume that $p_{F_{j-1}}(x)\le e^{-t^2/2}$ and $\Prb_{g}[M_{F_{j-1}}(x)>1.968t]\le t^{-4}\bPhi(t/2)$ for every $x\in W$, and set $W'=N_{F_{j-1}}(W)$.
    Then, for sufficiently large $t$, $\{u_{y}\}_{y\in W'}$ is a $((U+1/3000)t,q-t^{-100})$-cover.
\end{lemma}

In place of the choices in \autoref{sec:3.4}, define $\nu=\frac{437}{2000}$, and let
\begin{equation*}
    \psi(s)
    =
    \begin{cases}
    \sqrt{3}\sqrt{1-s^2}, & s\le\nu, \\[2mm]
    \displaystyle\frac{\sqrt{3}(1-\nu s)}{\sqrt{1-\nu^2}}, & s\ge\nu.
    \end{cases}
\end{equation*}

\begin{lemma}
    \label{lem:wy-transverse-bound-improved} Let $1\le j\le\ell$, and let $x\in F_{j}$ satisfy $\norm{u_x}>0$, $p_{F_{j-1}}(x)\le e^{-t^2/2}$ and $\Prb_{g}[M_{F_{j-1}}(x)>1.968t]\le t^{-4}\bPhi(t/2)$.
    Let $0<L<\sqrt{H}$ and $\rho \in[-1, 1]$.
    Then, for sufficiently large $t$, over the Gaussian coordinates orthogonal to $\Span(v_{r},v_{x})$,
    \begin{equation*}
        \Prb_{g}\left[\exists y\in N_{F_{j-1}}(x):\ip{g}{w_y}+L\rho\xi_{y}t\ge \left(\psi\left(\sqrt{H-L^{2}(1-\rho^{2})}\right)-200w\right)t\right]
        \ge \Phi\left(\sqrt{H-L^{2}}\,t\right).
    \end{equation*}
\end{lemma}

\begin{lemma}
    \label{lem:U-increase-improved} For every $0\le U\le \sqrt{H}$,
    \begin{equation*}
        \frac{\sqrt{3}}{2}\psi\left(\sqrt{H-U^{2}}\right)
        \ge\frac{U}{2}+\frac{1}{1500}.
    \end{equation*}
\end{lemma}

\autoref{lem:wy-transverse-bound-improved} is obtained by combining the bound on $\gmeasure(K_{0}(a,b))$ in \autoref{lem:source-body} with the bound on $\gmeasure(K_{1}(A,c))$ in \autoref{lem:profile-upper-body} through \autoref{lem:conv-bound}, and \autoref{lem:U-increase-improved} turns the threshold $\psi$ of \autoref{lem:wy-transverse-bound-improved} into a uniform positive drift.
We first deduce \autoref{lem:lifting-improved} from \autoref{lem:wy-transverse-bound-improved} and \autoref{lem:U-increase-improved}.

\begin{proof}[Proof of \autoref{lem:lifting-improved}]
    We follow the proof of \autoref{lem:lifting}, replacing the cover threshold $(U+40w)t$ by $(U+1/3000)t$.
    Define $\cE_{x}$ and $\cE^{\mathrm{bad}}_{x}$ as in that proof with this cover threshold.

    Fix $x\in W$ with $\norm{u_x}>0$, and condition on $Z_{x}=-\ip{g}{u_x/\norm{u_x}}=z$, where $Ut/\norm{u_x}\le z<\sqrt{H}\,t$.
    Set $L=z/t$, $\rho=\rho_x$, and $U'=\norm{u_x}z/t$.
    Then $U\le U'<\sqrt{H}$.
    By \autoref{lem:wy-transverse-bound-improved},
    \begin{equation*}
        \Prb_{g}\left[\exists y\in N_{F_{j-1}}(x):\ip{g}{w_y}+L\rho\xi_{y}t\ge \left(\psi\left(\sqrt{H-(U')^{2}}\right)-200w\right)t\ \middle|\ Z_{x}=z\right]
        \ge\Phi\left(\sqrt{Ht^{2}-z^{2}}\right).
    \end{equation*}
    Hence, for every $y$ as in this event and sufficiently large $t$,
    \begin{align*}
        \ip{g}{u_y}
        & \ge\frac{U't}{2}+\frac{\sqrt{3}}{2}\left(\psi\left(\sqrt{H-(U')^{2}}\right)-200w\right)t \\
        & \ge\frac{U't}{2}+\left(\frac{U'}{2}+\frac{1}{1500}\right)t-100\sqrt{3}wt
        & & (\because\text{\autoref{lem:U-increase-improved}}) \\
        & =\left(U'+\frac{1}{1500}-100\sqrt{3}w\right)t \\
        & \ge\left(U+\frac{1}{3000}\right)t.
        & & (\because U'\ge U,\ 100\sqrt{3}w\le1/3000)
    \end{align*}
    By \autoref{lem:wy-transverse-bound-improved}, $\Prb_{g}[\cE^{\mathrm{bad}}_{x}\mid Z_{x}=z]\le\bPhi(\sqrt{Ht^{2}-z^{2}})$.
    The same integration over $z$ and union bound as in the proof of \autoref{lem:lifting} give
    \begin{equation*}
        \Prb_{g}\left[\bigcup_{x\in W}\cE^{\mathrm{bad}}_{x}\right]
        \le\left(1+\frac{\sqrt{H}\,t}{\sqrt{2\pi}} \right)e^{-(H-H_0)t^2/2}
        \le t^{-100}
    \end{equation*}
    for sufficiently large $t$, where $H-H_{0}=3\eta>0$.
    Gaussian symmetry and the same event inclusion now show that $\{u_{y}\}_{y\in W'}$ is a $((U+1/3000)t,q -t^{-100})$-cover.
\end{proof}

Next, we prove \autoref{lem:U-increase-improved}.
\begin{proof}[Proof of \autoref{lem:U-increase-improved}]
    The estimate in the proof of \autoref{lem:U-increase}, before substituting $\nu^{2}=(D-3)/(4D)$, gives
    \begin{align*}
        \frac{\sqrt{3}}{2}\psi\left(\sqrt{H-U^{2}}\right)-\frac{U}{2}
        & \ge\frac{3}{2\sqrt{1-\nu^{2}}}-\sqrt{H\left(\frac{1}{4}+\frac{9\nu^{2}}{4(1-\nu^{2})}\right)} \\
        & =\frac{3}{2\sqrt{1-\nu^{2}}}-\frac{\sqrt{H(1+8\nu^{2})}}{2\sqrt{1-\nu^{2}}}
    \end{align*}
    for every $0\le U\le\sqrt{H}$.
    Since $H<6.50621$ and $\nu=\frac{437}{2000}$,
    \begin{equation*}
        9-H(1+8\nu^{2})
        >9-6.50621\cdot\frac{690969}{500000}
        =\frac{441058251}{50000000000}
        >\frac{8}{1000}.
    \end{equation*}
    Thus $\sqrt{H(1+8\nu^{2})}<3$, and rationalizing gives
    \begin{equation*}
        \frac{3-\sqrt{H(1+8\nu^{2})}}{2\sqrt{1-\nu^{2}}}
        =\frac{9-H(1+8\nu^{2})}{2\sqrt{1-\nu^{2}}\left(3+\sqrt{H(1+8\nu^{2})}\right)}
        >\frac{8/1000}{12}
        =\frac{1}{1500},
    \end{equation*}
    proving the claim.
\end{proof}

Now, we can prove \autoref{lem:wy-transverse-bound-improved}.

\begin{proof}[Proof of \autoref{lem:wy-transverse-bound-improved}]
    Put $s=\sqrt{H-L^{2}(1-\rho^{2})}$.
    We choose the parameters so that $|L\rho - (1 + \theta)b| \le \theta c$, ensuring that the terms dependent on $\xi_y$ produced by extrapolation have a nonnegative sum for every neighbor.
    If $s\le\nu$, choose
    \begin{align*}
        a_{0} & =\sqrt{1-s^{2}}, & b & =L\rho, & \mu & =\sqrt{H-L^{2}}, & c & =\kappa=\theta=0.
    \end{align*}
    If $s\ge\nu$, choose
    \begin{align*}
        a_{0} & =\sqrt{1-\nu^{2}}, & b & =\frac{\nu L\rho}{s}, & \mu & =\frac{\nu\sqrt{H-L^{2}}}{s}, \\
        c & =\frac{499|L\rho|}{1000s}, & \kappa & =\frac{499\sqrt{H-L^2}}{1000s}, & \theta & =\frac{s-\nu}{\nu+499/1000}.
    \end{align*}
    Since $L^{2}\rho^{2}+H-L^{2}=s^{2}$, both choices satisfy
    \begin{align*}
        a_{0}^{2}+b^{2}+\mu^{2} & =1, & c^{2}+\kappa^{2} & \le\left(\frac{499}{1000}\right)^{2}<\frac{1}{4}, \\
        \mu+\theta(\mu+\kappa) & =\sqrt{H-L^{2}}, & |L\rho-(1+\theta)b| & \le\theta c.
    \end{align*}
    Moreover, $a_{0}\ge\sqrt{1-\nu^{2}}>\tfrac{\sqrt{3}}{2}$, $|b|,\mu\le\nu<\frac{1}{4}$, and $0\le\theta<4$.
    For $s\ge\nu$, the bound $\theta<4$ is equivalent to $s<5\nu+4\cdot\frac{499}{1000}$, which follows from $s\le\sqrt{H}$ and $(5\nu +4\cdot\frac{499}{1000})^{2}-H>(5\nu +4\cdot\frac{499}{1000})^{2}-6.50621=\frac{12130489}{4000000}>0$.

    For $s\le\nu$, $(1+\theta)\sqrt{3}a_{0}-1.968\theta=\psi(s)$.
    For $s\ge\nu$, the exact inequality
    \begin{equation*}
        3\left(1+\frac{499}{1000}\nu\right)^{2}-1.968^{2}(1-\nu^{2})
        =\frac{6941936163}{4000000000000}
        >0
    \end{equation*}
    gives $1.968<\frac{\sqrt{3}(1+499\nu/1000)}{\sqrt{1-\nu^{2}}}=\sqrt{3}a_{0} +\frac{\sqrt{3}\nu(\nu+499/1000)}{\sqrt{1-\nu^{2}}}$, by taking square roots, dividing by $\sqrt{1-\nu^{2}}$ and using $a_{0}=\sqrt{1-\nu^{2}}$.
    Hence, since $\theta\ge0$ and $\theta(\nu+\frac{499}{1000})=s-\nu$, we obtain
    \begin{equation*}
        (1+\theta)\sqrt{3}a_{0}-1.968\theta
        \ge\sqrt{3}a_{0}-\frac{\sqrt{3}\nu(\nu+499/1000)}{\sqrt{1-\nu^{2}}} \theta
        =\frac{\sqrt{3}(1-\nu s)}{\sqrt{1-\nu^{2}}}
        =\psi(s).
    \end{equation*}

    Set $a=a_{0}-20w$.
    As in the proof of \autoref{lem:wy-transverse-bound}, for sufficiently large $t$, we have $a\ge0$ and $1-a^{2}-b^{2}=\mu^{2}+40a_{0}w-400 w^{2}\ge \mu^{2}+30w$, and hence
    \begin{equation*}
        2\pi t^{2}e^{-(1-a^2-b^2)t^2/2+a+|b|+t^{-2}}
        \le2\pi e^{2}t^{-13}e^{-\mu^2t^2/2}
        \le\bPhi(\mu t).
    \end{equation*}
    The last inequality follows from \autoref{fact:gaussian-tail} and the uniform bound $\mu\le\nu$.

    Since $a\ge0$ and $p_{F_{j-1}}(x)\le e^{-t^2/2}$, \autoref{lem:source-body} and the display above give $\gmeasure(K_{0}(a,b))\le\bPhi(\mu t)$.

    Next, since $c^{2}+\kappa^{2}<\frac{1}{4}$ and $\Prb_{g}[M_{F_{j-1}}(x)>1.968t]\le t^{-4}\bPhi(t/2)$, \autoref{lem:profile-upper-body} with $F=F_{j-1}$ gives, for sufficiently large $t$, $\gmeasure(K_{1}(1.968,c))\ge\Phi(\kappa t)$.

    Since $\ip{g}{w_y}=\ip{g_{\perp}}{w_y}$, the sets $K_{0}$ and $K_{1}$ in \autoref{lem:conv-bound} for $Y=N_{F_{j-1}}(x)$, $z_{y}=w_{y}$, $A_{y}=(\sqrt{3}a-b\xi_{y})t-\tfrac{1}{t}$ and $B_{y}=(1.968-c|\xi_{y}|)t$ are $K_{0}(a,b)$ and $K_{1}(1.968,c)$.
    Hence \autoref{lem:conv-bound} with $\mu_0=\mu t$ and $\kappa_0=\kappa t$ gives $\gmeasure(K_{-\theta})\le\bPhi((\mu+\theta(\mu+\kappa))t)$, and the complement of $K_{-\theta}$ is the event below.
    Therefore, since $\mu+\theta(\mu+\kappa)=\sqrt{H-L^{2}}$,
    \begin{equation*}
        \Prb_{g}\left[\exists y\in N_{F_{j-1}}(x):\ip{g_{\perp}}{w_y}>(1+\theta)(\sqrt{3}a-b\xi_{y})t-\theta(1.968-c|\xi_{y}|)t-\frac{1+\theta}{t}\right]
        \ge\Phi\left(\sqrt{H-L^{2}}\,t\right).
    \end{equation*}
    For every $y$ as in this event, we obtain
    \begin{equation*}
        \ip{g}{w_y}+L\rho\xi_{y}t
        >\Bigl((1+\theta)\sqrt{3}a_{0}-1.968\theta+\bigl(L\rho-(1+\theta)b\bigr)\xi_{y}+\theta c|\xi_{y}|\Bigr)t-20\sqrt{3}(1+\theta)wt-\frac{1+\theta}{t}.
    \end{equation*}
    Hence, since $(1+\theta)\sqrt{3}a_{0}-1.968\theta\ge\psi(s)$, for sufficiently large $t$,
    \begin{align*}
        \ip{g}{w_y}+L\rho\xi_{y}t
        & >\bigl((1+\theta)\sqrt{3}a_{0}-1.968\theta\bigr)t-20\sqrt{3}(1+\theta)wt-\frac{1+\theta}{t}
        & & (\because|L\rho-(1+\theta)b|\le\theta c) \\
        & \ge\psi(s)t-\left(100\sqrt{3}+\frac{5}{\log{t}}\right)wt & & (\because\theta<4) \\
        & \ge(\psi(s)-200w)t.
    \end{align*}
    
    Thus
    \begin{align*}
        & \Prb_{g}\left[\exists y\in N_{F_{j-1}}(x):\ip{g}{w_y}+L\rho\xi_{y}t\ge(\psi(s)-200w)t\right] \\
        & \qquad\ge\Prb_{g}\left[\exists y\in N_{F_{j-1}}(x):\ip{g_{\perp}}{w_y}>(1+\theta)(\sqrt{3}a-b\xi_{y})t-\theta(1.968-c|\xi_{y}|)t-\frac{1+\theta}{t}\right] \\
        & \qquad\ge\Phi\left(\sqrt{H-L^{2}}\,t\right),
    \end{align*}
    which is the claim since $s=\sqrt{H-L^{2}(1-\rho^{2})}$.
\end{proof}
\subsection{Improved Iterated Lifting}\label{sec:4.6}

We now complete the analysis of \textsc{SparseProgress} (\autoref{alg:improved}).

In each trial of \textsc{SparseProgress}, we amplify every call to \textsc{VectorColoring} inside a band-rounding procedure by $\Theta(\log{n})$ independent repetitions and return the largest independent set over these repetitions.
By \autoref{thm:vector-k-coloring}, each repetition meets its guarantee with probability at least $\frac{3}{4}$, so if the hidden constant is large enough, an amplified call misses its guarantee with probability at most $n^{-3}$.
Since one trial makes at most $n|\cP|$ such calls, all of them meet their guarantees with probability at least $1-|\cP|/n^{2}=1-o(1)$.

The preparatory induction proceeds toward deeper cores, whereas cover lifting proceeds in the opposite direction.
Since the preparation uses only $m_0 = o(\ell)$ levels and a constant number of improved lifts suffices, every lift takes place within the cores where the prepared tail bound on the neighborhood maximum is available.
\begin{lemma}
    \label{lem:improved-sparse-progress} Let $G=(V,E)$ be a $3$-colorable graph on $n$ vertices with $\Delta(G)\le n^{\beta}$.
    For sufficiently large $n$, one trial of \textsc{SparseProgress} runs in polynomial time and returns an independent set $I$ with $|I|\ge T$ with probability at least $\frac{1}{2}$.
    Consequently, for any fixed constant $C>0$, running at most $\lceil C\log_2{n}\rceil$ independent trials gives a polynomial-time algorithm that returns such a set with probability at least $1-n^{-C}$.
\end{lemma}

\begin{proof}
    Consider one trial of \textsc{SparseProgress}, and suppose that all amplified \textsc{VectorColoring} calls satisfy their guarantees and that all empirical estimates used by \autoref{alg:sparse} satisfy
    \begin{equation*}
        |\hat{p}_{j}(x)-p_{F_{j-1}}(x)|
        \le\delta/8.
    \end{equation*}
    The latter event has probability at least $\frac{3}{4}$ by \autoref{lem:estimates}.

    Suppose for contradiction that this execution returns no independent set of size at least $T$.
    Then every band-rounding procedure returns $\mathsf{FAIL}$, so \autoref{lem:band-bound} gives, for every vertex $r$, every $P\in\cP$, and every $\rho\in P$,
    \begin{equation*}
        |S_{r,P}|
        \le\exp{\left((C_{D}(\rho)+\eta)\frac{t^{2}}{2}\right)}.
    \end{equation*}
    \autoref{alg:sparse} must also return $\mathsf{FAIL}$.
    By \autoref{lem:deleted-bound},
    \begin{equation*}
        |F_{\ell}|
        >n-\frac{2\ell T}{\delta}
        =n-o(n),
    \end{equation*}
    and hence $F_{\ell}\ne\emptyset$ for sufficiently large $n$.
    Moreover, \autoref{lem:estimates} gives
    \begin{equation*}
        p_{F_{j-1}}(x)
        <\delta
        =e^{-t^2/2}\qquad(\forall j,\,\forall x\in F_{j}).
    \end{equation*}
    Thus the assumptions of \autoref{lem:neighborhood-preparation} hold.

    Choose a root $r\in F_{\ell}$, and let $W_{0}$ and $U_{0}$ be as in \autoref{sec:3.3}.
    By \autoref{cor:initial-cover}, $\{u_{x}\}_{x\in W_0}$ is a $(U_{0}t,1-O(te^{-t^2/4}))$-cover.

    Let $m$ be the smallest integer such that $U_{m}>\sqrt{H}$.
    Since $H=D_{\mathrm{a}}+2$ is constant, $m=O(1)$.
    Also $m_{0}=o(\ell)$ by \autoref{lem:neighborhood-preparation}, so $m+m_{0}+1<\ell$ for sufficiently large $n$.
    For $0\le i<m$, define $W_{i+1}=N_{F_{\ell-i-2}}(W_i)$.
    Then $W_{i}\subseteq F_{\ell-i-1}$, and for $0\le i<m$, since $\ell-i-1\ge\ell-m>m_{0}+1$, we have $F_{\ell-i-2}\subseteq F_{m_0+1}$.
    Since $A_{m_0}=1.968$, applying the bound on $\cL_{m_0}$ of \autoref{lem:neighborhood-preparation} to the vertex $x$ gives $\Prb_{g}[M_{F_{\ell-i-2}}(x)>1.968t]\le\Prb_{g}[M_{F_{m_0+1}}(x)>A_{m_0}t]\le(m_{0}+2)\varepsilon_{t}\le t^{-4}\bPhi(t/2)$ for every $x\in W_{i}$.
    Hence repeated application of \autoref{lem:lifting-improved} gives that $\{u_{x}\}_{x\in W_i}$ is a $(U_{i}t,1-O(te^{-t^2/4})-it^{-100})$-cover for every $0\le i\le m$.
    Since $m=O(1)$, it follows that, for sufficiently large $n$,
    \begin{equation*}
        \Prb_{g}\left[\max_{x\in W_m}\ip{g}{u_x}\ge U_{m}t\right]
        \ge\frac{1}{2}.
    \end{equation*}

    On the other hand, $\norm{u_x}\le1$ and $|W_{m}|\le n$.
    Since $n=e^{H_0t^2/2}$, we have
    \begin{equation*}
        \Prb_{g}\left[\max_{x\in W_m}\ip{g}{u_x}\ge U_{m}t\right]
        \le n\bPhi(U_{m} t)
        \le e^{(H_0-U_m^2)t^2/2}.
    \end{equation*}
    Since $U_{m}^{2}>H=D_{\mathrm{a}}+2$ and $H-H_{0}=D_{\mathrm{a}}-D=3\eta>0$, the right-hand side is $o(1)$, a contradiction.

    Therefore the trial returns an independent set of size at least $T$ whenever the events above hold.
    By a union bound, these events hold with probability at least $\frac{3}{4}-o(1)\ge\frac{1}{2}$ in each trial, for sufficiently large $n$.
    Each trial takes polynomial time.
    Run at most $\lceil C\log_2{n}\rceil$ trials with independent randomness.
    The failure probability is at most $2^{-\lceil C\log_2{n}\rceil}\le n^{-C}$, and the total running time remains polynomial.
\end{proof}

We now prove the independent-set theorem stated in \autoref{sec:1.3}.

\improvedisthm*

\begin{proof}[Proof of \autoref{thm:improved-is}]
    For a graph with maximum degree at most $n^{\beta}$, \autoref{lem:improved-sparse-progress} gives an independent set of size $|I| \geq T = \frac{n\delta}{100\ell^{4}} = n^{\frac{1+\beta}{2} - o(1)}$ with probability at least $1 - n^{-C}$, when $n$ is sufficiently large. The ``$-o(1)$'' part allows us to ignore the case of small $n$, and \autoref{thm:improved-is} follows.
\end{proof}

\section{Conclusion}\label{sec:6}

In this paper, we studied the polynomial-time graph coloring algorithms for 3-colorable graphs. In \autoref{sec:3}, we presented a randomized algorithm that colors any 3-colorable graph using $O(n^{1/6+\varepsilon})$ colors for any $\varepsilon > 0$, using the Gaussian local-maximum selection algorithm. This result also yielded near-tight bounds on the chromatic number of vector $k$-colorable graphs (\autoref{sec:1.4}). In \autoref{sec:4}, we further improved the result to $O(n^{0.1537})$ colors using band-rounding algorithms and additional SDP constraints. This improves the previous best result of $O(n^{0.1954})$ colors due to Bansal, Huang, and Lee (2026).

However, a large gap still remains between the upper and lower bounds. On the lower bound side, the known hardness results concern only constant numbers of colors. Indeed, \cite{BKO19} proves the NP-hardness of coloring 3-colorable graphs with 5 colors, and the claimed result of a very recent preprint \cite{FMW26} implies the NP-hardness of coloring 3-colorable graphs with any constant number of colors. Even under the stronger claimed result, the remaining gap is substantial. Therefore, it is natural to pose the following conjecture.

\begin{conjecture}
    For every fixed $\varepsilon > 0$, there exists a polynomial-time algorithm to color every 3-colorable graph with $O(n^\varepsilon)$ colors.
\end{conjecture}
\section*{Relation to Recent Preprints}

As mentioned in \autoref{tab:intro}, we are aware of recent preprints \cite{NT26, MM26, Ana26, Son26} that claim to color 3-colorable graphs in polynomial time, using as few as $O(n^{0.1653})$ colors. However, our work was conducted independently of these preprints.

\section*{AI Usage Disclosure}

The authors extensively used OpenAI's large language models, GPT-6 Astra and GPT-6 Pro, during the proof exploration phase, in collaboration with the authors. Initially, GPT-6 Astra generated an $O(n^{0.1736})$-coloring algorithm. 
Then, we studied the analysis and explicitly suggested an idea to GPT-6 Pro, which developed it into an $O(n^{0.1712})$-coloring algorithm.
Building on this approach, GPT-6 Astra developed an $O(n^{1/6 + \varepsilon})$-coloring algorithm.

Subsequently, the authors generalized the strict vector 3-coloring algorithm to ordinary vector $k$-coloring for an arbitrary real number $k > 2$, and with these hints, GPT-6 Astra generated an $O(n^{0.16})$-coloring algorithm.
We iteratively simplified its analysis, and GPT-6 Astra improved the simplified algorithm.
Finally, GPT-6 Astra generated the final $O(n^{0.1537})$-color algorithm.

This paper is entirely written by the human authors. However, we used GPT-5.6 Sol, GPT-6 Sol, GPT-6 Astra, and Claude Opus 5.5 to polish the presentation of the paper. Finally, the authors take full responsibility for the correctness and content of the paper.

\printbibliography

@inproceedings{ACC06,
  author    = {Arora, Sanjeev and Chlamt{\'a}{\v{c}}, Eden and Charikar, Moses},
  title     = {New approximation guarantee for chromatic number},
  booktitle = {Proceedings of the 38th Annual ACM Symposium on Theory of Computing (STOC 2006)},
  pages     = {215--224},
  year      = {2006},
  publisher = {Association for Computing Machinery},
  doi       = {10.1145/1132516.1132548}
}

@article{AK97,
  author  = {Alon, Noga and Kahale, Nabil},
  title   = {A Spectral Technique for Coloring Random 3-Colorable Graphs},
  journal = {SIAM Journal on Computing},
  volume  = {26},
  number  = {6},
  pages   = {1733--1748},
  year    = {1997},
  doi     = {10.1137/S0097539794270248}
}

@misc{Ana26,
  author        = {Anand, Emile},
  title         = {Coloring 3-Colorable Graphs with {$O(n^{4/23})$} Colors via a {Gaussian}-Cover Recursion},
  year          = {2026},
  eprint        = {2610.01071},
  archivePrefix = {arXiv},
  primaryClass  = {cs.DS},
  url           = {https://arxiv.org/abs/2610.01071}
}

@article{BE05,
  author  = {Beigel, Richard and Eppstein, David},
  title   = {3-coloring in time {$O(1.3289^n)$}},
  journal = {Journal of Algorithms},
  volume  = {54},
  number  = {2},
  pages   = {168--204},
  year    = {2005},
  doi     = {10.1016/j.jalgor.2004.06.008}
}

@incollection{BE85,
  author    = {Bar-Yehuda, Reuven and Even, Shimon},
  title     = {A local-ratio theorem for approximating the weighted vertex cover problem},
  editor    = {G. Ausiello and M. Lucertini},
  booktitle = {Analysis and Design of Algorithms for Combinatorial Problems},
  series    = {Annals of Discrete Mathematics},
  volume    = {25},
  pages     = {27--45},
  year      = {1985},
  publisher = {North-Holland},
  doi       = {10.1016/S0304-0208(08)73101-3}
}

@article{BHK09,
  author  = {Bj{\"o}rklund, Andreas and Husfeldt, Thore and Koivisto, Mikko},
  title   = {Set partitioning via inclusion-exclusion},
  journal = {SIAM Journal on Computing},
  volume  = {39},
  number  = {2},
  pages   = {546--563},
  year    = {2009},
  doi     = {10.1137/070683933}
}

@inproceedings{BHL27,
  author        = {Bansal, Nikhil and Huang, Neng and Lee, Euiwoong},
  title         = {Improved {SDP}-Based Algorithm for Coloring 3-Colorable Graphs},
  booktitle     = {Proceedings of the 2027 Annual {ACM}-{SIAM} Symposium on Discrete Algorithms ({SODA})},
  year          = {2027},
  publisher     = {Society for Industrial and Applied Mathematics},
  eprint        = {2602.05904},
  archivePrefix = {arXiv},
  primaryClass  = {cs.DS},
  url           = {https://arxiv.org/abs/2602.05904},
  note          = {To appear}
}

@article{BK97,
  author  = {Blum, Avrim and Karger, David},
  title   = {An {$\widetilde{O}(n^{3/14})$}-Coloring Algorithm for 3-Colorable Graphs},
  journal = {Information Processing Letters},
  volume  = {61},
  number  = {1},
  pages   = {49--53},
  year    = {1997},
  doi     = {10.1016/S0020-0190(96)00190-1}
}

@inproceedings{BKO19,
  author    = {Bul{\'\i}n, Jakub and Krokhin, Andrei and Opr{\v{s}}al, Jakub},
  title     = {Algebraic approach to promise constraint satisfaction},
  booktitle = {Proceedings of the 51st Annual ACM SIGACT Symposium on Theory of Computing (STOC 2019)},
  pages     = {602--613},
  year      = {2019},
  publisher = {Association for Computing Machinery},
  doi       = {10.1145/3313276.3316300}
}

@inproceedings{Blu89,
  author    = {Blum, Avrim},
  title     = {An {$\widetilde{O}(n^{0.4})$}-Approximation Algorithm for 3-Coloring (and Improved Approximation Algorithms for {$k$}-Coloring)},
  booktitle = {Proceedings of the 21st Annual ACM Symposium on Theory of Computing (STOC 1989)},
  pages     = {535--542},
  year      = {1989},
  publisher = {Association for Computing Machinery},
  doi       = {10.1145/73007.73058}
}

@inproceedings{Blu90,
  author    = {Blum, Avrim},
  title     = {Some tools for approximate 3-coloring},
  booktitle = {Proceedings of the 31st Annual IEEE Symposium on Foundations of Computer Science (FOCS 1990)},
  pages     = {554--562},
  year      = {1990},
  publisher = {IEEE},
  doi       = {10.1109/FSCS.1990.89576}
}

@article{Blu94,
  author  = {Blum, Avrim},
  title   = {New approximation algorithms for graph coloring},
  journal = {Journal of the ACM},
  volume  = {41},
  number  = {3},
  pages   = {470--516},
  year    = {1994},
  doi     = {10.1145/176584.176586}
}

@article{BR90,
  author  = {Berger, Bonnie and Rompel, John},
  title   = {A better performance guarantee for approximate graph coloring},
  journal = {Algorithmica},
  volume  = {5},
  pages   = {459--466},
  year    = {1990},
  doi     = {10.1007/BF01840398}
}

@article{CDP15,
  author  = {Chen, Wei-Kuo and Dafnis, Nikos and Paouris, Grigoris},
  title   = {Improved {H}{\"o}lder and Reverse {H}{\"o}lder Inequalities for {Gaussian} Random Vectors},
  journal = {Advances in Mathematics},
  volume  = {280},
  pages   = {643--689},
  year    = {2015},
  doi     = {10.1016/j.aim.2014.09.029}
}

@book{CFK+15,
  author    = {Cygan, Marek and Fomin, Fedor V. and Kowalik, {\L}ukasz and Lokshtanov, Daniel and Marx, D{\'a}niel and Pilipczuk, Marcin and Pilipczuk, Micha{\l} and Saurabh, Saket},
  title     = {Parameterized Algorithms},
  year      = {2015},
  publisher = {Springer},
  doi       = {10.1007/978-3-319-21275-3}
}

@inproceedings{Chl07,
  author    = {Chlamt{\'a}{\v{c}}, Eden},
  title     = {Approximation algorithms using hierarchies of semidefinite programming relaxations},
  booktitle = {Proceedings of the 48th Annual IEEE Symposium on Foundations of Computer Science (FOCS 2007)},
  pages     = {691--701},
  year      = {2007},
  publisher = {IEEE},
  doi       = {10.1109/FOCS.2007.4389537}
}

@article{CH14,
  author  = {Chlamt{\'a}{\v{c}}, Eden and Haviv, Ishay},
  title   = {Linear Index Coding via Semidefinite Programming},
  journal = {Combinatorics, Probability and Computing},
  volume  = {23},
  number  = {2},
  pages   = {223--247},
  year    = {2014},
  doi     = {10.1017/S0963548313000564}
}

@article{DMR09,
  author  = {Dinur, Irit and Mossel, Elchanan and Regev, Oded},
  title   = {Conditional Hardness for Approximate Coloring},
  journal = {SIAM Journal on Computing},
  volume  = {39},
  number  = {3},
  pages   = {843--873},
  year    = {2009},
  doi     = {10.1137/07068062X}
}

@article{Ehr83,
  author  = {Ehrhard, Antoine},
  title   = {Sym{\'e}trisation dans l'espace de Gauss},
  journal = {Mathematica Scandinavica},
  volume  = {53},
  pages   = {281--301},
  year    = {1983},
  doi     = {10.7146/math.scand.a-12035}
}

@article{FK98,
  author  = {Feige, Uriel and Kilian, Joe},
  title   = {Zero Knowledge and the Chromatic Number},
  journal = {Journal of Computer and System Sciences},
  volume  = {57},
  number  = {2},
  pages   = {187--199},
  year    = {1998},
  doi     = {10.1006/jcss.1998.1587}
}

@article{FLS04,
  author  = {Feige, Uriel and Langberg, Michael and Schechtman, Gideon},
  title   = {Graphs with tiny vector chromatic numbers and huge chromatic numbers},
  journal = {SIAM Journal on Computing},
  volume  = {33},
  number  = {6},
  pages   = {1338--1368},
  year    = {2004},
  doi     = {10.1137/S0097539703431391}
}

@misc{FMW26,
  author        = {Fei, Yumou and Minzer, Dor and Wang, Shuo},
  title         = {On the Hardness of 4-to-1 Games with Perfect Completeness},
  year          = {2026},
  eprint        = {2610.12378},
  archivePrefix = {arXiv},
  primaryClass  = {cs.CC},
  url           = {https://arxiv.org/abs/2610.12378}
}

@inproceedings{GJS74,
  author    = {Garey, Michael R. and Johnson, David S. and Stockmeyer, Larry},
  title     = {Some Simplified {NP}-Complete Problems},
  booktitle = {Proceedings of the 6th Annual ACM Symposium on Theory of Computing (STOC 1974)},
  pages     = {47--63},
  year      = {1974},
  publisher = {Association for Computing Machinery},
  doi       = {10.1145/800119.803884}
}

@article{GK04,
  author  = {Guruswami, Venkatesan and Khanna, Sanjeev},
  title   = {On the hardness of 4-coloring a 3-colorable graph},
  journal = {SIAM Journal on Discrete Mathematics},
  volume  = {18},
  number  = {1},
  pages   = {30--40},
  year    = {2004},
  doi     = {10.1137/S0895480100376794}
}

@article{GLS81,
  author  = {Gr{\"o}tschel, Martin and Lov{\'a}sz, L{\'a}szl{\'o} and Schrijver, Alexander},
  title   = {The ellipsoid method and its consequences in combinatorial optimization},
  journal = {Combinatorica},
  volume  = {1},
  number  = {2},
  pages   = {169--197},
  year    = {1981},
  doi     = {10.1007/BF02579273}
}

@incollection{GLS84,
  author    = {Gr{\"o}tschel, Martin and Lov{\'a}sz, L{\'a}szl{\'o} and Schrijver, Alexander},
  title     = {Polynomial algorithms for perfect graphs},
  editor    = {Berge, Claude and Chv{\'a}tal, V{\'a}clav},
  booktitle = {Topics on Perfect Graphs},
  series    = {Annals of Discrete Mathematics},
  volume    = {21},
  pages     = {325--356},
  year      = {1984},
  publisher = {North-Holland},
  doi       = {10.1016/S0304-0208(08)72943-8}
}

@book{GLS88,
  author    = {Gr{\"o}tschel, Martin and Lov{\'a}sz, L{\'a}szl{\'o} and Schrijver, Alexander},
  title     = {Geometric algorithms and combinatorial optimization},
  series    = {Algorithms and Combinatorics},
  volume    = {2},
  year      = {1988},
  publisher = {Springer},
  doi       = {10.1007/978-3-642-97881-4}
}

@inproceedings{GS20,
  author    = {Guruswami, Venkatesan and Sandeep, Sai},
  title     = {{{$d$}-to-1} Hardness of Coloring 3-Colorable Graphs with {$O(1)$} Colors},
  booktitle = {47th International Colloquium on Automata, Languages, and Programming (ICALP 2020)},
  series    = {Leibniz International Proceedings in Informatics (LIPIcs)},
  volume    = {168},
  pages     = {62:1--62:12},
  year      = {2020},
  publisher = {Schloss Dagstuhl -- Leibniz-Zentrum f{\"u}r Informatik},
  doi       = {10.4230/LIPIcs.ICALP.2020.62}
}

@article{GW95,
  author  = {Goemans, Michel X. and Williamson, David P.},
  title   = {Improved approximation algorithms for maximum cut and satisfiability problems using semidefinite programming},
  journal = {Journal of the ACM},
  volume  = {42},
  number  = {6},
  pages   = {1115--1145},
  year    = {1995},
  doi     = {10.1145/227683.227684}
}

@article{Has99,
  author  = {H{\aa}stad, Johan},
  title   = {Clique Is Hard to Approximate within {$n^{1-\varepsilon}$}},
  journal = {Acta Mathematica},
  volume  = {182},
  number  = {1},
  pages   = {105--142},
  year    = {1999},
  doi     = {10.1007/BF02392825}
}

@article{Hal93,
  author  = {Halld{\'o}rsson, Magn{\'u}s M.},
  title   = {A still better performance guarantee for approximate graph coloring},
  journal = {Information Processing Letters},
  volume  = {45},
  number  = {1},
  pages   = {19--23},
  year    = {1993},
  doi     = {10.1016/0020-0190(93)90246-6}
}

@article{HNZ02,
  author  = {Halperin, Eran and Nathaniel, Ram and Zwick, Uri},
  title   = {Coloring {$k$}-colorable graphs using relatively small palettes},
  journal = {Journal of Algorithms},
  volume  = {45},
  number  = {1},
  pages   = {72--90},
  year    = {2002},
  doi     = {10.1016/S0196-6774(02)00217-1}
}

@inproceedings{Hsi26,
  author    = {Hsieh, Jun-Ting},
  title     = {Coloring 3-colorable graphs with low threshold rank},
  booktitle = {Proceedings of the 2026 Annual ACM-SIAM Symposium on Discrete Algorithms (SODA)},
  pages     = {718--726},
  year      = {2026},
  publisher = {Society for Industrial and Applied Mathematics},
  doi       = {10.1137/1.9781611978971.27}
}

@incollection{Kar72,
  author    = {Karp, Richard M.},
  title     = {Reducibility among combinatorial problems},
  editor    = {Miller, Raymond E. and Thatcher, James W.},
  booktitle = {Complexity of Computer Computations},
  pages     = {85--103},
  year      = {1972},
  publisher = {Plenum Press},
  doi       = {10.1007/978-1-4684-2001-2_9}
}

@inproceedings{Kho02,
  author    = {Khot, Subhash},
  title     = {On the power of unique 2-prover 1-round games},
  booktitle = {Proceedings of the 34th Annual ACM Symposium on Theory of Computing (STOC 2002)},
  pages     = {767--775},
  year      = {2002},
  publisher = {Association for Computing Machinery},
  doi       = {10.1145/509907.510017}
}

@article{KLS00,
  author  = {Khanna, Sanjeev and Linial, Nathan and Safra, Shmuel},
  title   = {On the hardness of approximating the chromatic number},
  journal = {Combinatorica},
  volume  = {20},
  number  = {3},
  pages   = {393--415},
  year    = {2000},
  doi     = {10.1007/s004930070013}
}

@inproceedings{KMS94,
  author    = {Karger, David and Motwani, Rajeev and Sudan, Madhu},
  title     = {Approximate Graph Coloring by Semidefinite Programming},
  booktitle = {Proceedings of the 35th Annual IEEE Symposium on Foundations of Computer Science (FOCS 1994)},
  pages     = {2--13},
  year      = {1994},
  publisher = {IEEE},
  doi       = {10.1109/SFCS.1994.365710}
}

@article{KMS98,
  author  = {Karger, David and Motwani, Rajeev and Sudan, Madhu},
  title   = {Approximate graph coloring by semidefinite programming},
  journal = {Journal of the ACM},
  volume  = {45},
  number  = {2},
  pages   = {246--265},
  year    = {1998},
  doi     = {10.1145/274787.274791}
}

@inproceedings{KT12,
  author    = {Kawarabayashi, Ken-ichi and Thorup, Mikkel},
  title     = {Combinatorial coloring of 3-colorable graphs},
  booktitle = {Proceedings of the 53rd Annual IEEE Symposium on Foundations of Computer Science (FOCS 2012)},
  pages     = {68--75},
  year      = {2012},
  publisher = {IEEE},
  doi       = {10.1109/FOCS.2012.16}
}

@inproceedings{KT14,
  author    = {Kawarabayashi, Ken-ichi and Thorup, Mikkel},
  title     = {Coloring 3-colorable graphs with {$o(n^{1/5})$} colors},
  booktitle = {31st International Symposium on Theoretical Aspects of Computer Science (STACS 2014)},
  series    = {Leibniz International Proceedings in Informatics (LIPIcs)},
  volume    = {25},
  pages     = {458--469},
  year      = {2014},
  publisher = {Schloss Dagstuhl -- Leibniz-Zentrum f{\"u}r Informatik},
  doi       = {10.4230/LIPIcs.STACS.2014.458}
}

@article{KT17,
  author  = {Kawarabayashi, Ken-ichi and Thorup, Mikkel},
  title   = {Coloring 3-Colorable Graphs with Less than {$n^{1/5}$} Colors},
  journal = {Journal of the ACM},
  volume  = {64},
  number  = {1},
  pages   = {4:1--4:23},
  year    = {2017},
  doi     = {10.1145/3001582}
}

@inproceedings{KTY24,
  author    = {Kawarabayashi, Ken-ichi and Thorup, Mikkel and Yoneda, Hirotaka},
  title     = {Better coloring of 3-colorable graphs},
  booktitle = {Proceedings of the 56th Annual ACM Symposium on Theory of Computing (STOC 2024)},
  pages     = {331--339},
  year      = {2024},
  publisher = {Association for Computing Machinery},
  doi       = {10.1145/3618260.3649768}
}

@inproceedings{Las01,
  author    = {Lasserre, Jean B.},
  title     = {An explicit exact {SDP} relaxation for nonlinear 0-1 programs},
  booktitle = {Integer Programming and Combinatorial Optimization (IPCO 2001)},
  series    = {Lecture Notes in Computer Science},
  volume    = {2081},
  pages     = {293--303},
  year      = {2001},
  publisher = {Springer},
  doi       = {10.1007/3-540-45535-3_23}
}

@article{Lov79,
  author  = {Lov{\'a}sz, L{\'a}szl{\'o}},
  title   = {On the {Shannon} Capacity of a Graph},
  journal = {IEEE Transactions on Information Theory},
  volume  = {25},
  number  = {1},
  pages   = {1--7},
  year    = {1979},
  doi     = {10.1109/TIT.1979.1055985}
}

@misc{MM26,
  author        = {Makarychev, Konstantin and Makarychev, Yury},
  title         = {Sharper {Gaussian} Covers in the {Bansal--Huang--Lee} Coloring Framework},
  year          = {2026},
  eprint        = {2609.40136},
  archivePrefix = {arXiv},
  primaryClass  = {cs.DS},
  url           = {https://arxiv.org/abs/2609.40136}
}

@misc{NT26,
  author        = {Narang, Ijay and Tang, Yukai},
  title         = {Improved {SDP} Coloring of 3-Colorable Graphs from Recursive {Gaussian} Certificates},
  year          = {2026},
  eprint        = {2609.33684},
  archivePrefix = {arXiv},
  primaryClass  = {cs.DS},
  url           = {https://arxiv.org/abs/2609.33684}
}

@misc{Pra26,
  author        = {Pratt, Kevin},
  title         = {Breaking the {$2^n$} Barrier for Graph {$k$}-Coloring},
  year          = {2026},
  eprint        = {2607.27159},
  archivePrefix = {arXiv},
  primaryClass  = {cs.DS},
  url           = {https://arxiv.org/abs/2607.27159}
}

@misc{Son26,
  author = {Song, Zhao},
  title  = {An {$O(n^{1/6})$}-coloring algorithm for 3-colorable graphs},
  year   = {2026},
  doi    = {10.2139/ssrn.7551399},
  url    = {https://ssrn.com/abstract=7551399},
  note   = {SSRN preprint 7551399}
}

@inproceedings{WGJ+24,
  author    = {Wu, Pu and Gu, Huanyu and Jiang, Huiqin and Shao, Zehui and Xu, Jin},
  title     = {A faster algorithm for the 4-coloring problem},
  booktitle = {32nd Annual European Symposium on Algorithms (ESA 2024)},
  series    = {Leibniz International Proceedings in Informatics (LIPIcs)},
  volume    = {308},
  pages     = {103:1--103:18},
  year      = {2024},
  publisher = {Schloss Dagstuhl -- Leibniz-Zentrum f{\"u}r Informatik},
  doi       = {10.4230/LIPIcs.ESA.2024.103}
}

@inproceedings{Wig82,
  author    = {Wigderson, Avi},
  title     = {A new approximate graph coloring algorithm},
  booktitle = {Proceedings of the 14th Annual ACM Symposium on Theory of Computing (STOC 1982)},
  pages     = {325--329},
  year      = {1982},
  publisher = {Association for Computing Machinery},
  doi       = {10.1145/800070.802207}
}

@inproceedings{Zam21,
  author    = {Zamir, Or},
  title     = {Breaking the {$2^n$} barrier for 5-coloring and 6-coloring},
  booktitle = {48th International Colloquium on Automata, Languages, and Programming (ICALP 2021)},
  series    = {Leibniz International Proceedings in Informatics (LIPIcs)},
  volume    = {198},
  pages     = {113:1--113:20},
  year      = {2021},
  publisher = {Schloss Dagstuhl -- Leibniz-Zentrum f{\"u}r Informatik},
  doi       = {10.4230/LIPIcs.ICALP.2021.113}
}

@misc{Zam26,
  author        = {Zamir, Or},
  title         = {{$k$}-coloring is faster than computing the chromatic number},
  year          = {2026},
  eprint        = {2607.25973},
  archivePrefix = {arXiv},
  primaryClass  = {cs.DS},
  url           = {https://arxiv.org/abs/2607.25973}
}

@inproceedings{Zuc06,
  author    = {Zuckerman, David},
  title     = {Linear degree extractors and the inapproximability of max clique and chromatic number},
  booktitle = {Proceedings of the 38th Annual ACM Symposium on Theory of Computing (STOC 2006)},
  pages     = {681--690},
  year      = {2006},
  publisher = {Association for Computing Machinery},
  doi       = {10.1145/1132516.1132612}
}
\appendix
\makeatletter
\renewcommand{\subsection}{\@startsection{appsubsection}{2}{\z@}%
  {-3.25ex\@plus -1ex \@minus -.2ex}{1.5ex \@plus .2ex}{\normalfont\large\bfseries}}
\makeatother

\newpage

\section{Extension to Ordinary Vector \texorpdfstring{$k$}{k}-Colorings}
\label{sec:A}

We prove \autoref{thm:vector-k-coloring} by extending the proof of \autoref{thm:vector-3-coloring} to ordinary vector $k$-colorings for every fixed real number $k>2$.

Each lemma below marked \emph{Replaces} states the replacement of one lemma or corollary of the main text; the remaining lemmas are auxiliary.
The proof of \autoref{thm:vector-3-coloring} goes through verbatim once these replacements are made.
\autoref{alg:local-maxima}, \autoref{alg:sparse}, \autoref{lem:deleted-bound} and \autoref{lem:estimates} are unchanged together with their proofs; \autoref{lem:core-nonempty} and \autoref{lem:sparse-progress-one} keep their statements, and their proofs go through with the parameters and lemmas of this appendix.

Throughout this appendix, set $\lambda=1/(k-1)\in(0,1)$ and $c=\sqrt{k/(k-2)}$; for $k=3$ we have $\lambda=\frac{1}{2}$ and $c=\sqrt{3}$.
Nothing in this appendix uses that $k$ is an integer: every identity below is an algebraic identity in the real parameter $k>2$.
Recall that an ordinary vector $k$-coloring of $G=(V,E)$ is a family of unit vectors $\{v_x\}_{x\in V}$ with $\ip{v_x}{v_y}\le-\lambda$ for every $xy\in E$.

\subsection{Edge Residuals}\label{sec:A.1}

For every vertex $x\in V$ and $y\in N_V(x)$, write $\lambda_{xy}=-\ip{v_x}{v_y}$, so that $\lambda_{xy}\ge\lambda$, and redefine
\begin{equation*}
    v_{xy}
    =\frac{c}{1+\lambda_{xy}}(v_y+\lambda_{xy}v_x).
\end{equation*}
For $k=3$ and $\lambda_{xy}=\frac{1}{2}$, this agrees with the vector $v_{xy}$ of \autoref{sec:2.2}.
For a Gaussian vector $g$, write $g=Zv_x+g'$ with $Z=\ip{g}{v_x}$ and $g'\perp v_x$.
This normalization keeps the local-maximum condition in the common form $cZ>\ip{g'}{v_{xy}}$ for every edge (\autoref{lem:q-mf-bound}), although $\ip{v_x}{v_y}=-\lambda_{xy}$ may vary from edge to edge.
In exchange, $v_{xy}$ is no longer a unit vector, but it satisfies $\norm{v_{xy}} \le 1$ as seen below.
This upper bound suffices, since each $\ip{g'}{v_{xy}}$ is then a centered Gaussian with variance at most $1$.

\begin{lemma}
    \label{lem:q-residual} For every $x\in V$ and $y\in N_V(x)$, we have $v_{xy}\perp v_x$, $\norm{v_{xy}}\le1$, and
    \begin{equation*}
        v_y
        =-\lambda_{xy}v_x+\frac{1+\lambda_{xy}}{c}v_{xy}.
    \end{equation*}
\end{lemma}

\begin{proof}
    From the definition of $v_{xy}$, we have $v_y=-\lambda_{xy}v_x+\frac{1+\lambda_{xy}}{c}v_{xy}$.
    Since $\ip{v_y+\lambda_{xy}v_x}{v_x}=\ip{v_y}{v_x}+\lambda_{xy}\norm{v_x}^2=-\lambda_{xy}+\lambda_{xy}=0$, we have $v_y+\lambda_{xy}v_x\perp v_x$.

    Moreover, $\norm{v_y+\lambda_{xy}v_x}^2=1-\lambda_{xy}^2$, so
    \begin{equation*}
        \norm{v_{xy}}^2
        =c^2\frac{1-\lambda_{xy}}{1+\lambda_{xy}}
        \le c^2\frac{1-\lambda}{1+\lambda}
        =1,
    \end{equation*}
    because the middle expression is decreasing in $\lambda_{xy}$ and $\lambda_{xy}\ge\lambda$, and $c^2=\frac{k}{k-2}=\frac{1+\lambda}{1-\lambda}$.
\end{proof}

\subsection{Local-Maximum Probabilities}\label{sec:A.2}

\begin{lemma}[Replaces \autoref{lem:mf-bound}]
    \label{lem:q-mf-bound} Let $F\subseteq V$ and $x\in F$.
    Then, for every $a\in\R$,
    \begin{align*}
        p_F(x)
        &=\E_g[\bPhi(M_F(x)/c)],\\
        \Prb_g[M_F(x)\le a]
        &\le\frac{p_F(x)}{\bPhi(a/c)}.
    \end{align*}
\end{lemma}

\begin{proof}
    Write $g=Zv_x+g'$ with $Z\sim N(0,1)$ independent of $g'\perp v_x$.
    By \autoref{lem:q-residual}, $\ip{g}{v_{xy}}=\ip{g'}{v_{xy}}$, and for every $y\in N_F(x)$,
    \begin{align*}
        \ip{g}{v_x}>\ip{g}{v_y}
        &\iff Z>-\lambda_{xy}Z+\frac{1+\lambda_{xy}}{c}\ip{g'}{v_{xy}}\\
        &\iff cZ>\ip{g'}{v_{xy}}.
    \end{align*}
    Thus, $x$ is a local maximum exactly when $cZ>M_F(x)$.
    The rest of the proof of \autoref{lem:mf-bound} applies with $\sqrt{3}$ replaced by $c$.
\end{proof}

We prove a generalization of \autoref{lem:low-degree}, which handles the low-degree regime.
\begin{lemma}[Replaces \autoref{lem:low-degree}]
    \label{lem:q-low-degree} Let $0<\beta\le\frac{k-2}{k}$, and suppose that every vertex of $G$ has degree at most $n^{\beta}$.
    Then there is a randomized polynomial-time algorithm that, with probability at least $\frac{3}{4}$, returns an independent set $I$ with $|I|\ge n^{1-\frac{k-2}{k}\beta-o(1)}$.
\end{lemma}

\begin{proof}
    Let $a=\sqrt{2\log{(2n^{\beta})}}$ and $p_*=\frac{1}{2}\bPhi(a/c)$, and run one trial of \autoref{alg:sparse} with $\delta=\frac{1}{2}$, $T=\lfloor np_*/2\rfloor$ and $R=\lceil2\log{4}/p_*\rceil$, as in the proof of \autoref{lem:low-degree}.
    Fix $x\in V$.
    For each $y\in N_V(x)$, we have $\norm{v_{xy}}\le1$ by \autoref{lem:q-residual}, so $\Prb_g[\ip{g}{v_{xy}}>a] \le \Prb_g[\ip{g}{v_{xy}}>a\norm{v_{xy}}] \le \bPhi(a)$.
    Since $x$ has at most $n^{\beta}$ neighbors, a union bound and the upper bound of \autoref{fact:gaussian-tail} give
    \begin{equation*}
        \Prb_g[M_V(x)>a]
        \le n^{\beta}\bPhi(a)
        \le n^{\beta}e^{-a^2/2}
        =\frac{1}{2}.
    \end{equation*}
    Hence, by \autoref{lem:q-mf-bound} with $F=V$, $p_V(x)\ge\Prb_g[M_V(x)\le a]\,\bPhi(a/c)\ge p_*$.
    Since $a^2/(2c^2)=\frac{k-2}{k}\log{(2n^{\beta})}$, the lower bound of \autoref{fact:gaussian-tail} gives $p_* = \frac{1}{2}\bPhi(a/c) \ge \frac{c_0}{2(1 + a/c)}2^{-\frac{k-2}{k}}n^{-\frac{k-2}{k}\beta} \ge n^{-\frac{k-2}{k}\beta-o(1)}$.
    The rest of the proof is the same as in \autoref{lem:low-degree}.
\end{proof}

\subsection{The High-Degree Regime}\label{sec:A.3}

For the rest of this appendix, assume $\frac{k-2}{k}<\beta\le1$.
\autoref{tab:params-A} lists the parameters of this subsection that differ from \autoref{tab:params-3}; all other parameters, in particular $T$, $R$, $t$, $H_0$, $D$, $w$ and $H$, are given by the formulas of \autoref{tab:params-3}.

\begin{table}[!htbp]
    \centering
    \small
    \renewcommand{\arraystretch}{1.15}
    \setlength{\tabcolsep}{3pt}
    \caption{Parameters of \autoref{sec:A.3} that differ from \autoref{tab:params-3}.}
    \label{tab:params-A}
    \begin{tabular}{llll}
        \hline
        Symbol & \autoref{sec:A} & \autoref{sec:3} ($k=3$) & Role \\
        \hline
        $\delta$ & $n^{-\frac{(k-2)\beta}{2+k\beta}}\ell^{-C_{\ell}}$ & $n^{-\beta/(2+3\beta)}\ell^{-500f(\beta)}$ & deletion threshold \\
        $C_{\ell}$ & large; depends on $k,\beta$ & $500f(\beta)$ & exponent of $\ell$ in $\delta$ \\
        $U_0$ & $(1+\lambda)/\sqrt{2}$ & $3/(2\sqrt{2})$ & initial Gaussian cover threshold divided by $t$ \\
        $\nu$ & $\sqrt{\frac{(k-2)D-k}{2(k-1)D}}$ & $\frac{1}{2}\sqrt{(D-3)/D}$ & breakpoint of $\psi_k$ \\
        $\bar{\theta}$ & $\frac{\sqrt{H}/\nu-1}{D+1}$ & $20f(\beta)$ & upper bound on $\theta$ in \autoref{lem:q-exact-parameters} \\
        \hline
    \end{tabular}%

\end{table}

Here $C_{\ell}=C_{\ell}(k,\beta)$ is a sufficiently large constant fixed in the proof of \autoref{lem:q-U-increase}.
Since $\frac{(k-2)\beta}{2+k\beta}<1$ and $\ell=\lceil\log{n}\rceil$, we have $2\ell/\delta=2n^{\frac{(k-2)\beta}{2+k\beta}}\ell^{C_{\ell}+1}=o(n)$, so the proof of \autoref{lem:core-nonempty} applies.

\begin{lemma}[Replaces \autoref{cor:initial-cover}]
    \label{lem:q-initial-cover} Suppose that $p_{F_{\ell-1}}(r)\le e^{-t^2/2}$.
    Then $\{u_y\}_{y\in W_0}$ is a $(U_0t,1-O(te^{-t^2/4}))$-cover.
\end{lemma}

\begin{proof}
    Applying \autoref{lem:q-mf-bound} with $F=F_{\ell-1}$, $x=r$ and $a=ct/\sqrt{2}$ gives
    \begin{equation*}
        \Prb_g[M_{F_{\ell-1}}(r)\le ct/\sqrt{2}]
        \le\frac{e^{-t^2/2}}{\bPhi(t/\sqrt{2})}
        =O(te^{-t^2/4}).
    \end{equation*}
    Hence, with probability at least $1-O(te^{-t^2/4})$, some $y\in W_0$ satisfies $\ip{g}{v_{ry}}\ge ct/\sqrt{2}$.
    By \autoref{lem:q-residual} with $x=r$ and $v_{ry}\perp v_r$, we have $u_y=\frac{1+\lambda_{ry}}{c}v_{ry}$, and since $\lambda_{ry}\ge\lambda$, we get $\ip{g}{u_y}\ge\frac{1+\lambda}{c}\cdot\frac{ct}{\sqrt{2}}=U_0t$.
\end{proof}

The following lemma replaces the identity for $\ip{g}{u_y}$ preceding \autoref{lem:wy-transverse-bound}.
For $x\in F_j$, write $\rho=\rho_x$, so that $\norm{u_x}=\sqrt{1-\rho^2}$, and assume $\norm{u_x}>0$.
For $y\in N_{F_{j-1}}(x)$, decompose $v_{xy}$ as in \autoref{eq:vxy-decomposition}, writing the unit vector $v_{xr}$ of \autoref{sec:2.2} explicitly, since for $rx\in E$ it may differ from the edge residual defined above:
\begin{equation*}
    v_{xy}
    =\xi_y\frac{v_r-\rho v_x}{\norm{u_x}}+w_y
    =\xi_y\norm{u_x}v_r-\frac{\xi_y\rho u_x}{\norm{u_x}}+w_y,
\end{equation*}
where $w_y\perp\Span(v_r,v_x)$; by \autoref{lem:q-residual}, $(\xi_y)^2+\norm{w_y}^2\le1$.

\begin{lemma}
    \label{lem:q-child} Suppose that $\ip{g}{u_x/\norm{u_x}}=-Lt$, and write $U=\norm{u_x}L$.
    Let $B\ge-cU$, and suppose that $\ip{g}{w_y}+L\rho\xi_yt\ge Bt$.
    Then $\ip{g}{u_y}\ge\left(\lambda U+\frac{1+\lambda}{c}B\right)t$.
\end{lemma}

\begin{proof}
    Projecting $v_y=-\lambda_{xy}v_x+\frac{1+\lambda_{xy}}{c}v_{xy}$ of \autoref{lem:q-residual} onto $v_r^{\perp}$, we have
    \begin{equation*}
        u_y
        =-\lambda_{xy}u_x+\frac{1+\lambda_{xy}}{c}\left(-\frac{\xi_y\rho u_x}{\norm{u_x}}+w_y\right).
    \end{equation*}
    Taking the inner product with $g$ and using $\ip{g}{u_x}=-\norm{u_x}Lt=-Ut$,
    \begin{equation*}
        \ip{g}{u_y}
        =\lambda_{xy}Ut+\frac{1+\lambda_{xy}}{c}\left(\ip{g}{w_y}+L\rho\xi_yt\right)
        \ge\left(\lambda_{xy}U+\frac{1+\lambda_{xy}}{c}B\right)t.
    \end{equation*}
    Finally,
    \begin{equation*}
        \lambda_{xy}U+\frac{1+\lambda_{xy}}{c}B
        =\lambda U+\frac{1+\lambda}{c}B+(\lambda_{xy}-\lambda)\left(U+\frac{B}{c}\right)
        \ge\lambda U+\frac{1+\lambda}{c}B,
    \end{equation*}
    since $\lambda_{xy}\ge\lambda$ and $B\ge-cU$.
\end{proof}

For $k=3$ and $\lambda_{xy}=\frac{1}{2}$, the second display is the identity for $\ip{g}{u_y}$ in \autoref{sec:3.4}.
Thus, by \autoref{lem:q-child}, an edge with $\lambda_{xy}>\lambda$ yields at least the lower bound on $\ip{g}{u_y}$ obtained for $\lambda_{xy}=\lambda$, so the remaining lemmas only need to treat the case $\lambda_{xy}=\lambda$.

With the parameters of \autoref{tab:params-A}, we have $t^2=\frac{2(k-2)\beta}{2+k\beta}\log{n}+2C_{\ell}\log{\ell}$.
Then, as $n\to\infty$,
\begin{align*}
    D
    &\to\frac{2+k\beta}{k-2}
    >\frac{k}{k-2}
    =c^2,\\
    H
    &\to\frac{2+k\beta}{(k-2)\beta}
    >1,
\end{align*}
where the first inequality uses $\beta>\frac{k-2}{k}$.
Hence, $D>c^2$ and $H>1$ for sufficiently large $n$.
Since $D>c^2$, the parameter $\nu$ of \autoref{tab:params-A} is well defined; let
\begin{equation*}
    \psi_k(s)
    =
    \begin{cases}
    c\sqrt{1-s^2}, & s\le\nu,\\[2mm]
    \displaystyle\frac{c(1-\nu s)}{\sqrt{1-\nu^2}}, & s\ge\nu.
    \end{cases}
\end{equation*}
Thus $\psi_k$ agrees with $s\mapsto c\sqrt{1-s^2}$ on $[0,\nu]$ and with its tangent line at $s=\nu$ on $[\nu,\infty)$.
For $k=3$, these are the definitions of $\nu$ and $\psi$ in \autoref{sec:3.4}.
Since $D$ tends to a constant larger than $c^2$, we have $\nu=\Omega(1)$, and $\nu^2<\frac{k-2}{2(k-1)}\le\frac{1}{2}$.

\begin{lemma}[Replaces \autoref{lem:exact-parameters}]
    \label{lem:q-exact-parameters} Let $0<L<\sqrt{H}$, $\rho\in[-1,1]$ and $s=\sqrt{H-L^2(1-\rho^2)}$.
    For sufficiently large $n$, there exist $a_0\ge\sqrt{1-\nu^2}$, $\kappa\ge0$, $0\le\theta\le\bar{\theta}$, $|b|\le\nu$ and $0\le \mu\le\nu$ such that
    \begin{align*}
        a_0^2+b^2+\mu^2
        &=1,\\
        \mu+\theta(\mu+\kappa)
        &=\sqrt{H-L^2},
    \end{align*}
    and
    \begin{equation*}
        (1+\theta)ca_0-\sqrt{\bigl(L\rho-(1+\theta)b\bigr)^2+\theta^2(D+\kappa^2)}
        =\psi_k(s).
    \end{equation*}
\end{lemma}

\begin{proof}
    Since $s=\sqrt{H-L^2(1-\rho^2)}$, we have $L^2\rho^2+(H-L^2)=s^2$.

    First, suppose that $s\le\nu$.
    Choose
    \begin{align*}
        a_0&=\sqrt{1-s^2}, &b&=L\rho, &\mu&=\sqrt{H-L^2}, &\kappa&=\theta=0.
    \end{align*}
    By $L^2\rho^2+(H-L^2)=s^2$, we have $a_0^2+b^2+\mu^2=1$, $|b|,\mu\le s\le\nu$ and $a_0\ge\sqrt{1-\nu^2}$.
    Also, $\mu+\theta(\mu+\kappa)=\mu=\sqrt{H-L^2}$.
    Since $\theta=0$ and $b=L\rho$,
    \begin{equation*}
        (1+\theta)ca_0-\sqrt{\bigl(L\rho-(1+\theta)b\bigr)^2+\theta^2(D+\kappa^2)}
        =c\sqrt{1-s^2}
        =\psi_k(s).
    \end{equation*}

    Now suppose that $s\ge\nu$.
    Choose
    \begin{align*}
        a_0&=\sqrt{1-\nu^2}, &b&=\frac{\nu L\rho}{s}, &\mu&=\frac{\nu\sqrt{H-L^2}}{s},\\
        \kappa&=D\mu, &\theta&=\frac{s/\nu-1}{D+1}.
    \end{align*}
    As in the proof of \autoref{lem:exact-parameters}, $a_0^2+b^2+\mu^2=1$, $\mu+\theta(\mu+\kappa)=\sqrt{H-L^2}$ and $|b|,\mu\le\nu$.
    Since $s\le\sqrt{H}$, we also have $0\le\theta\le\bar{\theta}$.
    We have $D(1-\nu^2)=c^2(1+D\nu^2)$, since multiplying by $k-2$ turns it into $(k-2)D-k=2(k-1)D\nu^2$, which is the definition of $\nu$.
    Hence, since $L\rho-(1+\theta)b=D\theta b$, $\kappa=D\mu$ and $b^2+\mu^2=\nu^2$,
    \begin{align*}
        \bigl(L\rho-(1+\theta)b\bigr)^2+\theta^2(D+\kappa^2)
        &=\theta^2D(1+D\nu^2)\\
        &=\theta^2\frac{c^2(1+D\nu^2)^2}{1-\nu^2}.
    \end{align*}
    Since moreover $a_0=\sqrt{1-\nu^2}$ and $1+(D+1)\theta=s/\nu$,
    \begin{align*}
        (1+\theta)ca_0-\sqrt{\bigl(L\rho-(1+\theta)b\bigr)^2+\theta^2(D+\kappa^2)}
        &=(1+\theta)c\sqrt{1-\nu^2}-\theta\frac{c(1+D\nu^2)}{\sqrt{1-\nu^2}}\\
        &=\frac{c}{\sqrt{1-\nu^2}}\bigl((1+\theta)(1-\nu^2)-\theta(1+D\nu^2)\bigr)\\
        &=\frac{c}{\sqrt{1-\nu^2}}\bigl(1-\nu^2(1+(D+1)\theta)\bigr)\\
        &=\frac{c(1-\nu s)}{\sqrt{1-\nu^2}}\\
        &=\psi_k(s).
    \end{align*}
\end{proof}

\begin{lemma}[Replaces \autoref{lem:wy-transverse-bound}]
    \label{lem:q-wy-transverse-bound} Let $1\le j\le\ell$, let $x\in F_j$ satisfy $p_{F_{j-1}}(x)\le e^{-t^2/2}$, let $0<L<\sqrt{H}$ and $\rho\in[-1,1]$, and set $s=\sqrt{H-L^2(1-\rho^2)}$.
    Then, for sufficiently large $t$, over the Gaussian coordinates orthogonal to $\Span(v_r,v_x)$,
    \begin{equation*}
        \Prb_g\left[\exists y\in N_{F_{j-1}}(x): \ip{g}{w_y}+L\rho\xi_y t\ge\left(\psi_k(s)-(20c+1)(1+\bar{\theta})w\right)t\right]
        \ge\Phi\left(\sqrt{H-L^2}\,t\right).
    \end{equation*}
\end{lemma}

\begin{proof}
    As in the proof of \autoref{lem:q-mf-bound}, \autoref{lem:source-body} holds with $\sqrt{3}$ replaced by $c$, since the local-maximum condition becomes
    \begin{equation*}
        cZ_1-\xi_yZ_2
        >\ip{g_{\perp}}{w_y}
    \end{equation*}
    for all $y\in N_{F_{j-1}}(x)$.

    Take the parameters of \autoref{lem:q-exact-parameters}.
    Then $a_0\ge\sqrt{1-\nu^2}>1/\sqrt{2}$, $|b|,\mu\le\nu<1$, and $0\le\theta\le\bar{\theta}$.
    Set $a=a_0-20w$.
    For sufficiently large $t$, we have $a\ge0$ and
    \begin{equation*}
        1-a^2-b^2
        =\mu^2+40a_0w-400w^2
        \ge \mu^2+28w,
    \end{equation*}
    and hence, since $e^{-14wt^2}=t^{-14}$,
    \begin{equation*}
        2\pi t^2e^{-(1-a^2-b^2)t^2/2+a+|b|+t^{-2}}
        \le O(t^{-12})e^{-\mu^2t^2/2}
        \le\bPhi(\mu t),
    \end{equation*}
    where the last inequality uses the lower bound of \autoref{fact:gaussian-tail} and $\mu<1$.
    Set $B_y=\sqrt{D+\kappa^2}\sqrt{1-(\xi_y)^2}\,t$ for $y\in N_{F_{j-1}}(x)$.
    For each $y$, since $\ip{g}{w_y}\sim N(0,\norm{w_y}^2)$ and $\norm{w_y}\le\sqrt{1-(\xi_y)^2}$, we have $\Prb_g[\ip{g}{w_y}>B_y]\le\bPhi(\sqrt{D+\kappa^2}\,t)$.
    Since $|N_{F_{j-1}}(x)|\le n^{\beta}=e^{Dt^2/2}$, a union bound gives
    \begin{align*}
        \Prb_g\left[\exists y\in N_{F_{j-1}}(x):\ip{g}{w_y}>B_y\right]
        &\le e^{Dt^2/2}\bPhi\left(\sqrt{D+\kappa^2}\,t\right)\\
        &\le\bPhi(\kappa t),
        &&(\because\text{\autoref{fact:tail-monotone}})
    \end{align*}
    that is, $\gmeasure(K_1)\ge\Phi(\kappa t)$ for $K_1=\{g:\ip{g}{w_y}\le B_y\ \forall y\in N_{F_{j-1}}(x)\}$.
    Moreover, since $a\ge0$ and $p_{F_{j-1}}(x)\le e^{-t^2/2}$, \autoref{lem:source-body} with $\sqrt{3}$ replaced by $c$ and the bound $2\pi t^2e^{-(1-a^2-b^2)t^2/2+a+|b|+t^{-2}}\le\bPhi(\mu t)$ above give $\gmeasure(K_0(a,b))\le\bPhi(\mu t)$ for $K_0(a,b)=\{g:\ip{g_{\perp}}{w_y}\le(ca-b\xi_y)t-\tfrac{1}{t}\ \forall y\in N_{F_{j-1}}(x)\}$.
    Hence, since $\ip{g}{w_y}=\ip{g_{\perp}}{w_y}$, \autoref{lem:conv-bound} for $Y=N_{F_{j-1}}(x)$, $z_y=w_y$ and the thresholds $(ca-b\xi_y)t-\tfrac{1}{t}$ and $B_y$, with $\mu_0=\mu t$ and $\kappa_0=\kappa t$, gives $\gmeasure(K_{-\theta})\le\bPhi((\mu+\theta(\mu+\kappa))t)$.
    Taking complements, since $\mu+\theta(\mu+\kappa)=\sqrt{H-L^2}$, we obtain
    \begin{align*}
        &\Prb_g\left[\exists y\in N_{F_{j-1}}(x):\ip{g_{\perp}}{w_y}>(1+\theta)(ca-b\xi_y)t-\theta\sqrt{D+\kappa^2}\sqrt{1-(\xi_y)^2}\,t-\frac{1+\theta}{t}\right]\\
        &\qquad\ge\Phi\left(\sqrt{H-L^2}\,t\right).
    \end{align*}
    For every $y$ as in this event, since $\ip{g}{w_y}=\ip{g_{\perp}}{w_y}$,
    \begin{align*}
        \ip{g}{w_y}+L\rho\xi_yt
        &>\Bigl((1+\theta)ca+\bigl(L\rho-(1+\theta)b\bigr)\xi_y-\theta\sqrt{D+\kappa^2}\sqrt{1-(\xi_y)^2}\Bigr)t-\frac{1+\theta}{t}\\
        &\ge\left((1+\theta)ca-\sqrt{\bigl(L\rho-(1+\theta)b\bigr)^2+\theta^2(D+\kappa^2)}\right)t-\frac{1+\theta}{t}\\
        &=\bigl(\psi_k(s)-20c(1+\theta)w\bigr)t-\frac{1+\theta}{t}\\
        &\ge\bigl(\psi_k(s)-(20c+1)(1+\bar{\theta})w\bigr)t,
    \end{align*}
    where the equality uses \autoref{lem:q-exact-parameters} and $a=a_0-20w$, and the last inequality uses $\theta\le\bar{\theta}$.
    This proves the claim.
\end{proof}

\begin{lemma}[Replaces \autoref{lem:U-increase}]
    \label{lem:q-U-increase} If $C_{\ell}$ is sufficiently large, then, for sufficiently large $n$ and every $0\le U\le\sqrt{H}$,
    \begin{equation*}
        \frac{1+\lambda}{c}\psi_k\left(\sqrt{H-U^2}\right)
        \ge(1-\lambda)U+\left(\frac{1+\lambda}{c}(20c+1)(1+\bar{\theta})+40\right)w.
    \end{equation*}
\end{lemma}

\begin{proof}
    Put $s=\sqrt{H-U^2}$, so that $U^2+s^2=H$.
    First suppose that $s\ge\nu$.
    Then Cauchy--Schwarz gives
    \begin{equation*}
        (1-\lambda)U+\frac{(1+\lambda)\nu}{\sqrt{1-\nu^2}}s
        \le\sqrt{H\left((1-\lambda)^2+\frac{(1+\lambda)^2\nu^2}{1-\nu^2}\right)},
    \end{equation*}
    and hence
    \begin{align*}
        \frac{1+\lambda}{c}\psi_k(s)-(1-\lambda)U
        &=\frac{1+\lambda}{\sqrt{1-\nu^2}}-\left((1-\lambda)U+\frac{(1+\lambda)\nu}{\sqrt{1-\nu^2}}s\right)\\
        &\ge\frac{1+\lambda}{\sqrt{1-\nu^2}}-\sqrt{H\left((1-\lambda)^2+\frac{(1+\lambda)^2\nu^2}{1-\nu^2}\right)}\\
        &=\sqrt{\frac{2kD}{(k-1)(D+1)}}-\sqrt{\frac{2H\bigl((k-2)D-2\bigr)}{(k-1)(D+1)}},
    \end{align*}
    where the last equality uses $1+\lambda=\frac{k}{k-1}$, $1-\lambda=\frac{k-2}{k-1}$ and $1-\nu^2=\frac{k(D+1)}{2(k-1)D}$.
    Let $\zeta_k$ denote the last expression.
    Now suppose that $s\le\nu$.
    Then the left-hand side is $(1+\lambda)\sqrt{1-s^2}-(1-\lambda)\sqrt{H-s^2}$, and its derivative in $s$,
    \begin{equation*}
        s\left(\frac{1-\lambda}{\sqrt{H-s^2}}-\frac{1+\lambda}{\sqrt{1-s^2}}\right),
    \end{equation*}
    is nonpositive for $0\le s\le\nu$, since $\sqrt{H-s^2}>\sqrt{1-s^2}>0$ by $H>1$ and $\nu<1$, and $1-\lambda<1+\lambda$.
    Hence it is at least its value at $s=\nu$, which is at least $\zeta_k$ by the case $s\ge\nu$.

    It remains to show that $\zeta_k/w$ is at least $\frac{1+\lambda}{c}(20c+1)(1+\bar{\theta})+40$.
    By the definitions of $D$, $H_0$ and $t^2=\frac{2(k-2)\beta}{2+k\beta}\log{n}+2C_{\ell}\log{\ell}$,
    \begin{align*}
        kD-H_0\bigl((k-2)D-2\bigr)
        &=\frac{2\log{n}}{t^4}\bigl((2+k\beta)t^2-2(k-2)\beta\log{n}\bigr)\\
        &=\frac{4C_{\ell}(2+k\beta)\log{n}\log{\ell}}{t^4}.
    \end{align*}
    Since $H=H_0+40w$ and $w=\log{t}/t^2$, this gives
    \begin{align*}
        \frac{kD-H\bigl((k-2)D-2\bigr)}{w}
        &=4C_{\ell}(2+k\beta)\frac{\log{n}}{t^2}\cdot\frac{\log{\ell}}{\log{t}}-40\bigl((k-2)D-2\bigr)\\
        &\to\frac{4C_{\ell}(2+k\beta)^2}{(k-2)\beta}-40k\beta,
    \end{align*}
    where we used $\tfrac{\log{n}}{t^2}\to\tfrac{2+k\beta}{2(k-2)\beta}$, $\tfrac{\log{\ell}}{\log{t}}\to2$ and $(k-2)D-2\to k\beta$.
    Rationalizing $\zeta_k$, we have
    \begin{equation*}
        \zeta_k
        =\frac{2\bigl(kD-H((k-2)D-2)\bigr)}{(k-1)(D+1)}\cdot\frac{1}{\sqrt{\frac{2kD}{(k-1)(D+1)}}+\sqrt{\frac{2H((k-2)D-2)}{(k-1)(D+1)}}}.
    \end{equation*}
    Moreover, $D+1\to\frac{k(1+\beta)}{k-2}$, and both radicands in $\zeta_k$ tend to $\frac{2(2+k\beta)}{(k-1)(1+\beta)}$.
    Therefore
    \begin{equation*}
        \frac{\zeta_k}{w}
        \to\sqrt{\frac{(k-1)(1+\beta)}{2(2+k\beta)}}\left(\frac{4C_{\ell}(2+k\beta)^2}{k\beta(k-1)(1+\beta)}-\frac{40(k-2)\beta}{(k-1)(1+\beta)}\right),
    \end{equation*}
    which grows linearly in $C_{\ell}$ with a positive coefficient depending only on $k$ and $\beta$.
    On the other hand, $\frac{1+\lambda}{c}(20c+1)(1+\bar{\theta})+40$ converges to a constant that does not depend on $C_{\ell}$, so the claim follows by taking $C_{\ell}$ sufficiently large.
\end{proof}

\begin{lemma}[Replaces \autoref{lem:lifting}]
    \label{lem:q-lifting} Let $1\le j\le\ell$, and let $W\subseteq F_j$ be fixed independently of the Gaussian $g$, and suppose that $\{u_x\}_{x\in W}$ is a $(Ut,q)$-cover for some $0<U\le\sqrt{H}$ and $q\in[0,1]$.
    Assume that $p_{F_{j-1}}(x)\le e^{-t^2/2}$ for every $x\in W$, and set $W'=N_{F_{j-1}}(W)$.
    Then, for sufficiently large $t$, $\{u_y\}_{y\in W'}$ is a $((U+40w)t,q-t^{-18})$-cover.
\end{lemma}

\begin{proof}
    We follow the proof of \autoref{lem:lifting}, with $\cE_x$ and $\cE^{\mathrm{bad}}_x$ defined as there.
    Fix $x\in W$ with $\norm{u_x}>0$, and condition on $Z_x=-\ip{g}{u_x/\norm{u_x}}=z$, where $Ut/\norm{u_x}\le z<\sqrt{H}\,t$.
    Set $L=z/t$, $\rho=\rho_x$ and $U'=\norm{u_x}z/t$.
    Then $U\le U'<\sqrt{H}$ and $(U')^2=L^2(1-\rho^2)$.
    Put $B=\psi_k(\sqrt{H-(U')^2})-(20c+1)(1+\bar{\theta})w$.
    By \autoref{lem:q-wy-transverse-bound},
    \begin{equation*}
        \Prb_g\left[\exists y\in N_{F_{j-1}}(x):\ip{g}{w_y}+L\rho\xi_yt\ge Bt\ \middle|\ Z_x=z\right]
        \ge\Phi\left(\sqrt{Ht^2-z^2}\right).
    \end{equation*}
    By \autoref{lem:q-U-increase} at $U'$, we have $\lambda U'+\frac{1+\lambda}{c}B\ge U'+40w$, and in particular $\frac{1+\lambda}{c}B\ge(1-\lambda)U'+40w>0$, so $B>0\ge-cU'$.
    Hence, for every $y$ as in this event, \autoref{lem:q-child} gives
    \begin{equation*}
        \ip{g}{u_y}
        \ge\left(\lambda U'+\frac{1+\lambda}{c}B\right)t
        \ge(U'+40w)t
        \ge(U+40w)t.
    \end{equation*}
    Therefore $\Prb_g[\cE^{\mathrm{bad}}_x\mid Z_x=z]\le\bPhi(\sqrt{Ht^2-z^2})$, and the same integration over $z$ and union bound as in the proof of \autoref{lem:lifting} give
    \begin{equation*}
        \Prb_g\left[\bigcup_{x\in W}\cE^{\mathrm{bad}}_x\right]
        \le\left(1+\frac{\sqrt{H}\,t}{\sqrt{2\pi}}\right)e^{-(H-H_0)t^2/2}
        \le t^{-18}
    \end{equation*}
    for sufficiently large $t$, where $H-H_0=40w$.
    Gaussian symmetry and the same event inclusion now show that $\{u_y\}_{y\in W'}$ is a $((U+40w)t,q-t^{-18})$-cover.
\end{proof}

\begin{lemma}[Replaces \autoref{lem:high-degree}]
    \label{lem:q-high-degree} Let $\frac{k-2}{k}<\beta\le1$, and suppose that every vertex of $G$ has degree at most $n^{\beta}$.
    Then, there is a randomized polynomial-time algorithm that, with probability at least $\frac{3}{4}$, returns an independent set $I$ with $|I|\ge n^{1-\frac{(k-2)\beta}{2+k\beta}-o(1)}$.
\end{lemma}

\begin{proof}
    As in the proof of \autoref{lem:high-degree}, run one trial of \autoref{alg:sparse} on $G$ and $\{v_x\}_{x\in V}$ with the parameters of \autoref{tab:params-A}.
    Since $\ell=O(\log{n})$, we have $\delta^{-1}=n^{\frac{(k-2)\beta}{2+k\beta}}\polylog(n)$.
    One trial calls \autoref{alg:local-maxima} at most $\ell R$ times, each call takes polynomial time, and $R=O(\delta^{-2}\log{n})=n^{\frac{2(k-2)\beta}{2+k\beta}+o(1)}$, so one trial runs in polynomial time.

    By \autoref{lem:estimates}, with probability at least $\frac{3}{4}$, every empirical estimate satisfies $|\hat{p}_j(x)-p_{F_{j-1}}(x)|\le\delta/8$.
    Conditioned on this event, \autoref{lem:sparse-progress-one} shows that \autoref{alg:sparse} returns $\mathsf{SUCCESS}(I)$ for an independent set $I$ with
    \begin{equation*}
        |I| \geq T \geq \frac{n\delta}{100\ell^4} = n^{1-\frac{(k-2)\beta}{2+k\beta}-o(1)}.
    \end{equation*}
    As noted at the beginning of this appendix, the proof of \autoref{lem:sparse-progress-one} goes through, with \autoref{lem:q-initial-cover} and \autoref{lem:q-lifting} in place of \autoref{cor:initial-cover} and \autoref{lem:lifting}.
\end{proof}

\begin{proof}[Proof of \autoref{thm:vector-k-coloring}]
    Fix $C>0$. Run $\lceil C\log_4{n}\rceil$ independent trials of \autoref{alg:sparse}, using the parameters from \autoref{lem:q-low-degree} if $0<\beta\le\frac{k-2}{k}$ and from \autoref{lem:q-high-degree} if $\frac{k-2}{k}<\beta\le1$.
    Return the largest independent set obtained.
    By \autoref{lem:q-low-degree} or \autoref{lem:q-high-degree}, according to the regime, each trial returns an independent set of the required size with probability at least $\frac{3}{4}$.
    Hence the probability that no trial returns such a set is at most $4^{-\lceil C\log_4{n}\rceil}\le n^{-C}$. Since each trial runs in polynomial time and there are $O(\log{n})$ trials, the total running time remains polynomial.
\end{proof}

\begin{remark}
    \label{rem:q-dependence} The constant $C_{\ell}(k,\beta)$, the limit of $\bar{\theta}$, the threshold on $n$ and the $o(1)$ terms in \autoref{thm:vector-k-coloring} depend on $k$ and $\beta$, and they degenerate as $k\to2$, where $c\to\infty$ and $H\to\infty$, and as $\beta\to\frac{k-2}{k}$ from above, where $\nu\to0$.
    For fixed $k>2$ and fixed $\beta$ this is harmless, but applying the theorem with a parameter $k=k(n)$ tending to $2$ would require tracking this dependence.
\end{remark}
\section{Polynomial-Time Implementation of \textsc{BuildVector} and Treatment of Approximate Constraints}
\label{sec:B}
In this appendix, we discuss the details of the SDP implementation and the treatment of its constraints.
In \autoref{sec:B.1}, we describe \textsc{BuildVector} and prove \autoref{lem:build-vector}.
Next, in \autoref{sec:B.2}, we justify treating some of the approximate constraints of the SDP as exact.
\subsection{Ellipsoid Method}\label{sec:B.1}
\begin{algorithm}[H]
    \caption{\textsc{BuildVector}}
    \label{alg:build-vector}
    \begin{algorithmic}[1]
        \Require graph $G=(V,E)$
        \Ensure vectors $\{v_{i}\}_{i\in V}$ and band witnesses $\cS$ as in \autoref{lem:build-vector}
        \State Run the ellipsoid method over $\cX$ with initial ball $B(I,R_{1})$ and oracle $\cO$
        \State Let $(X,\cS)$ be the output returned by $\cO$
        \State Take a Gram realization $X_{ij}=\ip{v_i}{v_j}$
        \State \Return $(\{v_{i}\}_{i\in V},\cS)$
    \end{algorithmic}
\end{algorithm}

Here, let $\cX=\{X\in\R^{V\times V}:X=X^{\mathsf{T}},\,\forall i\in V:X_{ii}=1\}$.
We regard $\cX$ as an affine space of dimension $n(n-1)/2$ with the Frobenius norm, and write $B(X,R)=\{Y\in\cX:\norm{Y-X}_{\mathrm{F}}\le R\}$.
We define $K$ as the set of $X\in\cX$ satisfying
\begin{align*}
    X
    & \succeq 0, \\
    X_{ij}
    & \ge -\frac{1}{2}
    & & \text{for all }i,j\in V, \\
    X_{ij}
    & \le -\frac{1}{2}+\varepsilon_{\mathrm{SDP}}
    & & \text{for all }ij\in E, \\
    X_{ij}+X_{ij'}
    & \le \frac{1}{2}
    & & \text{for all }jj'\in E \text{ and }i\in V\setminus\{j,j'\}, \\
    \sum_{j\in S}\frac{1+2X_{ij}}{3}
    & \le \Theta_{1}(S)+\frac{\varepsilon_{\mathrm{SDP}}}{4}
    & & \text{for all }i\in V \text{ and }S\subseteq V, \\
    \sum_{j\in S}\frac{2-2X_{ij}}{3}
    & \le \Theta_{2}(S)+\frac{\varepsilon_{\mathrm{SDP}}}{4}
    & & \text{for all }i\in V \text{ and }S\subseteq V, \\
    \Gamma_{rx}(X)
    & \succeq-\varepsilon_{\mathrm{SDP}}I
    & & \text{for all }rx\in E.
\end{align*}
Since $X\succeq0$ and the last family of constraints are semidefinite constraints on affine functions of $X$, and every other constraint is linear in $X$, the set $K$ is convex.
The ellipsoid method below does not need to decide whether a candidate matrix $X$ lies in $K$.
For each $X$, the oracle $\cO$ checks the band constraints only for the bands $S_{r,P}=\{j\in V:X_{rj}\in P\}$ determined by $X$, rather than for all $S\subseteq V$, together with the polynomially many other constraints of $K$.
If one of these constraints is violated, $\cO$ returns the corresponding halfspace, which contains $K$ and excludes $X$; otherwise, it outputs $X$ with the band witnesses required by \autoref{lem:build-vector}, although $X$ need not lie in $K$.
In particular, $\cO$ never cuts at a point of $K$.

Fix a proper $3$-coloring of $G$, and let $\overline{X}\in\cX$ be given by $\overline{X}_{ij}=1$ if $i$ and $j$ have the same color and $\overline{X}_{ij}=-\frac{1}{2}$ otherwise.
For $\varepsilon_{0}=\frac{\varepsilon_{\mathrm{SDP}}}{16n}$, define $\tilde{X}=(1-\varepsilon_{0})\overline{X}+\varepsilon_{0} I$.
Let $R_{0} =\frac{\varepsilon_{0}}{8n}$ and $R_{1} = n$.

\begin{algorithm}[H]
    \caption{\textsc{OutputOrCut} oracle $\cO$}
    \label{alg:oracle}
    \begin{algorithmic}[1]
        \Require graph $G=(V,E)$, matrix $X\in\cX$
        \Ensure $(X,\cS)$, or a halfspace that contains $K$ and excludes $X$
        \If{$X\not\succeq0$}
            \State \Return $\{Y\in\cX:u^{\mathsf{T}}Yu\ge0\}$ for a vector $u$ with $u^{\mathsf{T}}Xu<0$
        \EndIf
        \If{$\Gamma_{rx}(X)\not\succeq-\varepsilon_{\mathrm{SDP}}I$ for some $rx\in E$}
            \State \Return $\{Y\in\cX:u^{\mathsf{T}}\Gamma_{rx}(Y)u\ge-\varepsilon_{\mathrm{SDP}}\norm{u}^{2}\}$ for a vector $u$ with $u^{\mathsf{T}}\Gamma_{rx}(X)u<-\varepsilon_{\mathrm{SDP}}\norm{u}^{2}$
        \EndIf
        \If{$X$ violates one of the second to fourth constraints of $K$}
            \State \Return the halfspace given by the violated constraint
        \EndIf
        \For{$r\in V$ and $P\in\cP$}
            \State $S\gets\{j\in V:X_{rj}\in P\}$
            \State $(\cW,\widehat{\Theta})\gets\Call{SolveTheta}{S,m}$, where $m=1$ if $P$ is positive and $m=2$ otherwise
            \If{$P$ is positive and $\sum_{j\in S}\frac{1+2X_{rj}}{3}>\widehat{\Theta}+\frac{\varepsilon_{\mathrm{SDP}}}{2}$}
                \State \Return $\{Y\in\cX:\sum_{j\in S}\frac{1+2Y_{rj}}{3}\le\widehat{\Theta}+\frac{\varepsilon_{\mathrm{SDP}}}{2}\}$
            \EndIf
            \If{$P$ is negative and $\sum_{j\in S}\frac{2-2X_{rj}}{3}>\widehat{\Theta}+\frac{\varepsilon_{\mathrm{SDP}}}{2}$}
                \State \Return $\{Y\in\cX:\sum_{j\in S}\frac{2-2Y_{rj}}{3}\le\widehat{\Theta}+\frac{\varepsilon_{\mathrm{SDP}}}{2}\}$
            \EndIf
            \State $(S_{r,P},\cW_{r,P})\gets(S,\cW)$
        \EndFor
        \State \Return $(X,\cS)$, where $\cS=\{(\rho_{P},S_{r,P},\cW_{r,P})\}_{r\in V,P\in\cP}$
    \end{algorithmic}
\end{algorithm}
Here, \textsc{SolveTheta}$(S,m)$ returns a feasible solution $\cW$ of $\Theta_{m}(S)$ together with its objective value $\widehat{\Theta}$ such that $\widehat{\Theta}\ge\Theta_{m}(S)-\frac{\varepsilon_{\mathrm{SDP}}}{4}$.
The witness is $\cW=(x_{\emptyset},\{x_{j}\}_{j\in S})$ for $m=1$ and $\cW=(x_{\emptyset},\{(x_{j,1},x_{j,2})\}_{j\in S})$ for $m=2$.
Since $m \in \{1,2\}$, the SDP defining $\Theta_m(S)$ has polynomial size, and its feasible region is bounded and has an inverse-polynomial relative interior margin.
Hence, standard polynomial-time algorithms~\cite{GLS88} yield an exactly feasible witness with objective value at least $\Theta_m(S) - \varepsilon_{\mathrm{SDP}}/4$ in $\mathrm{poly}(n, \log{1/\varepsilon_{\mathrm{SDP}}})$ time.
\begin{lemma}[\textsc{OutputOrCut} oracle $\cO$]
    \label{lem:output-or-cut} There is a polynomial-time oracle $\cO$ such that, given a matrix $X\in\cX$, the following holds:
    \begin{itemize}
        \item either $\cO$ returns a pair $(X,\cS)$ satisfying all the guarantees required in \autoref{lem:build-vector},

        \item or $\cO$ returns a halfspace $\cH\subseteq\cX$ such that $K\subseteq \cH$ and $X\notin \cH$.
    \end{itemize}
\end{lemma}

\begin{proof}
    We show that \autoref{alg:oracle} has the desired property.
    First, suppose that $\cO$ returns a halfspace $\cH$.
    By construction, $X\notin \cH$.
    If $\cH=\{Y\in\cX:u^{\mathsf{T}}Yu\ge0\}$, then $K\subseteq \cH$ since every $Y\in K$ satisfies $Y\succeq0$.
    If $\cH=\{Y\in\cX:u^{\mathsf{T}}\Gamma_{rx}(Y)u\ge-\varepsilon_{\mathrm{SDP}}\norm{u}^{2}\}$, then $\cH$ is a halfspace since $\Gamma_{rx}$ is affine, and $K\subseteq \cH$ since every $Y\in K$ satisfies $\Gamma_{rx}(Y)\succeq-\varepsilon_{\mathrm{SDP}}I$.
    If $\cH$ is given by one of the second to fourth constraints of $K$, then $K\subseteq \cH$ by the definition of $K$.
    Otherwise, $\cH$ is returned for some $r\in V$ and $P\in\cP$.
    If $m=1$, then $\cH=\{Y\in\cX:\sum_{j\in S}\frac{1+2Y_{rj}}{3}\le\widehat{\Theta}+\frac{\varepsilon_{\mathrm{SDP}}}{2}\}$, and every $Y\in K$ satisfies
    \begin{align*}
        \sum_{j\in S}\frac{1+2Y_{rj}}{3}
        & \le\Theta_{1}(S)+\frac{\varepsilon_{\mathrm{SDP}}}{4}
        & & (\because Y\in K) \\
        & \le\widehat{\Theta}+\frac{\varepsilon_{\mathrm{SDP}}}{2}.
        & & (\because \widehat{\Theta}\ge\Theta_{1}(S)-\tfrac{\varepsilon_{\mathrm{SDP}}}{4})
    \end{align*}
    If $m=2$, then $\cH=\{Y\in\cX:\sum_{j\in S}\frac{2-2Y_{rj}}{3}\le\widehat{\Theta}+\frac{\varepsilon_{\mathrm{SDP}}}{2}\}$, and every $Y\in K$ satisfies
    \begin{align*}
        \sum_{j\in S}\frac{2-2Y_{rj}}{3}
        & \le\Theta_{2}(S)+\frac{\varepsilon_{\mathrm{SDP}}}{4}
        & & (\because Y\in K) \\
        & \le\widehat{\Theta}+\frac{\varepsilon_{\mathrm{SDP}}}{2}.
        & & (\because \widehat{\Theta}\ge\Theta_{2}(S)-\tfrac{\varepsilon_{\mathrm{SDP}}}{4})
    \end{align*}
    Hence $K\subseteq \cH$ in both cases.

    Next, suppose that $\cO$ returns $(X,\cS)$, and let $\{v_{i}\}_{i\in V}$ be a Gram realization of $X$.
    Since $X\succeq0$, $X_{ii}=1$ for every $i\in V$, and $X$ satisfies the second to fourth constraints of $K$, the vectors $\{v_{i}\}_{i\in V}$ are unit vectors satisfying the conditions on vectors in \autoref{lem:build-vector}.
    Since no halfspace is returned by the check of the last family of constraints, we also have $\Gamma_{rx}(X)\succeq-\varepsilon_{\mathrm{SDP}}I$ for every $rx\in E$.
    Fix $r\in V$ and $P\in\cP$, and let $S=S_{r,P}=\{j\in V:\ip{v_r}{v_j}\in P\}$.
    Since no halfspace is returned for $(r,P)$, the witness $\cW_{r,P}$ with objective value $\widehat{\Theta}$ satisfies the following.
    If $P$ is positive, then $\cW_{r,P}$ is feasible for $\Theta_{1}(S)$ and $\sum_{j\in S}\frac{1+2X_{rj}}{3}\le\widehat{\Theta}+\frac{\varepsilon_{\mathrm{SDP}}}{2}$.
    Since $\rho_{P}=\inf P$, we have $\ip{v_r}{v_j}\ge \rho_{P}$ for every $j \in S$, and hence
    \begin{align*}
        \sum_{j\in S}\norm{x_j}^{2}
        & =\widehat{\Theta} \\
        & \ge\sum_{j\in S}\frac{1+2\ip{v_r}{v_j}}{3}-\varepsilon_{\mathrm{SDP}} \\
        & \ge\frac{1+2\rho_{P}}{3}|S|-\varepsilon_{\mathrm{SDP}},
    \end{align*}
    which is (i) of \autoref{lem:build-vector}.
    If $P$ is negative, then $\cW_{r,P}$ is feasible for $\Theta_{2}(S)$ and $\sum_{j\in S}\frac{2-2X_{rj}}{3}\le\widehat{\Theta}+\frac{\varepsilon_{\mathrm{SDP}}}{2}$.
    Since $\rho_{P}=\sup P$, we have $\ip{v_r}{v_j}\le \rho_{P}$ for every $j\in S$, and hence
    \begin{align*}
        \sum_{j\in S}\left(\norm{x_{j,1}}^{2}+\norm{x_{j,2}}^{2}\right)
        & =\widehat{\Theta} \\
        & \ge\sum_{j\in S}\frac{2-2\ip{v_r}{v_j}}{3}-\varepsilon_{\mathrm{SDP}} \\
        & \ge\frac{2-2\rho_{P}}{3}|S|-\varepsilon_{\mathrm{SDP}},
    \end{align*}
    which is (ii) of \autoref{lem:build-vector}.

    Finally, the checks before the loop take polynomial time, and the loop calls \textsc{SolveTheta} at most $|V||\cP|=O(n\log{n})$ times, each of which takes polynomial time.
\end{proof}
\begin{proof}[Proof of \autoref{lem:build-vector}]
    First, we prove $B(\tilde{X}, R_{0}) \subseteq K$.
    Write $Y = \tilde{X} + \Delta X$ where $\Delta X$ is an arbitrary perturbation with $\Delta X = \Delta X^{\mathsf{T}}$, $\Delta X_{ii}= 0$ and $\norm{\Delta X}_{\mathrm{F}} \le R_{0}$.
    Thus, $|\Delta X_{ij}| \le R_{0}$ and $\norm{\Delta X}_{\mathrm{op}}\le \norm{\Delta X}_{\mathrm{F}}\le R_{0}$ hold.
    This gives $Y = \tilde{X} + \Delta X \succeq(\varepsilon_{0} - \frac{\varepsilon_{0}}{8n})I \succeq 0$.
    Since $(\tilde{X})_{ij}\ge -\frac{1}{2} + \frac{\varepsilon_{0}}{2}$ holds, we have $Y_{ij}\ge-\frac{1}{2}+\frac{\varepsilon_{0}}{2}-R_{0}\ge-\frac{1}{2}$.
    Moreover, since $(\tilde{X})_{ij}= -\frac{1}{2} + \frac{\varepsilon_{0}}{2}$ for every edge $ij \in E$, we have $Y_{ij}\le-\frac{1}{2}+\frac{\varepsilon_{0}}{2}+R_{0}\le-\frac{1}{2} +\varepsilon_{\mathrm{SDP}}$.
    For the fourth constraint, $\overline{X}_{ij}+ \overline{X}_{ij'}\le \frac{1}{2}$ holds and $(\tilde{X})_{ij}+ (\tilde{X})_{ij'}\le(1-\varepsilon_{0})\frac{1}{2}$.
    Therefore, $Y_{ij}+Y_{ij'}\le(1-\varepsilon_{0})\frac{1}{2}+2R_{0}\le\frac{1}{2}$.

    Since $\frac{1+2\overline{X}_{ij}}{3}$ is $1$ if $j$ has the same color as $i$ and $0$ otherwise, $\sum_{j\in S}\frac{1+2\overline{X}_{ij}}{3}$ is the number of vertices of $S$ with the color of $i$, which form an independent set of $G[S]$.
    Similarly, $\sum_{j\in S}\frac{2-2\overline{X}_{ij}}{3}$ is the number of the other vertices of $S$, which induce a $2$-colorable subgraph of $G[S]$.
    Since $\Theta_{1}(S)$ and $\Theta_{2}(S)$ are relaxations of these problems, these sums are at most $\Theta_{1}(S)$ and $\Theta_{2}(S)$, respectively.
    Hence, we obtain $\sum_{j\in S}\frac{1 + 2(\tilde{X})_{ij}}{3}\le \Theta_{1}(S) +\frac{\varepsilon_{0} n}{3}$ and $\sum_{j\in S}\frac{2 - 2(\tilde{X})_{ij}}{3}\le \Theta_{2}(S) +\frac{2\varepsilon_{0} n}{3}$.
    Moreover,
    \begin{equation*}
        \left|\sum_{j\in S}\frac{1 + 2Y_{ij}}{3}- \sum_{j\in S}\frac{1 + 2(\tilde{X})_{ij}}{3} \right|
        \le \frac{2|S|R_{0}}{3}
        \le \frac{\varepsilon_{0}}{12}
    \end{equation*}
    and
    \begin{equation*}
        \left|\sum_{j\in S}\frac{2 - 2Y_{ij}}{3}- \sum_{j\in S}\frac{2 - 2(\tilde{X})_{ij}}{3} \right|
        \le \frac{2|S|R_{0}}{3}
        \le \frac{\varepsilon_{0}}{12}.
    \end{equation*}
    Thus, $\sum_{j\in S}\frac{1 + 2Y_{ij}}{3}\le \Theta_{1}(S) + \varepsilon_{0} n \le \Theta_{1}(S) + \frac{\varepsilon_{\mathrm{SDP}}}{4}$ and $\sum_{j\in S}\frac{2 - 2Y_{ij}}{3}\le \Theta_{2}(S) + \varepsilon_{0} n \le \Theta_{2}(S) + \frac{\varepsilon_{\mathrm{SDP}}}{4}$.

    It remains to check the last family of constraints.
    The matrix $\overline{X}$ is the Gram matrix of $\{e_{c(i)}\}_{i\in V}$, where $c(i)$ is the color of $i$ and $e_{1},e_{2},e_{3} \in\R^{2}$ are unit vectors with pairwise inner products $-\frac{1}{2}$.
    Fix an edge $rx\in E$.
    Let $d$ be the color other than $c(r)$ and $c(x)$.
    For $y\in N_{V}(r)$, we have $c(y)\in\{c(x),d\}$, and since $e_{1}+e_{2}+e_{3}=0$, we get $(2e_{c(y)}+e_{c(r)})/\sqrt{3}=\pm(e_{c(x)}-e_{d})/\sqrt{3}$; similarly, $(2e_{c(z)}+e_{c(x)})/\sqrt{3}=\pm(e_{c(r)}-e_{d})/\sqrt{3}$ for $z\in N_{V}(x)$.
    Thus, for $\alpha_1\in\R^{N_V(r)}$ and $\alpha_2\in\R^{N_V(x)}$, the vectors $u_1=\sum_{y}(\alpha_1)_{y}(2e_{c(y)}+e_{c(r)})/\sqrt{3}$ and $u_2=\sum_{z}(\alpha_2)_{z}(2e_{c(z)}+e_{c(x)})/\sqrt{3}$ lie in $\R(e_{c(x)}-e_{d})$ and $\R(e_{c(r)}-e_{d})$, respectively.
    Since $\norm{e_{c(x)}-e_{d}}^{2}=\norm{e_{c(r)}-e_{d}}^{2}=3$ and $\ip{e_{c(x)}-e_{d}}{e_{c(r)}-e_{d}}=\frac{3}{2}$, we have $|\ip{u_1}{u_2}|=\frac{1}{2}\norm{u_1}\norm{u_2}$, and hence
    \begin{equation*}
        (\alpha_1,\alpha_2)^{\mathsf{T}}\Gamma_{rx}(\overline{X})(\alpha_1,\alpha_2)
        =\norm{u_1}^{2}+4\ip{u_1}{u_2}+\norm{u_2}^{2}
        \ge(\norm{u_1}-\norm{u_2})^{2}
        \ge0.
    \end{equation*}
    Thus $\Gamma_{rx}(\overline{X})\succeq0$.
    Recall that, for $Z\in\cX$, the entry of $\Gamma_{rx}(Z)$ indexed by $((i,j),(i',j'))$ is $\frac{1+\mathbf{1}_{i\ne i'}}{3}(4Z_{jj'}+2Z_{ji'}+2Z_{ij'}+Z_{ii'})$.
    For every $Z\in\cX$ with $|Z_{ij}|\le1$, every entry of $\Gamma_{rx}(Z)$ has absolute value at most $\frac{2}{3}(4+2+2+1)=6$, and $\Gamma_{rx}(Z)$ has at most $2n$ rows, so $\norm{\Gamma_{rx}(I)}_{\mathrm{op}}\le\norm{\Gamma_{rx}(I)}_{\mathrm{F}}\le12n$.
    Since $\Gamma_{rx}$ is affine, $\Gamma_{rx}(\tilde{X})=(1-\varepsilon_{0})\Gamma_{rx}(\overline{X})+\varepsilon_{0}\Gamma_{rx}(I)\succeq-12 n\varepsilon_{0} I=-\frac{3}{4}\varepsilon_{\mathrm{SDP}}I$.
    Moreover, every entry of $\Gamma_{rx}(Y)-\Gamma_{rx}(\tilde{X})$ has absolute value at most $\frac{2}{3}(4+2+2+1)R_{0}=6R_{0}$, so $\norm{\Gamma_{rx}(Y)-\Gamma_{rx}(\tilde{X})}_{\mathrm{op}}\le12nR_{0}=\frac{3}{2} \varepsilon_{0}\le\frac{1}{4}\varepsilon_{\mathrm{SDP}}$.
    Hence $\Gamma_{rx}(Y)\succeq -\varepsilon_{\mathrm{SDP}}I$.
    Therefore, we obtain $B(\tilde{X}, R_{0}) \subseteq K$.

    Next, we prove $K \subseteq B(I, R_{1})$.
    Let $X\in K$. Since $X\succeq0$ and $X_{ii}=1$ for every $i\in V$, we have $|X_{ij}|\le1$ for all $i\ne j$.
    Thus, $\norm{X - I}_{\mathrm{F}}^{2} = 2\sum_{i < j}X_{ij}^{2} \le 2\binom{n}{2}$, which gives $\norm{X-I}_{\mathrm{F}}\le \sqrt{2\binom{n}{2}}\le n$.
    This implies that $K \subseteq B(I, R_{1})$.

    By the standard analysis of the ellipsoid method~\cite{GLS88}, in \autoref{alg:build-vector}, the \textsc{OutputOrCut} oracle is called at most $O(n^{4}\log{(R_{1}/R_{0})}) = O(n^{4}\log{(n/\varepsilon_{\mathrm{SDP}})})$ times, since $\cX$ has dimension $n(n-1)/2$.
    Hence, by \autoref{lem:output-or-cut}, \autoref{alg:build-vector} returns an output satisfying \autoref{lem:build-vector}.
\end{proof}

\subsection{Justification of the Exact-Constraint Assumption}\label{sec:B.2}
First, we justify treating the strict vector coloring constraints as exact.
The computed SDP solution satisfies these constraints up to $\varepsilon_{\mathrm{SDP}}$.
This induces $O(\varepsilon_{\mathrm{SDP}})$ perturbations in $v_{xy}$ and $w_y$, and hence $O(n\varepsilon_{\mathrm{SDP}})$ perturbations in their Gaussian projections with probability at least $1 - e^{-n}$.
These errors are absorbed either by the slack already present in the thresholds or into the $o(t^2)$ terms in the exponents of the probability bounds throughout our analysis.

Next, we discuss the semidefinite constraint.
In the proof of \autoref{lem:neighborhood-preparation}, we treat the approximate constraint $\Gamma_{rx}(X) \succeq -\varepsilon_{\mathrm{SDP}}I$ as if $\Gamma_{rx}(X) \succeq 0$.
This requires justification even when $\varepsilon_{\mathrm{SDP}}$ is arbitrarily small.
Indeed, if the vectors are nearly linearly dependent, an additive error in the semidefinite constraint need not imply a small error in the correlation bound of \autoref{lem:edge-neighborhood}.

We address this issue by adding tiny independent Gaussian noise to the Gaussian projections used in the proof of (ii) of \autoref{lem:local-bound}.
Fix a vertex $r$.
For every $y\in N_{V}(r)$, and every $x\in N_{V}(r)$ and $z\in N_{V}(x)$, let $h_{ry}$ and $h_{xz}$ be mutually independent standard Gaussian variables, also independent of $g$.
Define
\begin{align*}
    T_{ry}
    & =\ip{g}{v_{ry}}+\sigma h_{ry}, \\
    T_{xz}
    & =\ip{g}{v_{xz}}+\sigma h_{xz},
\end{align*}
where $\sigma=e^{-n}$.
Using $v_{xz}=\xi_{z}v_{xr}+w_{z}$, we further define
\begin{equation*}
    Z_{xz}
    =T_{xz}-\xi_{z}T_{xr}
    =\ip{g}{w_z}+\sigma(h_{xz}-\xi_{z}h_{xr}).
\end{equation*}
Moreover,
\begin{equation*}
    \E[T_{rx}Z_{xz}]
    =\ip{v_{rx}}{w_z}
    =0
\end{equation*}
holds for every $z\in N_{V}(x)$.
Since each edge residual $v_{ij}$ ($ij\in E$) differs from $(2v_j+v_i)/\sqrt{3}$ by $O(\varepsilon_{\mathrm{SDP}})$, the matrix
\begin{equation*}
    \widetilde{\Gamma}_{rx}
    =
    \begin{pmatrix}
    \E[T_{r}T_{r}^{\mathsf{T}}] & 2\E[T_{r}T_{x}^{\mathsf{T}}] \\
    2\E[T_{x}T_{r}^{\mathsf{T}}] & \E[T_{x}T_{x}^{\mathsf{T}}]
    \end{pmatrix}
\end{equation*}
differs from $\Gamma_{rx}(X)+\sigma^{2}I$ by $O(n\varepsilon_{\mathrm{SDP}})$ in operator norm, where $T_{r}=(T_{ry})_{y\in N_V(r)}$ and $T_{x}=(T_{xz})_{z\in N_V(x)}$.
Since $\Gamma_{rx}(X)\succeq-\varepsilon_{\mathrm{SDP}}I$ and $n\varepsilon_{\mathrm{SDP}}=ne^{-3n}=o(\sigma^{2})$, we have $\widetilde{\Gamma}_{rx}\succeq0$ for sufficiently large $n$.
Hence the argument of \autoref{lem:edge-neighborhood} gives $|\E[uv]|\le\frac{1}{2}\sqrt{\E[u^{2}]\E[v^{2}]}$ for all $u\in\Span\{T_{ry}:y\in N_{V}(r)\}$ and $v\in\Span\{T_{xz}:z\in N_{V}(x)\}$.

Since $\E[T_{rx}Z_{xz}]=0$ for every $z\in N_V(x)$ and these variables are jointly Gaussian, the family $\{Z_{xz}\}_{z\in N_V(x)}$ is independent of $T_{rx}$.
Moreover, writing
\begin{equation*}
    Z_{ry}
    =T_{ry}-\frac{\E[T_{ry}T_{rx}]}{\E[T_{rx}^{2}]}T_{rx}
\end{equation*}
for $y\in N_V(r)$, the family $\{Z_{ry}\}_{y\in N_V(r)}$ is also independent of $T_{rx}$.
Since $Z_{ry}\in\Span\{T_{ry'}:y'\in N_V(r)\}$ and $Z_{xz}\in\Span\{T_{xz'}:z'\in N_V(x)\}$ (recall that $r\in N_V(x)$), the correlation bound above remains valid between the spans of $\{Z_{ry}\}_{y}$ and $\{Z_{xz}\}_{z}$.
These spans play the roles of the subspaces $\cU_{r}\cap\Span(v_{r},v_{x})^{\perp}$ and $\cU_{x}\cap\Span(v_{r},v_{x})^{\perp}$ in the proof of (ii) of \autoref{lem:local-bound}.

Accordingly, we define the maximizer events and the events $\cB_{x}^{-}$ by the noisy variables.
Put $\widetilde{\cL}=\{\max_{y\in N_{F_{j}}(r)}T_{ry}>At\}$, let $\widetilde{\cL}^{(x)}$ be the event that $\widetilde{\cL}$ occurs and $x$ is the maximizer of $T_{ry}$, and, in the setting of the induction step for the local bound in the proof of \autoref{lem:neighborhood-preparation} (so that $j=i+2$, $A=A_{i+1}$ and $A'=A_{i}$), put
\begin{equation*}
    \widetilde{\cB}_{x}^{-}
    =\left\{\forall z\in N_{F_{j-1}}(x):-Z_{xz}\le(1+\theta)(\sqrt{3}a-b\xi_{z})t-\theta(A'-c|\xi_{z}|)t-\frac{1+\theta}{t}\right\},
\end{equation*}
which is the event $\cB_{x}^{-}$ in the proof of \autoref{lem:local-bound} with $-\ip{g_{\perp}}{w_z}$ replaced by $-Z_{xz}$.
With probability at least $1-e^{-\Omega(n^2)}$, all the introduced Gaussian noises have absolute value at most $n$.
On this event, the perturbations in the Gaussian projections are $O(n\sigma)$.
These perturbations are absorbed by the existing slack in the measure estimates and inequalities, and the exceptional probability is negligible compared with $\varepsilon_t$.
Since $T_{ry}=\frac{\E[T_{ry}T_{rx}]}{\E[T_{rx}^{2}]}T_{rx}+Z_{ry}$, conditional on $T_{rx}=s$, the event $\widetilde{\cL}^{(x)}$ depends only on $\{Z_{ry}\}_{y\in N_V(r)}$, while $\widetilde{\cB}_{x}^{-}$ depends only on $\{Z_{xz}\}_{z\in N_V(x)}$ and is independent of $T_{rx}$.
Hence, conditional on $T_{rx}$, \autoref{cor:gaussian-holder} applies to $\widetilde{\cL}^{(x)}$ and $\widetilde{\cB}_{x}^{-}$ with the same exponents $p_{1},p_{2}$ as in the proof of \autoref{lem:neighborhood-preparation}, since only the correlation bound $\frac{1}{2}$ is used, and the proof of (ii) of \autoref{lem:local-bound} goes through with $\widetilde{\cL}$, $\widetilde{\cL}^{(x)}$ and $\widetilde{\cB}_{x}^{-}$ in place of $\cL$, $\cL^{(x)}$ and $\cB_{x}^{-}$, and with $\Prb[T_{rx}>At]=\bPhi(At/\sqrt{1+\sigma^{2}})\le2\bPhi(At)$ in place of $\bPhi(At)$.
In the application of (ii) in the proof of \autoref{lem:neighborhood-preparation}, the resulting extra factor $2^{1-1/p_{1}}=2^{1/11997}$ is absorbed by the factor $e^{-\eta t^{2}/11997}$ there for sufficiently large $t$.

Finally, we transfer the bound on the perturbed event $\widetilde{\cL}$ to the original local event $\cL$.
Since 
\begin{equation*}
    \Prb[\cL \triangle \widetilde{\cL}] 
    \le \sum_{x \in N_{F_j}(r)} \Prb[|\ip{g}{v_{rx}} - At| \le \sigma|h_{rx}|] 
    \le \sum_{x \in N_{F_j}(r)} \E_{h_{rx}}\left[\frac{2\sigma|h_{rx}|}{\sqrt{2\pi}}\right] 
    \le n \sqrt{\frac{2}{\pi}}\frac{2\sigma}{\sqrt{2\pi}} 
    = O(ne^{-n}),
\end{equation*}
we have $\Prb[\cL] \le \Prb[\widetilde{\cL}] + O(ne^{-n})$ and the same induction argument applies in \autoref{lem:neighborhood-preparation}.

\section{Computational Verification}
\label{sec:C}

In this appendix, we prove \autoref{lem:source-bounds}, recalled as follows, and the property (P2) of \autoref{sec:4.4}, which states that if $A = A_0$, then $q_{\mathrm{c}} \leq f_A$ holds, that is, $q_A = q_{\mathrm{c}}$.

\magiccondition*

First, we restrict the domain of Case 1; as illustrated in \autoref{fig:appendix-c-1}, not all $(A, \rho)$ meet the assumption of Case 1. The assumption of Case 1 is $q_A(\rho) < q_{\mathrm{c}}(\rho)$, so it is not required to verify the property for the range of $(A, \rho)$ such that $q_A(\rho) = q_{\mathrm{c}}(\rho)$, that is, $f_A(\rho) \geq q_{\mathrm{c}}(\rho)$. With the following discussion, we also prove the property (P2).

\begin{figure}[t]
    \centering
    \includegraphics[width=\linewidth]{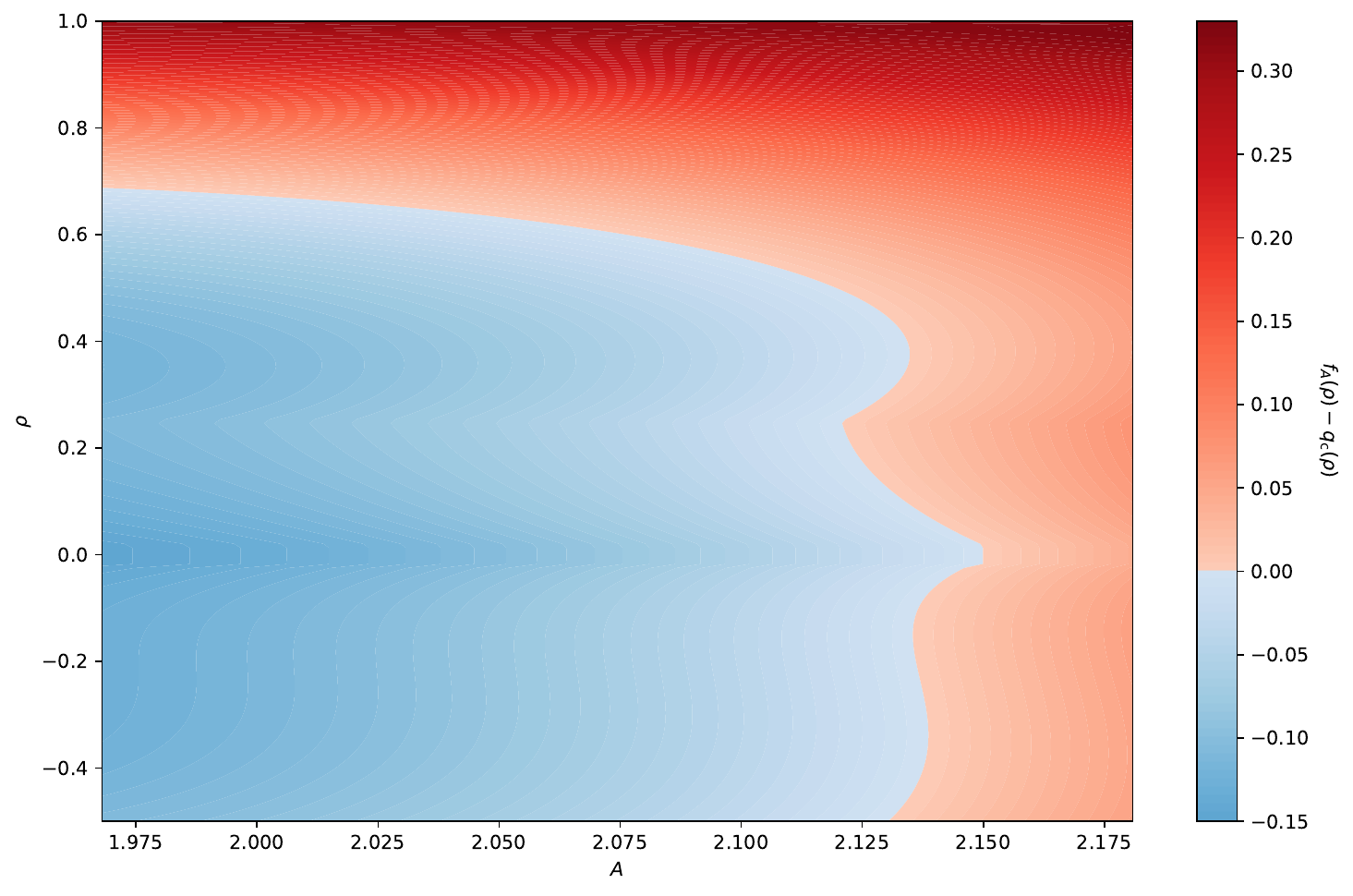}
    \caption{The graph showing the value of $f_A(\rho) - q_{\mathrm{c}}(\rho)$ over $(A, \rho) \in [1.968, A_0] \times [-\frac{1}{2}, 1]$. The blue range satisfies the assumption of Case 1.}
    \label{fig:appendix-c-1}
\end{figure}

\begin{lemma} \label{lem:appendix-c-1}
    Suppose $1.968 \leq A \leq A_0$. If $\rho \geq 0.7$, then $q_A(\rho) = q_{\mathrm{c}}(\rho)$, that is, $f_A(\rho) \geq q_{\mathrm{c}}(\rho)$.
\end{lemma}

\begin{proof}
    Assume $\rho > \frac{1}{4}$. We define $P(\rho) := f_A(\rho) - \tilde{f}(A) \sqrt{1-\rho^2} = -0.8\zeta_{2}+0.45\zeta_{2}^{2}+0.39\zeta_{3}+0.65\zeta_{3}^{2}$. Then, the following holds:
    \begin{align*}
        f_A(\rho) - q_{\mathrm{c}}(\rho) & = \left(\tilde{f}(A) \sqrt{1-\rho^2} + P(\rho)\right) - \sqrt{(1-\rho^2)\left(C_{D_{\mathrm{a}}}(\rho) + \frac{1}{4}\right)} \\
        & = \left(\tilde{f}(A) - \sqrt{D_{\mathrm{a}} + \frac{5}{4}}\right)\sqrt{1-\rho^2} + P(\rho) && (\because C_{D_{\mathrm{a}}}(\rho) = D_{\mathrm{a}} + 1) \\
        & \geq 0.05 \sqrt{1-\rho^2} + P(\rho) && (\because \tilde{f}(A) \geq \tilde{f}(1.968) > 2.45).
    \end{align*}
    Here, $P(\rho)$ is a quadratic function that is increasing for all $\rho \geq 0.7$. Hence:
    \begin{itemize}
        \item For $\rho \in [0.7, 0.735]$, since $P(\rho) \geq P(0.7) > -0.0288$ and $\sqrt{1-\rho^2} \geq \sqrt{1-0.735^2} > 0.678$, we have $f_A(\rho) - q_{\mathrm{c}}(\rho) > 0.05 \cdot 0.678 - 0.0288 = 0.0051 > 0$.
        \item For $\rho \in [0.735, 1]$, since $P(\rho) \geq P(0.735) > 0$ and $\sqrt{1-\rho^2} \geq 0$, we have $f_A(\rho) - q_{\mathrm{c}}(\rho) > 0$.
    \end{itemize}
    Therefore, \autoref{lem:appendix-c-1} follows.
\end{proof}

\begin{lemma} \label{lem:appendix-c-2}
    Suppose $-\frac{1}{2} \leq \rho \leq 1$. If $2.15 \leq A \leq A_0$, then $q_A(\rho) = q_{\mathrm{c}}(\rho)$, that is, $f_A(\rho) \geq q_{\mathrm{c}}(\rho)$. This also proves property (P2).
\end{lemma}

\begin{proof}
    Since $f_A(\rho)$ is an increasing function with respect to $A$ by property (P1), it suffices to prove for the case $A = 2.15$. The graph of $f_{2.15}(\rho) - q_{\mathrm{c}}(\rho)$ is presented in \autoref{fig:appendix-c-2}, which illustrates that this value is positive for all $\rho$. However, to prove that $f_{2.15}(\rho) - q_{\mathrm{c}}(\rho) \geq 0$ in a rigorous way, we use a known technique called \emph{interval arithmetic}.

    \paragraph{Interval arithmetic.}
    Interval arithmetic calculates a possible range of an expression given possible ranges of each value. For example, given that $x \in [-0.5, 1.2]$, we have $x^2 \in [0^2, \max(0.5, 1.2)^2] = [0, 1.44]$, and hence $x^2 + x \in [0 - 0.5, 1.44 + 1.2] = [-0.5, 2.64]$. Note that the calculated range is not always tight: in fact, $x^2 + x$ is actually in the range $[-0.25, 2.64]$. Interval arithmetic supports addition, subtraction, multiplication, division, and square root, so we can compute a plausible range for functions built from these elementary operations.
    
    In our program, to avoid floating-point operations, we round both endpoints (i.e., round down the left and round up the right) to multiples of $2^{-72}$.

    \paragraph{Proving the inequality with interval arithmetic.}
    For a sufficiently small range $[\ell, r]$, we can prove that ``$f_{2.15}(\rho) - q_{\mathrm{c}}(\rho) \geq 0$ for all $\rho \in [\ell, r]$'' by calculating a possible range of $f_{2.15}(\rho) - q_{\mathrm{c}}(\rho)$ with interval arithmetic, if the lower endpoint is positive.
    
    To prove $f_{2.15}(\rho) - q_{\mathrm{c}}(\rho) \geq 0$ for all $\rho \in [-\frac{1}{2}, 1]$, we divided the interval $[-\frac{1}{2}, 1]$ into sufficiently small intervals $[\ell, r]$, so that the inequality is proved inside these smaller intervals. We note that the case $\rho \geq \frac{7}{10}$ is already proved in \autoref{lem:appendix-c-1}. We provide a complete certificate and a verifier in the supplementary material.
\end{proof}

\begin{figure}[t]
    \centering
    \includegraphics[width=\linewidth]{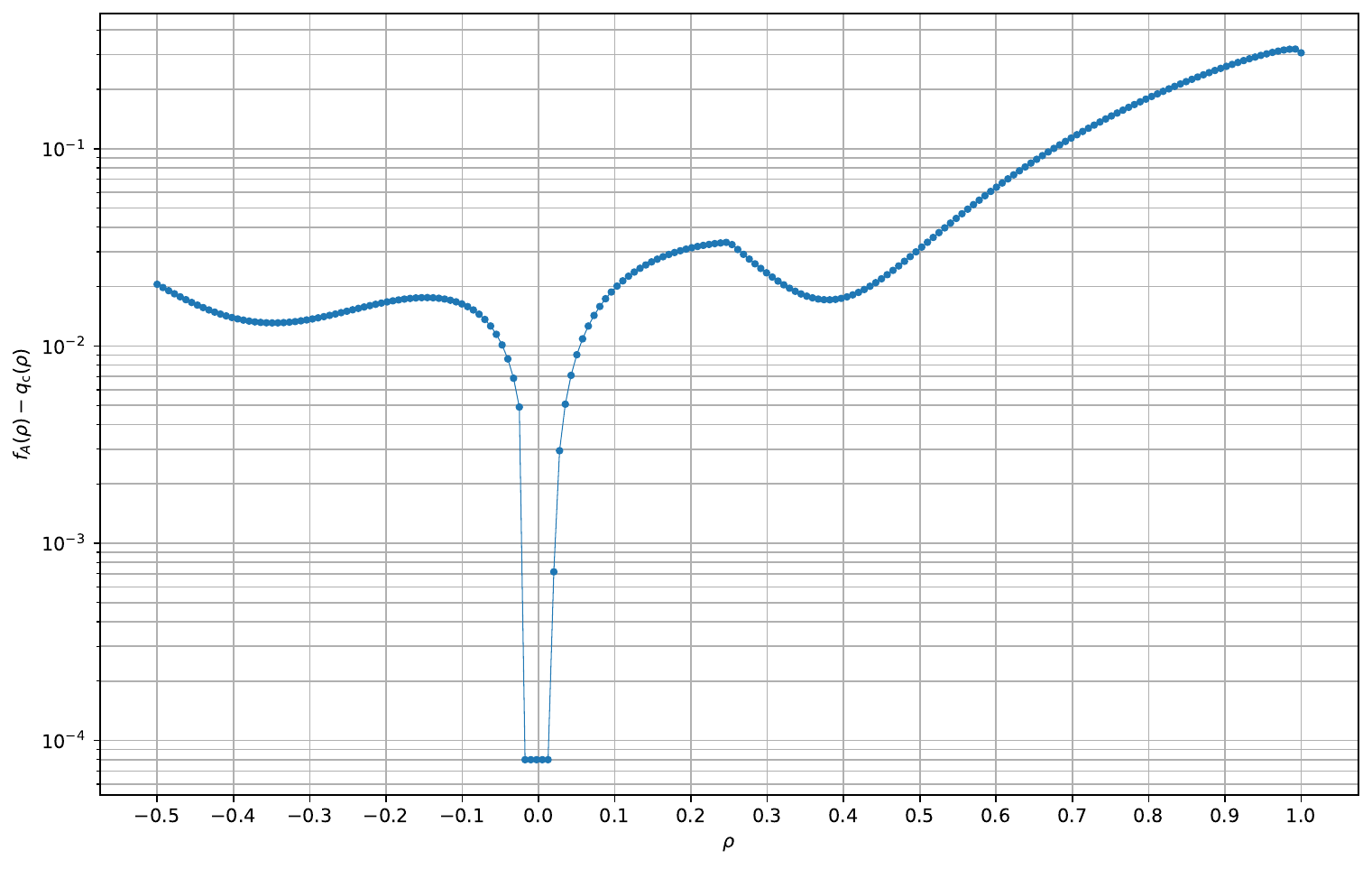}
    \caption{The graph of $f_{2.15}(\rho) - q_{\mathrm{c}}(\rho)$ for $-\frac{1}{2} \leq \rho \leq 1$. Note that the vertical axis is logarithmic.}
    \label{fig:appendix-c-2}
\end{figure}

Therefore, it remains to prove Case 1 for the range $(A, \rho) \in [1.968, 2.15] \times [-\frac{1}{2}, 0.7]$ and Case 2 for all $A \in [1.968, A_0]$. This is proved as follows: we provide a complete certificate and a verifier in the supplementary material.

\begin{proof}[Proof of \autoref{lem:source-bounds}]
    Here, we explain how to prove Case 1, since Case 2 can be proved similarly (and more easily). As in the proof of \autoref{lem:appendix-c-2}, we divide the rectangle $[1.968, 2.15] \times [-\frac{1}{2}, 0.7]$ into sufficiently small subrectangles $[A_\ell, A_r] \times [\rho_\ell, \rho_r]$, and for each of them, we prove \autoref{lem:source-bounds} for all $(A, \rho) \in [A_\ell, A_r] \times [\rho_\ell, \rho_r]$. Specifically, for each $[A_\ell, A_r] \times [\rho_\ell, \rho_r]$ we show either of the following:
    
    \begin{itemize}
        \item Prove that $q_A(\rho) = q_{\mathrm{c}}(\rho)$ for all $(A, \rho) \in [A_\ell, A_r] \times [\rho_\ell, \rho_r]$, which means that the assumption of Case 1 does not hold. This can be done very similarly to the proof of \autoref{lem:appendix-c-2}.
        \item Give one tuple $(a, b, \mu, c, \kappa, \theta)$ such that the conditions (i), (ii), (iii), and (iv) hold for all $(A, \rho) \in [A_\ell, A_r] \times [\rho_\ell, \rho_r]$.
    \end{itemize}

    Next, we explain how to prove conditions (i), (ii), (iii), and (iv) for a given $(a, b, \mu, c, \kappa, \theta)$. The verification of conditions (i) and (ii) is trivial. For condition (iii), we need to show that
    \begin{equation} \label{eq:appendix-cond3}
        \left(\dfrac{q_A(\rho)}{\sqrt{1-\rho^{2}}}\right)^{2} + \left(\mu+\theta(\mu+\kappa)\right)^{2} - \left(C_{D_{\mathrm{a}}}(\rho)+\dfrac{1}{4}\right) - 10^{-18} \geq 0
    \end{equation}
    for all $(A, \rho) \in [A_\ell, A_r] \times [\rho_\ell, \rho_r]$. The verification of condition (iv) is trickier: we show that
    \begin{equation} \label{eq:appendix-cond4}
        \frac{L(1/2+\rho \rho')-(1+\theta)b(\rho'+\rho/2)+\theta c |\rho'+\rho/2|}{\sqrt{1-\rho^{2}}} +\frac{3}{2}(1+\theta)a-\frac{\sqrt{3}}{2}\theta A - q_A(\rho') - 3\cdot10^{-6} \geq 0,
    \end{equation}
    where $L = q_A(\rho) / \sqrt{1-\rho^{2}}$, for all $(A, \rho, \rho') \in [A_\ell, A_r] \times [\rho_\ell, \rho_r] \times [-\frac{1}{2}, \frac{1}{2} - \rho_\ell]$. Here, we need to prove this over the three-dimensional box (rather than the two-dimensional rectangle).

    For condition (iv), we further divide $\rho' \in [-\frac{1}{2}, \frac{1}{2} - \rho_\ell]$ into multiple intervals $[\rho'_\ell, \rho'_r]$, and for each of them, we prove the inequality for all $(A, \rho, \rho') \in [A_\ell, A_r] \times [\rho_\ell, \rho_r] \times [\rho'_\ell, \rho'_r]$. In order to guarantee differentiability, for each $[\rho'_\ell, \rho'_r]$, we prove an inequality that replaces $|\rho'+\rho/2|$ with either $+(\rho'+\rho/2)$ or $-(\rho'+\rho/2)$ (which sign to choose can differ for each $[\rho'_\ell, \rho'_r]$); this is indeed a stronger inequality, due to $\theta > 0, c \geq 0$, and $\sqrt{1-\rho^2} \geq 0$. Similarly, we also replace $q_A(\rho')$ with either $f_A(\rho')$ or $q_{\mathrm{c}}(\rho')$; this only makes the inequality stronger.

    \paragraph{Differentiability of the functions.}
    Generally speaking, let $F: \mathbb{R}^d \to \mathbb{R}$ be a multivariable function composed of addition, subtraction, multiplication, division, and square root. The function is $C^\infty$ (i.e., infinitely differentiable) on any open region where every denominator of division is non-zero and every argument of square root is positive.
    
    Therefore, the left-hand sides of the two inequalities above (for conditions (iii) and (iv); \autoref{eq:appendix-cond3} and \autoref{eq:appendix-cond4}) are $C^\infty$ for $\rho, \rho' \neq -\frac{1}{40}, r_-, 0, r_+, \frac{1}{50}, \frac{1}{4}, 1$, where
    \begin{equation*}
        r_- := \frac{4 - D_{\mathrm{a}}}{2D_{\mathrm{a}} + 16} \approx -0.020238, \quad r_+ := \frac{D_{\mathrm{a}} - 4}{4D_{\mathrm{a}} + 8} \approx 0.019451.
    \end{equation*}
    These are the points where the capacity $C_D(\rho)$ in \autoref{def:capacity} matches the threshold $D + 2$. Moreover, for condition (iii), if we divide the rectangle $[1.968, 2.15] \times [-\frac{1}{2}, 0.7]$ into seven subrectangles splitting at $\rho = -\frac{1}{40}, r_-, 0, r_+, \frac{1}{50}, \frac{1}{4}$, the left-hand side of \autoref{eq:appendix-cond3} can be considered $C^\infty$ for each box, including the boundary, because this eliminates case divisions (including $\min$ and $\max$). Similarly, for condition (iv), if we divide the box so that its interior does not contain $\rho' = -\frac{1}{40}, r_-, 0, r_+, \frac{1}{50}, \frac{1}{4}$, the left-hand side of \autoref{eq:appendix-cond4} can be considered $C^\infty$ for each box (except for $\rho' = 1$).

    \paragraph{Proving the positivity of a function.}
    In general, we need to prove that a $C^\infty$ function $F: D \to \mathbb{R}$ is positive inside the domain $D$, where $D$ is a rectangle for condition (iii) and a box for condition (iv). For each region, we choose either of the following ways to prove the positivity.
    \begin{itemize}
        \item Use interval arithmetic directly to compute the lower (and upper) bounds of $F$ in the region, as explained in the proof of \autoref{lem:appendix-c-2}.
        \item Apply a Taylor expansion to the function $F$ and obtain a lower bound explained below.
    \end{itemize}
    If the resulting lower bound is positive, it proves that $F$ is positive.

    \paragraph{Taylor expansion.}
    Here, we consider the two-dimensional case where $D = [x_\ell, x_r] \times [y_\ell, y_r]$; the three-dimensional case can be handled similarly. We define
    \begin{equation*}
        x_m := \frac{x_\ell + x_r}{2}, \quad y_m := \frac{y_\ell + y_r}{2}
    \end{equation*}
    for the center. Then, for $(x, y) \in D$, by Taylor's theorem, there exists some $(x', y') \in D$ such that:
    \begin{multline*}
        f(x, y) = f(x_m, y_m) + \frac{\partial f}{\partial x} (x_m, y_m) \cdot (x-x_m) + \frac{\partial f}{\partial y} (x_m, y_m) \cdot (y - y_m) \\
        + \frac{1}{2} \frac{\partial^2 f}{\partial x^2} (x', y') \cdot (x-x_m)^2 + \frac{1}{2} \frac{\partial^2 f}{\partial y^2} (x', y') \cdot (y-y_m)^2 + \frac{\partial^2 f}{\partial x \partial y} (x', y') \cdot (x-x_m) (y-y_m).
    \end{multline*}
    Therefore, a lower bound on the minimum value of $f(x, y)$ can be obtained using
    \begin{equation*}
        \min_{(x', y') \in D} \frac{\partial^2 f}{\partial x^2} (x', y'), \quad \min_{(x', y') \in D} \frac{\partial^2 f}{\partial y^2} (x', y'), \quad \max_{(x', y') \in D} \left|\frac{\partial^2 f}{\partial x \partial y} (x', y')\right|,
    \end{equation*}
    where the bounds on these three values are calculated by interval arithmetic.

    \paragraph{Concluding the proof.}
    With the methods explained above, for all $(A, \rho)$ that meet the assumption of Case 1, we proved the existence of a tuple $(a, b, \mu, c, \kappa, \theta)$ that satisfies all of conditions (i), (ii), (iii), and (iv). We also proved this similarly for Case 2. A complete certificate and a verifier is provided at \url{https://github.com/ka714/coloring_01537}.
\end{proof}

\begin{remark}
    Although we explained the proof of \autoref{lem:source-bounds} in the discussion above, we did not explain how to obtain a certificate, that is, how to divide a rectangle $[1.968, 2.15] \times [-\frac{1}{2}, 0.7]$, so we explain it here. We start with the entire $[1.968, 2.15] \times [\rho_\ell, \rho_r]$ (which is one of the seven rectangles mentioned in the paragraph ``Differentiability of the Functions'') and attempt to find a common $(a, b, \mu, c, \kappa, \theta)$ that meets conditions (i), (ii), (iii), and (iv) in the entire domain. If no such tuple $(a, b, \mu, c, \kappa, \theta)$ can be found, we divide the rectangle into two equal parts, either horizontally or vertically, and recursively apply the same procedure. If we recurse deep enough and the rectangle becomes sufficiently small, we can find a tuple $(a, b, \mu, c, \kappa, \theta)$ that satisfies the conditions. In this way, we obtain a finite-size certificate that covers the entire rectangle.
\end{remark}

\end{document}